\documentclass[10pt,twocolumn,twoside]{IEEEtran}

\IEEEoverridecommandlockouts

\usepackage{cite}
\usepackage{algorithmic}
\usepackage{textcomp}
\usepackage{xcolor}
\usepackage{soul}
\usepackage{graphicx,color}
\usepackage{amsfonts,amsthm,amssymb}
\usepackage{mathrsfs,amsmath,arydshln,latexsym}
\usepackage{cite,url}
\usepackage[bookmarks,colorlinks]{hyperref}
\usepackage[linesnumbered,ruled,vlined]{algorithm2e}
\usepackage{multirow,array,makecell}
\usepackage{stfloats}
\usepackage{fancyhdr}
\usepackage{bm}
\usepackage{enumitem}
\usepackage{pifont}
\usepackage{subcaption}
\usepackage{nomencl}
\usepackage{colortbl}
\usepackage{diagbox}
\usepackage{balance}
\usepackage{braket}
\usepackage{physics}
\usepackage{framed}
\usepackage{datatool} 
\usepackage{booktabs} 
\usepackage{xparse}   
\usepackage{arydshln}  
\usepackage{longtable}
\usepackage{etoolbox}
\usepackage[linewidth=0.6pt]{mdframed}
\usepackage{threeparttablex}
\allowdisplaybreaks[4]

\newcommand{\tabincell}[2]{\begin{tabular}{@{}#1@{}}#2\end{tabular}}

\newtheorem{theorem}{Theorem}

\newtheorem{corollary}{Corollary}
\newtheorem{lemma}{Lemma}

\begin{document}

	\title{
		Rydberg Atomic Quantum MIMO Receivers\\ for The Multi-User Uplink
	}

	\author{Tierui Gong,~\IEEEmembership{Member,~IEEE}, 
		Chau Yuen,~\IEEEmembership{Fellow,~IEEE}, 
		Chong Meng Samson See,~\IEEEmembership{Member,~IEEE},\\
		Mérouane Debbah,~\IEEEmembership{Fellow,~IEEE},
		Lajos Hanzo,~\IEEEmembership{Life Fellow,~IEEE}
		\vspace{-0.6cm}
		\thanks{
			This work was supported by the Ministry of Education (MoE), Singapore, under its Academic Research Fund (AcRF) Tier 1 Thematic Grant RT12/23 023780-00001. 
			The work of Lajos Hanzo was supported by the following Engineering and Physical Sciences Research Council (EPSRC) projects: Platform for Driving Ultimate Connectivity (TITAN) (EP/X04047X/1; EP/Y037243/1); Robust and Reliable Quantum Computing (RoaRQ, EP/W032635/1); PerCom  EP/X01228X/1; EP/Y026721/1, India-UK Intelligent Spectrum Innovation ICON UKRI-1859.
			Part of the work was published in ICC 2025 \cite{Gong2025RAQ_MIMO_conf}.
			\textit{(Corresponding author: Chau Yuen.)}
		}
		\thanks{
			Tierui Gong and Chau Yuen are with School of Electrical and Electronics Engineering, Nanyang Technological University, Singapore 639798 (e-mail: tierui.gong@ntu.edu.sg, chau.yuen@ntu.edu.sg).
		} 
		\thanks{
			Chong Meng Samson See is with DSO National Laboratories, Singapore 118225 (e-mail: schongme@dso.org.sg).
		} 
		\thanks{
			Mérouane Debbah is with the KU 6G Research Center, Department of Computer and Information Engineering, Khalifa University, Abu Dhabi 127788, UAE, and also with the CentraleSupelec, University Paris-Saclay, 91192 Gifsur-Yvette, France (e-mail: merouane.debbah@ku.ac.ae).
		}
		\thanks{
			Lajos Hanzo is with School of Electronics and Computer Science, University of Southampton, SO17 1BJ Southampton, U.K. (e-mail: lh@ecs.soton.ac.uk).
		}
	}

	\markboth{Accepted by IEEE Transactions on Wireless Communications, 2026}%
	{Gong \MakeLowercase{\textit{et al.}}: Rydberg Atomic Quantum MIMO Receivers for The Multi-User Uplink}

	\maketitle

	\begin{abstract}
		Rydberg atomic quantum receivers (RAQRs) have emerged as a promising solution for evolving wireless receivers from the classical to the quantum domain. To further unleash their great potential in wireless communications, we propose a flexible architecture for Rydberg atomic quantum multiple-input multiple-output (RAQ-MIMO) receivers in the multi-user uplink. Then the corresponding signal model of the RAQ-MIMO system is constructed by paving the way from quantum physics to classical wireless communications. Explicitly, we outline the associated operating principles and transmission flow. We also validate the linearity of our model and its feasible region. Based on our model, we derive closed-form asymptotic formulas for the ergodic achievable rate (EAR) of both the maximum-ratio combining (MRC) and zero-forcing (ZF) receivers operating in uncorrelated fading channels (UFC) and the correlated fading channels (CFC), as well as in the standard quantum limit (SQL) and photon shot limit (PSL) regimes, respectively. Furthermore, we unveil that the EAR scales logarithmically without bound with the product of effective number $N_{\text{atom}}$ and coherence time $T_2$ of the atomic ensemble in the SQL regime, but exhibits non-monotonic trade-off between the collective atomic enhancement and optical-depth-dependent attenuation in the PSL regime. More particularly, the transmit power of users can be scaled down quadratically with $N_{\text{atom}} \tau$, $\tau \in \{ T_2, \frac{ {\cal C} (\Omega_{\ell}) }{A_p} \}$, but the EAR per user retains fixed, by increasing $N_{\text{atom}}$ while retaining the sensor number $M \propto N_{\text{atom}} \tau$ in the SQL regime or $M \propto \exp \big( \frac{N_{\text{atom}} {\bar \chi}}{A_p} \big)$ in the PSL regime. We also quantify the superiority of RAQ-MIMO receivers over the classical massive MIMO receivers, specifying an increase of $\log_{2} \Pi$ of the EAR per user, $\Pi$-fold reduction of the users' transmit power, and $\sqrt[\nu]{\Pi}$-fold increase of the transmission distance, respectively, where $\Pi$ $\triangleq$ $\textit{ReceiverGainRatio} / \textit{NoisePowerRatio}$ of the single-sensor receiver and $\nu$ is the path-loss exponent. Lastly, numerical simulations validate our theoretical results and reveal that the Doppler-broadening free RAQ-MIMO scheme can either realize $12$ bits/s/Hz/user ($8$ bits/s/Hz/user) higher EAR, or $10000$-fold ($500$-fold) lower transmit power, or alternatively $100$-fold ($21$-fold) longer distance in free-space transmissions, in the SQL (PSL) regime, compared to classical massive MIMO receivers. 
	\end{abstract}

	\begin{IEEEkeywords}
	Rydberg atomic quantum MIMO (RAQ-MIMO) receiver, uncorrelated and correlated fading channels, ergodic achievable rate and power scaling law, maximum-ratio combining (MRC) and zero-forcing (ZF), standard quantum limit (SQL) and photon shot limit (PSL)
	\vspace{-0.2cm}
	\end{IEEEkeywords}

	\section{Introduction}
	
	The year 2025 has been proclaimed as the International Year of Quantum Science and Technology (IYQ) by the United Nations. The IYQ declaration recognizes the significant progress of quantum information science and technology (QIST) during the past century since the initial development of quantum mechanics. It is also expected to be a spring-board for QIST research. Among the numerous hot topics, quantum sensing, quantum communications, and quantum computing constitute three main pillars \cite{hanzo2025quantum}. More particularly, quantum sensing relies on quantum phenomena to realize the measurements of a physical quantity at an unprecedented accuracy \cite{degen2017quantum}. By exploiting specific quantum sensing principles, a multitude of quantum sensors have been invented for various applications, such as the quantum electrometer, quantum magnetometer, quantum accelerometer, quantum gyroscope, quantum gravimeter, and quantum clock. On the same note, the novel concept of Rydberg atomic quantum receivers (RAQRs) exhibits great potential in terms of detecting radio-frequency (RF) signals in wireless communication and sensing \cite{gong2024RAQRs, gong2026rydberg,gong2025rydberg}.

	By harnessing the compelling properties of Rydberg atoms, diverse choices of electron transitions, and beneficial atom-light interaction phenomena, RAQRs exhibit ultra-high sensitivity, broadband tunability, narrowband selectivity, International System of Units (SI)-traceability, enhanced instantaneous bandwidth, and ultra-wide input power range \cite{gong2024RAQRs}. They also allow a direct passband-to-baseband conversion without employing sophisticated integrated circuits. Additionally, the receiver sensor of RAQRs is in the form of an optical region consisting of Rydberg atoms, allowing them to achieve an ultra-high scalability. For example, it is convenient to implement a large-scale optical array by relying on parallel laser beams. Furthermore, as a benefit of their totally different operating principle from those of classical antennas, RAQRs are less susceptible to mutual coupling, have a sensor size independent of the wavelength, and experience a wider angular reception range. The above distinctive features of RAQRs have the promise of revolutionizing classical RF receivers, while supporting the data rate, connectivity, and latency specifications of next-generation wireless systems \cite{recommendation2023framework}.

	To date, the investigations of RAQRs have mainly emerged from the physics society through experimental verifications with respect to their enhanced sensitivity, various functionalities, and novel architectures. Specifically, the sensitivity has been experimentally demonstrated to be on the order of $\sim$ \textmu V/cm/$\sqrt{\text{Hz}}$ using a standard structure \cite{sedlacek2012microwave, fan2015atom}. It has then also been further improved to be on the order of $\sim$ nV/cm/$\sqrt{\text{Hz}}$ using a superheterodyne structure \cite{jing2020atomic,Tu2024Approaching}. By further extending the receiver from a single sensor to multiple sensors upon harnessing an optimal laser array, the authors of \cite{wu2024enhancing} demonstrated that it is capable of achieving a sensitivity of $< 1$ nV/cm/$\sqrt{\text{Hz}}$. The sensitivity limit of RAQRs is governed by the standard quantum limit (SQL). Additionally, by exploiting the above-mentioned laser array scheme, accurate angular direction detection is realizable \cite{robinson2021determining,mao2023digital,richardson2024study}. Furthermore, RAQRs are also capable of detecting the amplitude, phase, polarization, modulation, and the continuous-band or multiband nature of signals \cite{schlossberger2024rydberg,zhang2024rydberg,yuan2023quantum,Fancher2021Rydberg}. Lastly, recently several novel schemes emerged by exploiting specific physical phenomena, such as the many-body interaction \cite{ding2022enhanced}, stochastic resonance \cite{wu2024nonlinearity}, as well as Schrodinger cat and squeezed states \cite{yuan2023quantum} for approaching or even surpassing the SQL.

	Based on the above advances predominantly attained by the physics society, communication-oriented applications started to arise. Specifically, \cite{gong2024RAQRs} presented a comprehensive overview of RAQRs conceived for classical wireless communications and sensing. Explicitly, both Rydberg atomic quantum single-input single-output (RAQ-SISO) and Rydberg atomic quantum multiple-input multiple-output (RAQ-MIMO) schemes were conceived, and their promising potential was unveiled. Additionally, reference \cite{gong2026rydberg} conceived a reception scheme and constructed its corresponding signal model for wireless communication and sensing. By obeying the realistic operating principle and transmission flow of RAQRs, this study closed the knowledge gap between the physics and communication communities, unveiling their superiority over classical RF receivers, and paved the way for the design of upper-layer applications. Furthermore, reference \cite{gong2025rydberg} proposed a Rydberg atomic quantum uniform linear array for multi-target direction-of-arrival estimation, where the sensor gain mismatch problem was solved and a significant increase in estimation accuracy was attained by proposing a novel RAQ-ESPRIT algorithm. There are also other emerging applications, such as those discussed in  \cite{otto2021data,zhang2023quantum,cui2025towards,yuan2024electromagnetic,chen2025harnessing}, aiming for harnessing RAQRs for both spatial displacement detection and wireless communications.

	Focusing briefly on the family of massive MIMO (M-MIMO) schemes \cite{Marzetta2010Noncooperative,Larsson2014Massive,Yang2015Fifty,lim2015performance} which constitute a key enabler for fifth-generation (5G) wireless networks, they have evolved further to reconfigurable intelligent surfaces \cite{huang2019reconfigurable,di2020smart,xu2023reconfiguring}, holographic MIMO arrangements \cite{Gong2023Holographic,Gong2024HMIMO,Gong2024Near} and extremely large-scale MIMO (XL-MIMO) systems \cite{Lu2024Tutorial,wang2023extremely,Yajun2024Nearfield}. The latter ones constitute promising candidates for sixth-generation (6G) networks. All these technologies rely on the deployment of metallic antenna arrays to receive the desired signals, which determines the fundamental performance limits of these classical MIMO technologies. Specifically, the classical antenna arrays face limited sensitivity and inherent physical limitations, governed by the Chu–Harrington and Hannan limits \cite{wei2024electromagnetic}, where the antenna size, bandwidth, radiation efficiency, gain, and reception angle are intricately coupled to each other. But again, this restricts the performance of these conventional antenna arrays and exacerbates their design challenges, especially for high-frequency, large-scale antenna arrays.

	To circumvent the bottleneck faced by classical MIMO receivers and further unlock the potential of RAQRs in wireless communications, we conceive and investigate RAQ-MIMO receivers. Specifically, our contributions are as follows 
	\begin{itemize}
		\item 
		We propose a flexible architecture for the RAQ-MIMO receiver capable of simultaneously receiving multi-user signals in the uplink of communication systems. This flexible scheme has the capability to handle a range of different RF wavelengths spanning from low to high frequencies, as well as supporting practical implementations having either compact or large apertures. Furthermore, we construct the corresponding signal model for such RAQ-MIMO systems by building a bridge between quantum physics and wireless communications by relying on realistic operating principles and practical transmission flow. More particularly, the linearity and feasible region of our signal model are also investigated.

		\item 
		We derive closed-form asymptotic results for the ergodic achievable rate (EAR) of RAQ-MIMO systems relying on maximum-ratio combining (MRC) and zero-forcing (ZF) receivers operating in the SQL and photon shot limit (PSL) regimes, as well as in the face of uncorrelated fading channels (UFC) and correlated fading channels (CFC), respectively, leveraging random matrix theory. Our results reveal that the EAR scaling behaviour of the RAQ-MIMO receiver is cooperatively determined by both the macroscopic array dimension and the microscopic quantum/optical-related parameters.

		\item 
		We study the scaling behavior of homogeneous RAQ-MIMO systems operating in SQL/PSL regimes, respectively. Specifically, the EAR scales logarithmically without bound with the product of effective number $N_{\text{atom}}$ and coherence time $T_2$ of the atomic ensemble in the SQL regime, namely $\propto$ $\log_2(N_{\text{atom}} T_2)$. By contrast, the EAR exhibits non-monotonic trade-off between the collective atomic enhancement $\big[ \frac{ N_{\text{atom}} {\cal C} (\Omega_{\ell}) }{A_p} \big]^{2}$ and optical-depth-dependent attenuation $\exp \big(- \frac{N_{\text{atom}} {\bar \chi}}{A_p} \big)$ in the PSL regime, which is maximally achieved at an optimal atom number of $N_{\text{atom}}^{\star} = \frac{2 A_p}{{\bar \chi}}$. Additionally, the transmit power of users scales quadratically with $\frac{1}{N_{\text{atom}} \tau}$, $\tau \in \{ T_2, \frac{ {\cal C} (\Omega_{\ell}) }{A_p} \}$, but the EAR per user retains fixed, by increasing $N_{\text{atom}}$ while retaining the sensor number $M \propto N_{\text{atom}} \tau$ in the SQL regime or $M \propto \exp \big( \frac{N_{\text{atom}} {\bar \chi}}{A_p} \big)$ in the PSL regime.

		\item 
		We extend the asymptotic scaling analysis to inhomogeneous RAQ-MIMO systems, characterizing the impact of sensor non-uniformity. Specifically, in the SQL regime, the EAR is bounded by a logarithmic function of the average eigenmode-weighted collective quantum coherence factor $\overline{\gamma}_{\rm{coh}}^{\rm{ave}}$ across all $M$ sensors, namely $\propto \log_2(\overline{\gamma}_{\rm{coh}}^{\rm{ave}})$. By contrast, in the PSL regime, the EAR exhibits a highly coupled trade-off driven by the aggregate inhomogeneous atomic enhancement and optical-dependent attenuation. Additionally, the power scaling law for inhomogeneous RAQ-MIMO systems yields $1 / (N_{\text{atom}}^{\text{ref}} \overline{\tau})^2$, $\overline{\tau} \in \{ \overline{\tau}_{\rm{sql}}, \overline{\tau}_{\rm{psl}} \}$ while maintaining a stable EAR with the requirement of $M \propto \{ N_{\rm atom}^{\rm ref} \overline{\tau}_{\rm{sql}},\; \exp ( \overline X_{\rm{abs}} ) \}$ in the SQL and PSL regimes, respectively.

		\item 
		We compare the RAQ-MIMO to classical M-MIMO receivers in terms of their EAR, transmit power, and transmission distance. Specifically, compared to the classical counterparts, the MRC RAQ-MIMO receiver approaches an identical EAR as ${\cal P}_{s} \to \infty$, while the ZF RAQ-MIMO receiver has an extra EAR of $\log_{2} \Pi$ per user, where $\Pi$ $=$ $\textit{ReceiverGainRatio} / \textit{NoisePowerRatio}$ represents the SNR ratio of single-sensor receivers. Furthermore, the RAQ-MIMO receiver allows the transmit power of users to be $\Pi$-fold lower than those of M-MIMO systems at identical EARs. Viewed from yet another different perspective, the RAQ-MIMO receiver expands the transmission distance by a factor of $\sqrt[\nu]{ \Pi }$ compared to that of classical M-MIMO receivers, when having the same transmit power. 
		
	\end{itemize}

	\textit{Organization and Notations}: 
	The rest of the article is organized as follows: In Section \ref{sec:AtomicReceiverModel}, we propose a flexible RAQ-MIMO receiver and construct its signal model. In Section \ref{sec:EAR}, we provide asymptotic analysis of the EAR of the multi-user uplink RAQ-MIMO. Then we provide our simulation results in Section \ref{sec:Simulations}, next present a list of discussions in Section \ref{sec:Discussions}, and finally conclude in Section \ref{sec:Conclusions}. 
	We use the following notations: $\jmath^2 = 1$; $\frac{\mathrm{d} \bm{\rho}}{\mathrm{d} t}$ is the differential of $\bm{\rho}$ with respect to time; $\comm{\bm{H}}{\bm{\rho}} = \bm{H}\bm{\rho} - \bm{\rho}\bm{H}$ represents the commutator; $\left\{ \bm{\varGamma}, \bm{\rho} \right\} = \bm{\varGamma}\bm{\rho} + \bm{\rho}\bm{\varGamma}$ stands for the anticommutator; $\mathscr{R} \{ \cdot \}$ and $\mathscr{I} \{ \cdot \}$ take the real and imaginary parts of a complex number; ${\rm{diag}}\{ \cdot \}$ is the diagonal matrix; $\chi' (\Omega_{\ell,m})$ represents the derivative of $\chi (\Omega_{{\rm{RF}},m})$ when $\Omega_{{\rm{RF}},m} = \Omega_{\ell,m}$; $\hbar$ is the reduced Planck constant; $c$ and $\epsilon_0$ are the speed of light in free space and the vacuum permittivity; $\eta_0$ and $\eta$ are the antenna efficiency and the quantum efficiency; $q$ and $a_0$ are the elementary charge and Bohr radius; $\mu_{i-1,i}$ is the dipole moment of the $\ket{i-1}$\textrightarrow $\ket{i}$ transition, $i=2, 3, 4$; $\odot$ represents the Hadamard product; If $\bm{S}$ has a Wishart distribution of $M$ degrees-of-freedom and the scaling matrix obeys $\bm{V} \in \mathbb{C}^{K \times K}$, then we have $\bm{S} \sim \mathcal{W}_{K}(M, \bm{V})$.

	\vspace{-0.8em}
	\section{The Proposed RAQ-MIMO Scheme and Constructed Signal Model}
	\label{sec:AtomicReceiverModel}
	
	In this section, we propose a RAQ-MIMO aided multi-user uplink system by first detailing its system composition and configuration, and by then investigating its signal model from an equivalent baseband signal perspective.

	\vspace{-0.8em}
	\subsection{System Description of RAQ-MIMO}
	\label{SubSection:Description}

	\subsubsection{Array Structure}
	In our RAQ-MIMO system, we exploit Cesium (Cs) atoms for receiving $K$ simultaneous RF signals having the same carrier frequency of $f_c$ from $K$ users. We then construct a flexible RAQ-MIMO receiver architecture, as portrayed in Fig. \ref{fig:RAQ-MIMO}(a). Specifically, we employ $M_1 \ge 1$ vapor cells, where each vapor cell filled up with Cs atoms is spatially penetrated by $M_2 \ge 1$ pairs of counter-propagating laser (probe and coupling) beams to form a total of $M = M_1 M_2$ receiver sensors. More particularly, we assume that these $M$ sensors form a uniform linear array (ULA) having a spacing of $d \le \frac{\lambda}{2}$\footnote{We note that this condition is an algorithmic requirement for arrays, so that different directional cosines can be physically distinguished \cite[Ch. 7]{Tse2005Fundamentals}.}, where $\lambda$ is the wavelength of the $K$ RF signals.  
	Based on this spacing, we note that our RAQ-MIMO structure is flexible enough for realizing arrays operating both at low and high frequencies. Specifically, we have the following two special cases, as discussed in \cite{gong2024RAQRs}: (i) For low-frequency signals having longer wavelengths, our RAQ-MIMO receiver reduces to a vapor cell array, where each vapor cell encompasses a receiver sensor, namely we have $M_2 = 1$ and $M = M_1$. (ii) For high-frequency signals having shorter wavelengths, our RAQ-MIMO receiver becomes a beam array within a single vapor cell, where we have $M_1 = 1$ and $M = M_2$.

	\begin{figure}[t!]
		\centering
		\includegraphics[width=0.476\textwidth]{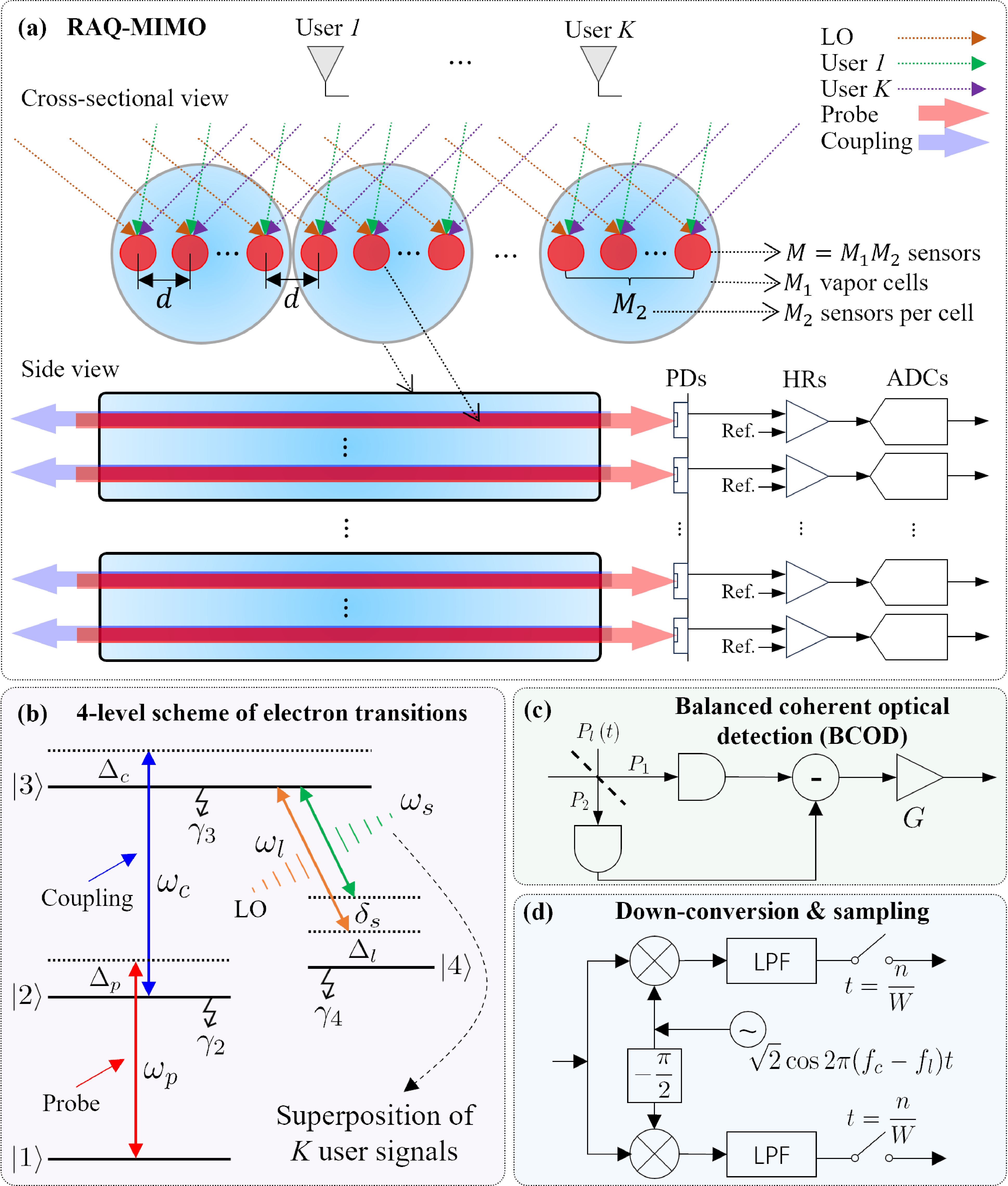}
		\caption{(a) The superheterodyne structure of RAQ-MIMO, (b) the four-level scheme of electron transitions, (c) the balanced coherent optical detection (BCOD) scheme, and (d) the down-conversion and sampling by homodyne receivers (HRs) and the analog-to-digital converters (ADCs).}
		\vspace{-1.6em}
		\label{fig:RAQ-MIMO}
	\end{figure}

	\subsubsection{Dedicated Chain}
	In each receiver sensor, Cs atoms are excited to a Rydberg state to absorb all $K$-user signals, thereby triggering another corresponding Rydberg state. These RF absorptions affect the laser beams obeying a specific \emph{RF-atomic-optical transformation} relationship, as detailed in Section \ref{SubSection:TransferFunction}. In order to capture both the amplitudes and phases of all $K$-user signals, we apply the superheterodyne philosophy for our receiver. 
	Specifically, a superheterodyne RAQ-MIMO receiver exploits the quantum response of Rydberg atoms to mix the locally applied RF field of the local oscillator (LO), with the incident multi-user signals, intrinsically generating an intermediate (beat) frequency within the atomic medium via the electromagnetically induced transparency (EIT) and Rydberg-level coherence. This beat frequency is mapped onto the probe transmission, then detected by a photodetector (PD), down-converted and sampled by a HR and ADC, respectively, allowing a precise amplitude and phase recovery \cite{jing2020atomic,gong2026rydberg}.

	\subsubsection{Electron Transition}
	\label{subsub:Transition}
	Each Rydberg atom undergoes a four-level ladder-type scheme, as seen in Fig. \ref{fig:RAQ-MIMO}(b). Briefly, the so-called ground state $\ket{1}$, immediate excited state $\ket{2}$, and Rydberg states $\ket{3}$, $\ket{4}$ are involved in this scheme, where $\omega_{i-1,i}$ is the resonant angular frequency between $\ket{i-1}$ and $\ket{i}$, and each energy level has its own decay rate $\gamma_{i}$, $i = 2, 3, 4$. These energy levels are coupled by the probe beam, coupling beam, and the superposition of the LO and $K$-user signals, resonant or near-resonant, respectively. When the laser beams and LO are resonant, their so-called detuning frequencies $\Delta_{p,c,\ell}$ are zero. By contrast, when the laser beams and LO are near-resonant, their detuning frequencies $\Delta_{p,c,\ell}$ do exist. Explicitly, the detuning frequency represents a small frequency difference between the coupling frequency required for the electron transition of two energy levels and the actual frequency of the external electromagnetic wave to enable this transition. Specifically for the $m$-th sensor, we have $\Delta_{p,m} = \omega_{12} - \omega_{p,m}$, $\Delta_{c,m} = \omega_{23} - \omega_{c,m}$, and $\Delta_{\ell,m} = \omega_{34} - \omega_{\ell,m}$, where $\omega_{p,m}$, $\omega_{c,m}$ and $\omega_{\ell,m}$ represent the angular frequencies of the probe beam, coupling beam, and LO, respectively.

	\vspace{-0.8em}
	\subsection{Multi-User Signals to be Detected}
	
	To proceed, we assume that the user signals are plane waves. We then express the passband signal of the $k$-th user at the $m$-th receiver sensor as 
	\begin{align}
		\nonumber
		X_{m, k} (t) 
		&= \sqrt{2\mathcal{P}_{x, m, k}} \cos ( 2 \pi f_c t + \theta_{x, m, k} ) \\
		\label{eq:RFsignal}
		&= \sqrt{2} \mathscr{R} \left\{ x_{m, k} (t) \exp( \jmath 2 \pi f_c t ) \right\}, 
	\end{align}
	where $\mathcal{P}_{x, m, k}$ and $\theta_{x, m, k}$ are the received signal power and phase of the $k$-th user at the $m$-th sensor, respectively. Specifically, we have $\mathcal{P}_{x, m, k} = \frac{c \epsilon_0 A_e}{2} \left| U_{x, m, k} \right|^{2}$, where $U_{x, m, k}$ is the amplitude of $X_{m, k} (t)$ and $A_e$ is the effective aperture of the $m$-th sensor.  Furthermore, $x_{m, k} (t) = \sqrt{\mathcal{P}_{x, m, k}} \exp( \jmath \theta_{x, m, k} )$ is the equivalent baseband signal. Note that $U_{x, m, k}$, $\theta_{x, m, k}$ and their related variables are time-invariant during a symbol period, we thus neglect their time index. Likewise, we assume that the LO signal is plane-wave and formulate the corresponding channel output at the $m$-th sensor of the RAQ-MIMO by 
	\begin{align}
		\nonumber
		Y_{m} (t) 
		&= \sqrt{2\mathcal{P}_{\ell, m}} \cos ( 2 \pi f_{\ell} t + \theta_{\ell, m} ) \\
		\label{eq:LOsignal}
		&= \sqrt{2} \mathscr{R} \left\{ y_{m} (t) \exp( \jmath 2 \pi f_{\ell} t ) \right\}, 
	\end{align}
	where $\mathcal{P}_{\ell, m} = \frac{1}{2} c \epsilon_0 A_{e} \left| U_{\ell, m} \right|^{2}$ is the power of $Y_{m} (t)$ with $U_{\ell, m}$ being its amplitude, and $y_{m} (t) = \sqrt{\mathcal{P}_{\ell, m}} \exp( \jmath \theta_{\ell, m} )$ is the equivalent baseband signal. Since the LO is pre-designed, we can shape $U_{\ell, m}$ and $\theta_{\ell, m}$ into any configuration.

	The RF signal to be detected is a superimposition of signals impinging from the $K$ users and the LO. We first define $f_{\delta} = f_c - f_{\ell}$ and $\theta_{\delta, m, k} = \theta_{x, m, k} - \theta_{\ell, m}$ as the frequency difference and phase difference between the $k$-th user and the LO at the $m$-th receiver sensor. Then we obtain the superimposed signal received at the $m$-th sensor of the RAQ-MIMO as follows 
	\begin{align}
		\nonumber
		&Z_{m} (t) 
		= Y_{m} (t) + \sum\nolimits_{k=1}^{K} X_{m, k} (t) \\
		\nonumber
		&= \sqrt{2} \mathscr{R} \left\{ \left[ y_{m} (t) + \sum_{k=1}^{K} x_{m, k} (t) \exp( \jmath 2 \pi f_{\delta} t ) \right] \exp( \jmath 2 \pi f_{\ell} t ) \right\} \\
		\label{eq:PassbandRFplusLO}
		&\qquad \;\; \approx \sqrt{2\mathcal{P}_{z, m}} \cos( 2 \pi f_{\ell} t + \theta_{\ell, m} ),
	\end{align}
	where $\mathcal{P}_{z,m} = \frac{1}{2} c \epsilon_0 A_{e} \left| U_{z,m} \right|^{2}$ denotes the power of $Z_{m} (t)$ with its amplitude $U_{z, m}$ given by \eqref{eq:Amp_PassbandRFplusLO} at the top of the next page. 
	\begin{figure*}[t!]
		\begin{align}
			\label{eq:Amp_PassbandRFplusLO}
			U_{z, m} 
			= \sqrt{ U_{\ell, m}^{2} + \sum_{k=1}^{K} U_{x, m, k}^{2} + 2 U_{\ell, m} \sum_{k=1}^{K} U_{x, m, k} \cos( 2 \pi f_{\delta} t + \theta_{\delta, m, k} ) } 
			\approx U_{\ell, m} + \sum_{k=1}^{K} U_{x, m, k} \cos( 2 \pi f_{\delta} t + \theta_{\delta, m, k} ).
		\end{align}
		\hrulefill
	\end{figure*}
	The equivalent baseband signal of $Z_{m}(t)$ is thus $z_{m} (t) = \sqrt{\mathcal{P}_{z, m}} \exp( \jmath \theta_{\ell, m} )$. See Appendix \ref{Appendix:RFsuperpositionProof} for proofs of \eqref{eq:PassbandRFplusLO}, \eqref{eq:Amp_PassbandRFplusLO}.

	\vspace{-0.8em}
	\subsection{Quantum Evolution and Steady-State Solution}
	\label{SubSection:AtomicTransitionModel}
	
	The superimposed RF signal together with the laser beams jointly influence the quantum evolution of Rydberg atoms through their Rabi frequencies when coupling the four-level scheme aforementioned in Section \ref{subsub:Transition}. Furthermore, let us denote the Rabi frequency of the superimposed RF signal, probe beam, and coupling beam by $\Omega_{{\rm{RF}},m}$, $\Omega_{p,m}$, and $\Omega_{c,m}$, respectively, for the $m$-th receiver sensor. Upon exploiting the relationship between the Rabi frequency and signal amplitude via $\Omega = \frac{ \mu_{34} }{ \hbar } U =  \frac{ \mu_{34} }{ \hbar } \sqrt{ \frac{2\mathcal{P}}{A_e c \epsilon_0} }$, we can associate $U_{z,m}$, $U_{\ell,m}$, $U_{x,m,k}$ in \eqref{eq:Amp_PassbandRFplusLO} with their Rabi frequencies of $\Omega_{{\rm{RF}}, m}$, $\Omega_{\ell,m}$, $\Omega_{x,m,k}$, respectively, and obtain the following expression 
	\begin{align}
		\label{eq:RFRabiFrequency} 
		\Omega_{{\rm{RF}}, m} \approx \Omega_{\ell,m} + \sum_{k=1}^{K} \Omega_{x,m,k} \cos{(2 \pi f_{\delta} t + \theta_{\delta, k,m})},  
	\end{align}
	where we have $\Omega_{\ell,m} \gg \sum_{k=1}^{K} \Omega_{x,m,k}$. Based on this, the quantum evolution of the four-level electron transition of Fig. \ref{fig:RAQ-MIMO}(b) is characterized by the Lindblad master equation \cite{auzinsh2010optically}
	\begin{align}
		\label{eq:Lindbladian}
		\frac{\mathrm{d} \bm{\rho}_{m}}{\mathrm{d} t} 
		= - \jmath \comm{\bm{\mathcal{H}}_{m}}{\bm{\rho}_{m}} - \frac{1}{2} \left\{ \bm{\varGamma}_{m}, \bm{\rho}_{m} \right\} + \bm{\varLambda}_{m},
	\end{align}
	where $\bm{\mathcal{H}}_{m}$, $\bm{\varGamma}_{m}$ and $\bm{\varLambda}_{m}$ are the Hamiltonian, relaxation matrix, and decay matrix for Rydberg atoms in the $m$-th receiver sensor, respectively. They are given by 
	\begin{align}
		\label{eq:Hamiltonian}
		\bm{\mathcal{H}}_{m} 
		= \left[ {\begin{array}{*{20}{c}}
				0 & {\frac{{{\Omega_{p, m}}}}{2}} & 0 & 0\\
				{\frac{{{\Omega_{p, m}}}}{2}} & {{\Delta_{p, m}}} & {\frac{{{\Omega_{c, m}}}}{2}} & 0\\
				0 & {\frac{{{\Omega_{c, m}}}}{2}} & \sum\limits_{u = p,c} \Delta_{u, m} & {\frac{{{\Omega_{{\rm{RF}},m}}}}{2}}\\
				0 & 0 & {\frac{{{\Omega_{{\rm{RF}},m}}}}{2}} & \sum\limits_{u = p,c,\ell} \Delta_{u, m}
		\end{array}} \right],
	\end{align}
	$\bm{\varGamma}_{m} = \rm{diag} \{ \gamma, \gamma + \gamma_2, \gamma + \gamma_3 + \gamma_c, \gamma + \gamma_4 \}$, and $\bm{\varLambda}_{m} = \rm{diag} \{ \gamma + \gamma_2 \rho_{22} + \gamma_4 \rho_{44}, \gamma_3 \rho_{33}, 0, 0 \}$, where $\gamma$ and $\gamma_c$ represent the relaxation rates related to the atomic transition effect and collision effect, respectively. 
	
	We can obtain an analytical expression of the steady-state solution of the density matrix $\bm{\rho}_{m}$ with certain assumptions, where we assume $\gamma = \gamma_c = 0$ and $\gamma_3 = \gamma_4 = 0$ as the decay rates of $\ket{3}$ and $\ket{4}$ are comparatively small and can be reasonably ignored. Specifically, the element $\rho_{21, m}$ of the matrix $\bm{\rho}_{m}$ is relevant to the probe beam that is directly linked to the measurements.  This relationship can be achieved by first connecting $\rho_{21, m}$ to the susceptibility of a single Rydberg atom via $\chi_{m} (\Omega_{{\rm{RF}},m}) = - \frac{ \varsigma }{ \Omega_{p,m} } \rho_{21,m}(\Omega_{{\rm{RF}},m})$, which is obtained in the steady state as follows
	\begin{align}
		\nonumber
		\chi_{m} (\Omega_{{\rm{RF}},m}) 
		&= - \varsigma \frac{{{A_{1,m}}\Omega_{{\rm{RF}},m}^4 + {A_{2,m}}\Omega_{{\rm{RF}},m}^2 + {A_{3,m}}}}{{{C_{1,m}}\Omega_{{\rm{RF}},m}^4 + {C_{2,m}}\Omega_{{\rm{RF}},m}^2 + {C_{3,m}}}} \\
		\label{eq:Susceptibility}
		&\;+ \jmath \varsigma \frac{{{B_{1,m}}\Omega_{{\rm{RF}},m}^4 + {B_{2,m}}\Omega_{{\rm{RF}},m}^2 + {B_{3,m}}}}{{{C_{1,m}}\Omega_{{\rm{RF}},m}^4 + {C_{2,m}}\Omega_{{\rm{RF}},m}^2 + {C_{3,m}}}}, 
	\end{align} 
	where $\varsigma \triangleq \frac{2 \mu_{12}^{2}}{\epsilon_0 \hbar }$; $A_{1,m}$, $A_{2,m}$, $A_{3,m}$, $B_{1,m}$, $B_{2,m}$, $B_{3,m}$, $C_{1,m}$, $C_{2,m}$, $C_{3,m}$ are given in (51)-(59) of Appendix A in \cite{gong2026rydberg}.

	\vspace{-0.8em}
	\subsection{RF-Atomic-Optical Transformation of RAQ-MIMO}
	\label{SubSection:TransferFunction}
	
	Let us denote the amplitude, frequency, and initial phase of the $m$-th probe beam at the access area of the vapor cell, respectively, by $U_{0,m}$, $f_{p}$, and $\phi_{0,m}$. After propagating through the vapor cell, the amplitude and phase of the probe beam are affected by the Rydberg atoms at the output area of the vapor cell. Upon denoting the amplitude and the phase of the $m$-th output probe beam, respectively, by $U_{p, m} (\Omega_{{\rm{RF}},m})$ and $\phi_{p, m} (\Omega_{{\rm{RF}},m})$, they can be associated with their input counterparts obeying \cite{gong2026rydberg}
	\begin{align}
		\label{eq:AmplitudeRelation}
		U_{p, m} (\Omega_{{\rm{RF}},m}) 
		&= U_{0,m} \exp( - \tfrac{k_p L N_0}{2} \mathscr{I} \left\{ \chi_{m} (\Omega_{{\rm{RF}},m}) \right\} ), \\
		\label{eq:PhaseRelation}
		\phi_{p, m} (\Omega_{{\rm{RF}},m})  
		&= \phi_{0,m} + \tfrac{k_p L N_0}{2} \mathscr{R} \left\{ \chi_{m} (\Omega_{{\rm{RF}},m}) \right\},
	\end{align}
	where $k_p$ is the wavenumber of the probe beam, $L$ is the length of the vapor cell, and $N_0$ denotes the effective atomic density within the vapor cell. Based on \eqref{eq:AmplitudeRelation} and \eqref{eq:PhaseRelation}, the probe beam at the output of the vapor cell is obtained as 
	\begin{align}
		\nonumber
		\hspace{-0.8em} P_{m} (\Omega_{{\rm{RF}},m}, t) 
		&= \sqrt{ 2\mathcal{P}_{m} (\Omega_{{\rm{RF}},m}) } \cos \left[ 2 \pi f_{p} t + \phi_{p,m} (\Omega_{{\rm{RF}},m}) \right] \\
		\label{eq:OutputProbeSig}
		&\hspace{-0.2em} = \sqrt{2} \mathscr{R} \left\{ P_{b,m} (\Omega_{{\rm{RF}},m}, t) \exp( \jmath 2 \pi f_p t ) \right\}, 
	\end{align}
	where $P_{b,m} (\Omega_{\rm{RF}}, t) \triangleq \sqrt{\mathcal{P}_{m} (\Omega_{{\rm{RF}},m})} \exp( \jmath \phi_p (\Omega_{{\rm{RF}},m}) )$ is the equivalent baseband signal of this output (passband) probe beam, having a full width at half maximum (FWHM) of $F_{p}$ and a power of $\mathcal{P}_{m} (\Omega_{{\rm{RF}},m}) = \frac{\pi c \epsilon_0}{8 \ln {2}} F_p^2 \left| U_{p,m} (\Omega_{{\rm{RF}},m}) \right|^{2}$. Upon denoting the power of the input probe beam by $\mathcal{P}_{0,m} = \frac{\pi c \epsilon_0}{8 \ln {2}} F_p^2 \left| U_{0,m} \right|^{2}$, we reformulate $\mathcal{P}_{m} (\Omega_{{\rm{RF}},m})$ as follows 
	\begin{align}
		\label{eq:OutputProbePower}
		\hspace{-0.5em} 
		\mathcal{P}_{m} (\Omega_{{\rm{RF}},m}) 
		= \mathcal{P}_{0,m} \exp \big( - k_p L N_0 \mathscr{I} \left\{ \chi_{m} (\Omega_{{\rm{RF}},m}) \right\} \big). 
	\end{align} 

	\begin{figure*}[t!]
		\begin{align}
			\nonumber
			V_{m}(t) &= 2\alpha \sqrt{G_{\text{pd}} \mathcal{P}_{m}^{\text{(lob)}} \mathcal{P}_{m} (\Omega_{{\rm{RF}},m})}
			\cos \left( {\phi_{m}^{\text{(lob)}}} - {\phi_{p,m}}({\Omega_{{\rm{RF}},m}}) \right)
			\approx 2\alpha \sqrt{G_{\text{pd}} \mathcal{P}_{m}^{\text{(lob)}} \mathcal{P}_{m} (\Omega_{\ell,m})}  \\ 
			\label{eq:PhotodetectorOutputBCOD} 
			&\qquad \qquad \times
			\left[ \cos \left( {\phi_{m}^{\text{(lob)}}} - {\phi_{p,m}}({\Omega_{\ell,m}}) \right) -  \kappa_{m}({\Omega_{\ell,m}}) \cos {\varphi_{m}({\Omega_{\ell,m}})} \sum\nolimits_{k=1}^{K} U_{x,m,k} \cos \left( {2\pi {f_\delta }t + {\theta_{\delta,m,k} }} \right) \right]. 
		\end{align}
		\hrulefill
	\end{figure*}

	\vspace{-0.6em}
	\subsection{Photodetection of RAQ-MIMO}
	\label{subsec:RAQS}
	
	Furthermore, the $m$-th probe beam will be detected by the $m$-th PD. For each PD, we employ the balanced coherent optical detection (BCOD) scheme as a benefit of its excellent performance (PSL achievable) and its popularity in physical experiments, as shown in Fig. \ref{fig:RAQ-MIMO}(c) and studied in \cite{gong2026rydberg}. In this scheme, $M$ local optical beam (LOB) sources exist, where each LOB source is subtracted and added to the $m$-th probe beam \eqref{eq:OutputProbeSig} to form a pair of mixed optical signals, respectively. These mixed laser beams are then detected by two independent PDs, followed by a substractor and a low noise amplifier (LNA) (gain is $G_{\text{pd}}$) in order to output the final signal.  
	
	Let us assume that all $M$ LOB sources are Gaussian beams and denote the $m$-th LOB signal as 
	\begin{align}
		\nonumber
		P_{m}^{\text{(lob)}} (t) 
		&= \sqrt{2\mathcal{P}_{m}^{\text{(lob)}}} \cos ( 2 \pi f_p t + \phi_{m}^{\text{(lob)}} ) \\
		\label{eq:LOB} 
		&= \sqrt{2} \mathscr{R} \left\{ P_{b,m}^{\text{(lob)}} (t) \exp( \jmath 2 \pi f_p t ) \right\},  
	\end{align}
	where $P_{b,m}^{\text{(lob)}} (t) = \sqrt{\mathcal{P}_{m}^{\text{(lob)}}} \exp( \jmath \phi_{m}^{\text{(lob)}} )$ represents the equivalent baseband signal  with power $\mathcal{P}_{m}^{\text{(lob)}} = \frac{\pi c \epsilon_0}{8 \ln {2}}  F_p^2 \left| U_{m}^{\text{(lob)}} \right|^{2}$ and amplitude $U_{m}^{\text{(lob)}}$. Upon defining PD's responsivity $\alpha \triangleq \frac{\eta q}{\hbar \omega_{p}}$, we obtain the output of the $m$-th BCOD  via $V_{m}(t) = \alpha \sqrt{ G_{\text{pd}} }$ $[ {P_{b,m}^{\text{(lob)}}}\left( t \right) {P_{b,m}^*}\left( {{\Omega_{{\rm{RF}},m}},t} \right) + P_{b,m}^{\text{(lob)}*}(t) P_{b,m} \left( {{\Omega_{{\rm{RF}},m}},t} \right) ]$, which is furthermore obtained in \eqref{eq:PhotodetectorOutputBCOD}, as seen at the top of the this page, where the approximation is obtained through the first-order Taylor series expansion. Furthermore, $\kappa_{m} (\Omega_{\ell,m})$ and $\varphi_{m}({\Omega_{\ell,m}})$ in \eqref{eq:PhotodetectorOutputBCOD} are expressed as follows 
	\begin{align}
		\vspace{-0.2em}
		\nonumber
		&\kappa_{m} (\Omega_{\ell,m}) = L {N_0} \mathcal{C}_m (\Omega_{\ell,m}) = N_{\text{atom},m} \tfrac{ {\cal C} (\Omega_{\ell,m}) }{A_{p,m}}, \\
		\nonumber
		&{\varphi_{m}({\Omega_{\ell,m}})} = {\phi_{m}^{\text{(lob)}}} - {\phi_{p,m}}({\Omega_{\ell,m}}) + \arccos \left[ \tfrac{ \varsigma_1 {\mathscr I}\{ \chi_{m}^{\prime} ({\Omega_{\ell,m}})\} }{ \mathcal{C}_m (\Omega_{\ell,m}) } \right], \\
		\nonumber
		&\mathcal{C}_m (\Omega_{\ell,m}) = \varsigma_1 \sqrt{ \big| {\mathscr I} \left\{ \chi_{m}^{\prime} ({\Omega_{\ell,m}}) \right\} \big|^{2} + \big| {\mathscr R} \left\{ \chi_{m}^{\prime} ({\Omega_{\ell,m}}) \right\} \big|^{2} }, 
	\end{align}
	where $N_{\text{atom},m}$ is the effective number of Rydberg atoms in the $m$-th sensor, $A_{p,m}$ is the cross-sectional area of the $m$-h sensor, and $\varsigma_1 \triangleq \frac{ \mu_{34} }{ 2 \hbar } k_p$. Notably, we have $N_{\text{atom},m} = N_0 V_m$, where $V_m$ is the volume of a sensor containing Rydberg atoms. Upon assuming that each sensor is a cylinder with radius $r_0 = \frac{F_p}{\sqrt{2 \ln{2}}}$, we have $A_{p,m} = \pi r_0^2 = \frac{\pi}{2 \ln{2}} F_{p,m}^2$ and $V_m = A_{p,m} L$.

	We note that $\kappa_{m} (\Omega_{\ell,m})$ is the \emph{collective atomic quantum responsivity} of the $m$-th sensor, and $\arccos \left[ \frac{ \varsigma_1 {\mathscr I}\{ \chi_{m}^{\prime} ({\Omega_{\ell,m}})\} }{ \mathcal{C}_m (\Omega_{\ell,m}) } \right]$ is an extra phase (introduced by the BCOD scheme) to the probe beam on the basis of ${\phi_{p,m}}({\Omega_{\ell,m}})$ characterized in \eqref{eq:PhaseRelation}. Notably, the larger the value of $\kappa_{m} (\Omega_{\ell,m})$, the higher the sensitivity. In expressions of ${\varphi_{m}({\Omega_{\ell,m}})}$, $\mathcal{C}_m (\Omega_{\ell,m})$, we have 
	\begin{align}
		\vspace{-0.2em}
		\nonumber
		&\mathscr{R} \left\{ \chi_{m}^{\prime} \left( \Omega_{\ell,m} \right) \right\} 
		= - 2 \varsigma {\Omega_{\ell,m}} \left[ \frac{ {2{A_{1,m}}\Omega_{\ell,m}^2 + {A_{2,m}}} }{{{C_{1,m}}\Omega_{\ell,m}^4 + {C_{2,m}}\Omega_{\ell,m}^2 + {C_{3,m}}}} \right. \\
		\nonumber
		&\quad \left. - \frac{ \big( {{A_{1,m}}\Omega_{\ell,m}^4 + {A_{2,m}}\Omega_{\ell,m}^2 + {A_{3,m}}} \big) \big( {2{C_{1,m}}\Omega_{\ell,m}^2 + {C_{2,m}}} \big) }{ \big( {{C_{1,m}}\Omega_{\ell,m}^4 + {C_{2,m}}\Omega_{\ell,m}^2 + {C_{3,m}}} \big)^2 } \right], \\
		\nonumber
		&\mathscr{I} \left\{ \chi_{m}^{\prime} \left( \Omega_{\ell,m} \right) \right\}
		= 2 \varsigma {\Omega_{\ell,m}} \left[ \frac{ {2{B_{1,m}}\Omega_{\ell,m}^2 + {B_{2,m}}} }{{{C_{1,m}}\Omega_{\ell,m}^4 + {C_{2,m}}\Omega_{\ell,m}^2 + {C_{3,m}}}} \right. \\
		\nonumber
		&\quad \left. - \frac{ \big( {{B_{1,m}}\Omega_{\ell,m}^4 + {B_{2,m}}\Omega_{\ell,m}^2 + {B_{3,m}}} \big) \big( {2{C_{1,m}}\Omega_{\ell,m}^2 + {C_{2,m}}} \big) }{ \big( {{C_{1,m}}\Omega_{\ell,m}^4 + {C_{2,m}}\Omega_{\ell,m}^2 + {C_{3,m}}} \big)^2 } \right].
	\end{align}

	\subsection{Down-Conversion and Sampling}
	We now employ $M$ HRs for down-converting \eqref{eq:PhotodetectorOutputBCOD} in a one-to-one manner, as seen in Fig. \ref{fig:RAQ-MIMO}(d). Specifically, the $m$-th HR will remove the first (constant) term of $V_{m}(t)$, while retaining the second (varying) term, yielding 
	\begin{align}
		\nonumber
		\tilde{V}_{m}(t) 
		&= 2\alpha \sqrt{G_{\text{pd}} \mathcal{P}_{m}^{\text{(lob)}} \mathcal{P}_{m} (\Omega_{\ell,m})} 
		\cos {\varphi_{m}({\Omega_{\ell,m}})} \\
		\nonumber
		&\times \kappa_{m}({\Omega_{\ell,m}}) \sum_{k=1}^{K} U_{x,m,k} \cos \left( {2\pi {f_\delta }t + {\theta_{\delta,m,k} }} \right) \\
		\label{eq:PhotodetectorOutput_Varying}
		&= \sqrt{2} \mathscr{R} \left\{ v_{m}(t) \exp( \jmath 2\pi f_{\delta} t ) \right\}.   
	\end{align}
	By exploiting the relationship $U_{x,m,k} = \sqrt{ \frac{ 2 \mathcal{P}_{x,m,k} }{A_e c \epsilon_0} }$, we obtain the equivalent baseband signal as 
	\begin{align}
		\nonumber
		v_{m}(t) 
		&= 2\alpha \sqrt{G_{\text{pd}} \mathcal{P}_{m}^{\text{(lob)}} \mathcal{P}_{m} (\Omega_{\ell,m})} 
		\cos {\varphi_{m}({\Omega_{\ell,m}})} \\
		\nonumber
		&\times \kappa_{m}({\Omega_{\ell,m}}) \sum\nolimits_{k=1}^{K} U_{x,m,k} \exp( \jmath \theta_{x, m, k} - \jmath \theta_{\ell, m} ) \\
		 \label{eq:Baseband_Varying}
		&\triangleq \sqrt{ \frac{\varrho_{m}}{A_e} } \varPhi_{m} \sum_{k=1}^{K} x_{m,k}(t), 
	\end{align}
	where $\varrho_{m}$ and $\varPhi_{m}$ denote the gain and phase shift corresponding to the $m$-th receiver sensor, respectively, given by
	\begin{align}
		\label{eq:Gain}
		\varrho_{m} &= \frac{ 4 \alpha^2 }{ c \epsilon_0 } G_{\text{pd}} \mathcal{P}_{m}^{\text{(lob)}} \mathcal{P}_{m} (\Omega_{\ell,m}) \kappa_{m}^{2}({\Omega_{\ell,m}}), \\
		\nonumber
		\varPhi_{m} &= \frac{1}{2} \exp( - \jmath \left[ {\theta_{\ell,m}} - \varphi_{m} ({\Omega_{\ell,m}}) \right] ) \\
		\label{eq:Phase}
		&\qquad + \frac{1}{2} \exp( - \jmath \left[ {\theta_{\ell,m}} + \varphi_{m} ({\Omega_{\ell,m}}) \right] ). 
	\end{align}
	
	Sampling both sides of \eqref{eq:Baseband_Varying} at multiples of the sampling rate, we arrive at the sampled output as follows
	\begin{align}
		\label{eq:SampledOutput}
		v_{m} (n) 
		&= \sqrt{ \frac{\varrho_{m}}{A_e} } \varPhi_{m} \sum_{k=1}^{K} x_{m,k}(n). 
	\end{align}

	\vspace{-0.6em}
	\subsection{Signal Model of Multi-User RAQ-MIMO}
	
	To further construct an end-to-end signal model from multi-user signals to $v_{m} (n)$, we consider a narrowband transmission/sensing action and express $x_{m,k}(n)$ in the form of $x_{m,k}(n) = \sqrt{A_e} h_{m,k}(n) s_{k}(n)$ \cite{gong2026rydberg}, where $h_{m,k}(n)$ represents the wireless channel between the $k$-th user and the $m$-th receiver sensor, and $s_{k}(n)$ is the signal transmitted by the $k$-th user. By incorporating the noise, we reformulate \eqref{eq:SampledOutput} as 
	\begin{align}
		\label{eq:SampledOutput_E2E}
		v_{m} (n) 
		= \sqrt{ \varrho_{m} } \varPhi_{m} \sum_{k=1}^{K} h_{m,k}(n) s_{k}(n) + w_{m}(n), 
	\end{align}
	where we have $w_{m} \sim \mathcal{CN} (0, \sigma_{m}^2)$, $\sigma_{m}^2 = \frac{ N_{\rm{QPN}} + N_{\rm{PSN}} + N_{\rm{ITN}} }{2}$ is the sum of the quantum projection noise (QPN) power, of the photon shot noise (PSN) power, and of the intrinsic thermal noise (ITN) power \cite{gong2026rydberg}. Upon expressing \eqref{eq:SampledOutput_E2E} in matrix form by collecting all $M$ measurements, we arrive at  
	\begin{align}
		\label{eq:SignalModel_MatrixForm}
		\bm{v} = \bm{\varTheta} \bm{H} \bm{s} + \bm{w},
	\end{align}
	where $\bm{\varTheta} = {\rm{diag}} \{ \sqrt{\varrho_1} \varPhi_1, \sqrt{\varrho_2} \varPhi_2, \cdots, \sqrt{\varrho_M} \varPhi_M \} \in \mathbb{C}^{M \times M}$, $\bm{H} \in \mathbb{C}^{M \times K}$ denotes the channel matrix, $\bm{s} \in \mathbb{C}^{K \times 1}$ is the transmit signal vector, and $\bm{w} \in \mathbb{C}^{M \times 1}$ represents the complex additive white Gaussian noise (AWGN) vector. More particularly, we have $\bm{w} \sim \mathcal{CN} (\bm{0}, \bm{R}_{w})$, where $\bm{R}_{w} = \text{diag}\{ \sigma_{1}^2, \sigma_{2}^2, \cdots, \sigma_{M}^2 \}$ represents the noise covariance matrix. 
	
	We emphasize that all $M$ gains $\varrho_{m}$, $m=1, \cdots, M$ may be different in realistic implementations due to having different practical impairments that influence ${\cal P}_{m}^{\text{(lob)}}$, $\mathcal{P}_{0,m}$, ${\cal P}_{m} (\Omega_{\ell,m})$, ${\kappa_{m}}({\Omega_{\ell,m}})$ and $\varphi_{m} (\Omega_{\ell,m})$ for all $M$ sensors. By contrast, these gains could be theoretically identical, if we employ ${\cal P}_{m}^{\text{(lob)}} \triangleq {\cal P}_{\text{lob}}$, $\phi_{m}^{\text{(lob)}} \triangleq \phi_{\text{lob}}$, $\mathcal{P}_{0,m} \triangleq \mathcal{P}_{0}$ for all $M$ sensors, as well as assume same $L$ and $N_0$ for all vapor cells. Furthermore, we theoretically have ${\cal P}_{m} (\Omega_{\ell,m}) \triangleq {\cal P} (\Omega_{\ell})$, ${\kappa_{m}}({\Omega_{\ell,m}}) \triangleq {\kappa}({\Omega_{\ell}})$, $\varphi_{m} (\Omega_{\ell,m}) \triangleq \varphi (\Omega_{\ell})$, and $\mathcal{C}_m (\Omega_{\ell,m}) \triangleq \mathcal{C} (\Omega_{\ell})$ for all $M$ sensors. More particularly, $\Omega_{{\ell,m}} \triangleq \Omega_{\ell}$ requires that the LO imposes identical signal strength on all $M$ sensors. This can be realized by implementing the LO via the parallel-plate based method \cite{simons2019embedding, arumugam2024remote}, where the LO signal is a plane wave (as seen in Fig. \ref{fig:RAQ-MIMO}(a)) and the signal strength is identical everywhere within the parallel-plate. Based on all the above configurations, \eqref{eq:Gain} becomes  
	\begin{align}
		\label{eq:GainBCOD_refoumulate}
		\varrho_{m} = \frac{ 4 \alpha^2 }{ c \epsilon_0 } G_{\text{pd}} \mathcal{P}_{\text{lob}} \mathcal{P} (\Omega_{\ell}) \kappa^{2}({\Omega_{\ell}}) \triangleq \varrho. 
	\end{align}

	Furthermore, let us denote the angle of arrival of the LO signal by  $\vartheta$. Upon defining the phase of the LO signal at the first (reference) receiver sensor as $\theta_{\ell,1}$, we obtain the phase of the LO signal at the $m$-th sensor as $\theta_{\ell,m} = \theta_{\ell,1} + \frac{2\pi}{\lambda} (m-1) d \sin \vartheta$. Therefore, \eqref{eq:Phase} becomes 
	\begin{align}
		\label{eq:PhaseBCOD_refoumulate}
		\varPhi_{m} = \varPhi \exp( - \jmath \frac{2\pi}{\lambda} (m-1) d \sin \vartheta ), 
	\end{align}
	where $\varPhi \triangleq \frac{\exp( - \jmath \left[ {\theta_{\ell,1}} - \varphi ({\Omega_{\ell}}) \right] )}{2} + \frac{\exp( - \jmath \left[ {\theta_{\ell,1}} + \varphi ({\Omega_{\ell}}) \right] )}{2}$ represents the phase shift of the first (reference) receiver sensor. Based on \eqref{eq:GainBCOD_refoumulate} and \eqref{eq:PhaseBCOD_refoumulate}, we reformulate $\bm{\varTheta} = \sqrt{\varrho} \varPhi \bm{Q}$, where $\bm{Q} = {\rm{diag}} \{ 1, \exp( - \jmath \frac{2\pi}{\lambda} d \sin \vartheta ), \cdots, \exp( - \jmath \frac{2\pi}{\lambda} (M-1) d \sin \vartheta ) \}$.  Therefore, we can reformulate \eqref{eq:SignalModel_MatrixForm} as 
	\begin{align}
		\label{eq:SignalModel_MatrixForm_SC}
		\bm{v} = \sqrt{ \varrho } \varPhi \bm{Q} \bm{H} \bm{s} + \bm{w},  
	\end{align}
	where we have $\bm{w} \sim \mathcal{CN} (\bm{0}, \sigma^2 \bm{I}_{M})$, $\sigma_{m}^2 = \sigma^2$ for any $m$.

	\textit{\textbf{Remark 1}: The baseband expressions \eqref{eq:SignalModel_MatrixForm} and \eqref{eq:SignalModel_MatrixForm_SC} constitute a quantum-to-classical transformation framework, encapsulating quantum-optical parameters, e.g., collective atomic quantum responsivity $\kappa_{m}({\Omega_{l,m}})$ and atom-related phase $\varphi_{m} ({\Omega_{l,m}})$, while offering a classical interface that is convenient for adapting classical signal processing approaches to RAQ-MIMO systems. Although this might resemble a classical form, the physical process, noise statistics, and scaling laws of the RAQ-MIMO receiver differ fundamentally from classical antenna arrays, as discussed in Section \ref{subsec:SQLandPSL} and \ref{subsubsec:Insights}.}

	\textit{\textbf{Remark 2}: We consider narrowband RAQ-MIMO receivers, where Rydberg atoms can reach their steady-state quantum response within each symbol period, allowing RAQ-MIMO receivers to behave as an approximately channel-independent linear transformation, as modeled in \eqref{eq:SampledOutput_E2E}, \eqref{eq:SignalModel_MatrixForm}, and \eqref{eq:SignalModel_MatrixForm_SC}. By contrast, the wideband operations require to model non-steady-state quantum dynamics and may introduce coupled channel–atom effects that are left for future work.}

	\section{Ergodic Achievable Rate of RAQ-MIMO}
	\label{sec:EAR}
	
	We assume that the RAQ-MIMO receiver has perfect channel state information (CSI). The linear receiver combining matrix used for recovering the transmit signal is denoted by $\bm{C}$. Then the $k$-th user's signal estimate is given by 
	\begin{align}
		\label{eq:ReceivedSignal_k}
		\hspace{-0.2em}
		r_{k} 
		= {\bm{ c}}_k^* {\bm{R}}_w^{ -\frac{1}{2} } {\bm{\varTheta }} {{\bm{h}}_k} {s_k} + \sum_{i=1 \atop i \ne k}^{K} {{\bm{ c}}_k^*{\bm{R}}_w^{ -\frac{1}{2} }{\bm{\varTheta }}{{\bm{h}}_i}{s_i}}  + {\bm{ c}}_k^*{\bm{R}}_w^{ -\frac{1}{2} }{\bm{w}}. 
	\end{align} 
	Thus we can formulate the EAR and SINR of the $k$-th user as 
	\begin{align}
		\label{eq:AchievableRate_MRC}
		R_{k} 
		&= \mathbb{E} \left\{ \log_{2} \left( 1 + \mathsf{SINR}_{k} \right) \right\}, \\
		\label{eq:Gamma_MRC}
		\mathsf{SINR}_{k}
		&= 
		\frac{ {\cal P}_{s} {{\left| {{\bm{c}}_k^*{\bm{R}}_w^{ - \frac{1}{2} }{\bm{\varTheta }}{{\bm{h}}_k}} \right|}^2} }{ {\cal P}_{s} \sum_{i = 1, \ne k}^K {{{\left| {{\bm{c}}_k^*{\bm{R}}_w^{ - \frac{1}{2} }{\bm{\varTheta }}{{\bm{h}}_i}} \right|}^2}}  + {\bm{c}}_k^*{{\bm{c}}_k} }. 
	\end{align}

	With the aid of perfect CSI, we can employ a pair of typical linear combining receivers, namely the MRC and ZF schemes 
	\begin{align}
		\label{eq:CombiningMatrix}
		\bm{C}
		= \begin{cases}
			{\bm{R}}_w^{-\frac{1}{2}} \bm{\varTheta} \bm{H} & {\rm{MRC}}, \\
			{\bm{R}}_w^{-\frac{1}{2}} \bm{\varTheta} \bm{H} \left( \bm{H}^{*} \bm{\varTheta}^{*} {\bm{R}}_w^{-1}  \bm{\varTheta} \bm{H} \right)^{-1} & {\rm{ZF}}. 
		\end{cases}
	\end{align} 
	
	Furthermore, we assume that the multi-user RAQ-MIMO channel is in the form of $\bm{H} = \bm{R}^{\frac{1}{2}} \bm{G} \bm{T}^{\frac{1}{2}} \bm{D}^{\frac{1}{2}}$, where $\bm{R} \in \mathbb{C}^{M \times M}$ and $\bm{T} \in \mathbb{C}^{K \times K}$ represent the receiver-side and transmitter-side correlation matrices, $\bm{G} \in \mathbb{C}^{M \times K}$ has entries obeying $g_{mk} \sim \mathcal{CN}(0, 1)$, and $\bm{D} \in \mathbb{C}^{K \times K}$ is a diagonal matrix having the large-scale fading coefficient $\beta_{k}$ as its diagonal elements. Since the users are spatially-separated by a long distance, we can ignore their correlations, namely we have $\bm{T} = \bm{I}_{K}$. In the following, we investigate the UFC and the CFC characterized by the Weichselberger model \cite{weichselberger2006stochastic}, 
	\begin{align}
		\label{eq:ChannelModel}
		\bm{H} 
		= \begin{cases}
			\bm{G} \bm{D}^{\frac{1}{2}} & \rm{UFC}, \\
			\bm{U} \left( \bm{\varSigma} \odot \bm{G} \right) \bm{D}^{\frac{1}{2}} & \rm{CFC}. 
		\end{cases}
	\end{align}
	In \eqref{eq:ChannelModel}, $\bm{U} \in \mathbb{C}^{M \times M}$ is the unitary matrix of the eigenvalue decomposition of the receiver-side correlation matrix, hence we have $\bm{R} = \bm{U} \bm{\varLambda} \bm{U}^{*}$. Upon denoting the corresponding eigenvalue vector by $\bm{\lambda}^{\frac{1}{2}} = {\rm{diag}} \{ \bm{\varLambda}^{\frac{1}{2}} \} \in \mathbb{C}^{M \times 1}$ and defining $\bm{1}_{K} = \{ 1, 1, \cdots, 1 \}^{T} \in \mathbb{C}^{K \times 1}$, we have $\bm{\varSigma} = \bm{\lambda}^{\frac{1}{2}} \bm{1}_{K}^{T}$. The entries of $\bm{R}$ can be constructed by employing Jakes' model, namely $[\bm{R}]_{m_1,m_2} = J_0 (\varpi |m_1-m_2|)$, where $J_0 (\cdot)$ is the zero-order Bessel function of the first kind \cite{jakes1994microwave}.

	\vspace{-0.8em}
	\subsection{Asymptotic Analysis of the EAR}
	\label{subsec:LB}
	
	\subsubsection{MRC}
	First, we apply the MRC scheme for RAQ-MIMO in the form of \eqref{eq:SignalModel_MatrixForm}. Hence the expectation of the SINR \eqref{eq:Gamma_MRC} is expressed in \eqref{eq:Exp_gamma_MRC}, as seen at the top of the this page, 
	\begin{figure*}
		\begin{align}
			\label{eq:Exp_gamma_MRC}
			\mathbb{E} \left\{ \mathsf{SINR}_{k} \right\} 
			= 
			\left[ {\frac{{{\sigma ^2}}}{{{{\cal P}_s}}}} + \sum\limits_{i = 1, \ne k\hfill}^K \mathbb{E} \left\{ {{\left( \frac{{{\bm{h}}_k^*{{\bm{\varPhi }}^{\frac{1}{2}}}}}{{\sqrt {\left| {{\bm{h}}_k^*{\bm{\varPhi }}{{\bm{h}}_k}} \right|} }}{{\bm{\varPhi }}^{\frac{1}{2}}}{{\bm{h}}_i} \right)}^*} \left( \frac{{{\bm{h}}_k^*{{\bm{\varPhi }}^{\frac{1}{2}}}}}{{\sqrt {\left| {{\bm{h}}_k^*{\bm{\varPhi }}{{\bm{h}}_k}} \right|} }}{{\bm{\varPhi }}^{\frac{1}{2}}}{{\bm{h}}_i} \right) \right\} \right]^{-1} 
			\left[ \mathbb{E} \left\{ {\frac{1}{ {\bm{h}}_k^*{\bm{\varPhi }}{{\bm{h}}_k} }} \right\} \right]^{-1}.
		\end{align}
		\hrulefill
		\vspace{-0.3cm}
	\end{figure*}
	where we have employed the following definition
	\begin{equation}
		\label{eq:PHi}
		{\bm{\varPhi }} \triangleq {{\bm{\varTheta }}^*} {\bm{R}}_w^{-1} {\bm{\varTheta }} = {\rm{diag}} \left\{ \frac{\varrho_1|\varPhi_1|^2}{\sigma_1^2}, \cdots, \frac{\varrho_M|\varPhi_M|^2}{\sigma_M^2} \right\}. 
	\end{equation}
	For both UFC/CFC scenarios, we have the following Lemma.

	\begin{lemma}
		\label{lemma1}
		Considering model \eqref{eq:SignalModel_MatrixForm} and both UFC and CFC of \eqref{eq:ChannelModel}, as well as assuming a perfect CSI, we derive the EAR of the MRC RAQ-MIMO for the $k$-th user as follows 
		\begin{equation}
			\label{eq:AchievableRate_MRC_Asymptotic}
			\hspace{-0.9em}
			R_{k}^{\rm{(mrc)}} \ge {\log _2} \left[ 
			\begin{split}
				1 + \frac{ \left[ {{\rm{Tr}}^2\left( {\bm{\varOmega }} \right)} - {\rm{Tr}}\left( {{{\bm{\varOmega }}^2}} \right) \right] \beta_k }{ \frac{1}{{\cal P}_s} {\rm{Tr}}\left( {\bm{\varOmega }} \right) + {\rm{Tr}}\left( {{{\bm{\varOmega }}^2}} \right)\sum_{i = 1, \ne k}^K {{\beta _i}} } 
			\end{split}
			\right], 
		\end{equation}
		where ${\bm{\varOmega }} = {\bm{\varPhi }}$ for the UFC case and ${\bm{\varOmega }} = {{\bm{\varLambda }}^{\frac{1}{2}}}{{\bm{U}}^*}{\bm{\varPhi U}}{{\bm{\varLambda }}^{\frac{1}{2}}}$ for the CFC case. More particularly, when ${\bm{\varPhi }}$ forms a scaled identity matrix based on \eqref{eq:SignalModel_MatrixForm_SC}, namely we have ${\bm{\varPhi }} = \frac{\varrho \varPhi^{*} \varPhi}{\sigma^2} \bm{Q}^{*} \bm{Q} = \frac{\varrho {{\cos }^2}\varphi ({\Omega_{\ell}})}{\sigma^2} \bm{I}_{M}$, we can furthermore derive \eqref{eq:AchievableRate_MRC_Asymptotic} as 
		\begin{equation}
			\label{eq:AchievableRate_MRC_Asymptotic_1}
			\hspace{-0.9em} R_{k}^{\rm{(mrc)}} \ge
			{\log _2} \left[
			\begin{split}
				1 + \frac{ (M-\varepsilon) \mathsf{SNR}_{1,k} \cos^{2} \varphi (\Omega_{\ell}) }{ 1 + \varepsilon \frac{ \mathsf{SNR}_{1,k} \cos^{2} \varphi (\Omega_{\ell}) }{\beta_{k}} \sum_{i = 1, \ne k}^K {{\beta _i}} } 
			\end{split}
			\right], 
		\end{equation}
		where $\mathsf{SNR}_{1,k} \triangleq {\varrho} \frac{ {\cal P}_{s}}{\sigma^2} {\beta_k}$ represents the maximal received SNR corresponding to each RAQ-MIMO receiver sensor and the $k$-th user, without experiencing any small-scale fading and any inter-user interference (IUI); Additionally, $\varepsilon \in \left\{ 1, \zeta \right\}$ for the UFC and CFC, respectively, with $\zeta$ derived for odd $M$ and even $M$ as follows 
		\begin{align}
			\label{eq:zetaOdd}
			&\zeta_{\rm{odd}} 
			= 1 + \frac{2}{\pi \varpi}\left[ \epsilon + \ln \left( \left\lceil \frac{M}{2} \right\rceil - 1 \right) \right] + \left( \frac{1}{\varpi} - \frac{2}{\pi} \right), \\
			\nonumber
			&\zeta_{\rm{even}}  
			= 1 + \frac{2}{\pi \varpi} \left[ \epsilon + \ln \left( \frac{M}{2} - 1 \right) \right] + \left( \frac{1}{\varpi} - \frac{2}{\pi} \right) \\
			\label{eq:zetaEven}
			&\qquad \quad \;\; + \frac{4}{\pi \varpi} \left[ \frac{{1 + \sin \left( {\varpi M} \right)}}{M} \right].  
		\end{align}
	\end{lemma}
	\begin{IEEEproof}
		Based on \eqref{eq:ChannelModel}, we can express ${{\bm{h}}_k} = \sqrt {{\beta_k}} {{\bm{g}}_k}$ for the UFC case and reformulate ${{\bm{h}}_k} = \sqrt {{\beta _k}} {\bm{U}} ( {\bm{\varLambda }}^{\frac{1}{2}} \odot {{\bm{g}}_k} ) = \sqrt {{\beta _k}}$ ${\bm{U}} {{\bm{\varLambda }}^{\frac{1}{2}}}{{\bm{g}}_k}$ for the CFC case. Therefore, we can obtain ${{\bm{h}}_k^*{\bm{\varPhi }}{{\bm{h}}_i}} =$ $\beta_{k} {\bm{g}}_k^* {\bm{\varOmega }} {{\bm{g}}_i}$, where we have ${\bm{\varOmega }} \in \{ {\bm{\varPhi }}, {{\bm{\varLambda }}^{\frac{1}{2}}}{{\bm{U}}^*}{\bm{\varPhi U}}{{\bm{\varLambda }}^{\frac{1}{2}}} \}$ for the UFC and CFC, respectively. As a consequence, we can obtain the first expectation term in \eqref{eq:Exp_gamma_MRC} as follows 
		\begin{align}
			\nonumber
			&\mathbb{E} \left\{ {{\left( \frac{{{\bm{h}}_k^*{{\bm{\varPhi }}^{\frac{1}{2}}}}}{{\sqrt {\left| {{\bm{h}}_k^*{\bm{\varPhi }}{{\bm{h}}_k}} \right|} }}{{\bm{\varPhi }}^{\frac{1}{2}}}{{\bm{h}}_i} \right)}^*} \left( \frac{{{\bm{h}}_k^*{{\bm{\varPhi }}^{\frac{1}{2}}}}}{{\sqrt {\left| {{\bm{h}}_k^*{\bm{\varPhi }}{{\bm{h}}_k}} \right|} }}{{\bm{\varPhi }}^{\frac{1}{2}}}{{\bm{h}}_i} \right) \right\} \\
			\label{eq:E1} 
			&= {\beta _i}\frac{{{\bm{g}}_k^*{{\bm{\varOmega }}^2}{{\bm{g}}_k}}}{{{\bm{g}}_k^*{\bm{\varOmega }}{{\bm{g}}_k}}} 
			= {\beta _i} \frac{ {\rm{Tr}} \left( {\bm{g}}_k^*{{\bm{\varOmega }}^2}{{\bm{g}}_k} \right) }{ {\rm{Tr}} \left( {\bm{g}}_k^*{\bm{\varOmega }}{{\bm{g}}_k} \right) } 
			\xrightarrow[a.s.]{M \rightarrow \infty}  \frac{ {\rm{Tr}} \left( \bm{\varOmega}^{2} \right) }{ {\rm{Tr}} \left( \bm{\varOmega} \right) },  
		\end{align}
		where \eqref{eq:E1} is obtained by exploiting the asymptotic properties $\lim_{M \rightarrow \infty} \frac{ {\bm{g}}_{i}^{*}{\bm{\varOmega }^2}{{\bm{g}}_{i}} }{M} \xrightarrow[]{p.} \frac{1}{M} \sum_{m=1}^{M} \bar{\lambda}_{m}^2$ and $\lim_{M \rightarrow \infty} \frac{ {\bm{g}}_{i}^{*}{\bm{\varOmega }}{{\bm{g}}_{i}} }{M} \xrightarrow[]{p.} \frac{1}{M} \sum_{m=1}^{M} \bar{\lambda}_{m}$, where $\bar{\lambda}_{m}$ represents the eigenvalues of $\bm{\varOmega }$ \cite{kamga2016spectral}. 
		To obtain the second expectation term of \eqref{eq:Exp_gamma_MRC}, we resort to the moment matching approximation, where $X \triangleq {\bm{g}}_k^* {\bm{\varOmega }} {{\bm{g}}_i}$ can be accurately approximated by a Gamma random variable of $X\sim \varGamma(k,\theta)$ having the expectation $\mathbb E[X]=k\theta$, variance $\operatorname{Var}(X)=k\theta^2$, and the probability density function $f_X(x) = \frac{x^{k-1}e^{-x/\theta}}{\varGamma(k)\theta^k}$, $x>0$. Since we additionally obtain $\mathbb E[X] = \operatorname{Tr}(\bm\varOmega)$ and $\operatorname{Var}(X)=\operatorname{Tr}(\bm\varOmega^2)$, we derive
		\begin{align}
			\label{eq:ktheta}
			k
			=\frac{\operatorname{Tr}^2(\bm\varOmega)}{\operatorname{Tr}(\bm\varOmega^2)},
			\quad
			\theta
			=\frac{\operatorname{Tr}(\bm\varOmega^2)}{\operatorname{Tr}(\bm\varOmega)}.
		\end{align}
		Furthermore, for $k>1$, we can obtain 
		\begin{align}
			\nonumber
			\mathbb E\!\left[\frac1X\right]
			&= \int_{0}^{\infty} \frac{1}{x} f_X(x) dx 
			= \frac{1}{\varGamma(k) \theta^k} \int_{0}^{\infty} x^{(k-1)-1} e^{-\frac{x}{\theta} } dx \\
			\label{eq:EXP}
			&\overset{(a)}{=} \frac{\varGamma(k-1) \theta^{k-1}}{\varGamma(k) \theta^k} 
			\overset{(b)}{=} \frac1{(k-1)\theta}, 
		\end{align}
		where equality $(a)$ is obtained by using $\int_{0}^{\infty} t^{z-1} e^{-\frac{t}{s}} dt = \varGamma(z)s^z$ and equality $(b)$ is obtained by exploiting the recursive properties of the Gamma function $\varGamma(k) = (k-1)\varGamma(k-1)$. Upon substituting $k$ and $\theta$ in \eqref{eq:ktheta} into \eqref{eq:EXP}, we arrive at   
		\begin{align}
			\label{eq:E2}
			\hspace{-0.2em}
			\mathbb{E} \left\{ {\frac{1}{ {\bm{h}}_k^*{\bm{\varPhi }}{{\bm{h}}_k} }} \right\} 
			= \frac{1}{\beta_k} \mathbb E\!\left[\frac1X\right]
			= \frac{\operatorname{Tr}(\bm\varOmega)}
			{ \beta_k \left[ \operatorname{Tr}^2(\bm\varOmega)-\operatorname{Tr}(\bm\varOmega^2) \right] }. 
		\end{align}
		Lastly, substituting \eqref{eq:E1} and \eqref{eq:E2} into \eqref{eq:Exp_gamma_MRC} and exploiting $R_{k}^{\rm{(mrc)}} \ge \log_{2} \left( 1 + \mathbb{E} \left\{ \mathsf{SINR}_{k} \right\} \right)$, we arrive at \eqref{eq:AchievableRate_MRC_Asymptotic}.

		For the specific model \eqref{eq:SignalModel_MatrixForm_SC}, where ${\bm{\varPhi }}$ forms a scaled identity matrix of ${\bm{\varPhi }} = \frac{\varrho \varPhi^{*} \varPhi}{\sigma^2} \bm{Q}^{*} \bm{Q} = \frac{\varrho {{\cos }^2}\varphi ({\Omega_{\ell}})}{\sigma^2} \bm{I}_{M}$, we thus have ${\bm{\varOmega }} \in \{ \bm{I}_{M}, {\bm{\varLambda }} \} \times \frac{\varrho {{\cos }^2}\varphi ({\Omega_{\ell}})}{\sigma^2}$ for both UFC and CFC cases. 
		Additionally, we can obtain ${\rm Tr} ({\bm{\varOmega }}) = M \frac{\varrho {{\cos }^2}\varphi ({\Omega_{\ell}})}{\sigma^2}$ for both cases by exploiting ${\rm{Tr}} \left( \bm{\varLambda} \right) = {\rm{Tr}} \left( \bm{R} \right) = M$ for the Toeplitz matrix $\bm{R}$. We can also derive $\frac{ {\rm{Tr}}\left( {{\bm{\varOmega }}^2} \right) }{ {\rm{Tr}}\left( {\bm{\varOmega }} \right) } = \frac{\varrho {{\cos }^2}\varphi ({\Omega_{\ell}})}{\sigma^2}$ for the UFC case and $\frac{ {\rm{Tr}}\left( {{\bm{\varOmega }}^2} \right) }{ {\rm{Tr}}\left( {\bm{\varOmega }} \right) } = \frac{1}{M} {\rm{Tr}}\left( {{{\bm{\varLambda }}^2}} \right) \frac{\varrho {{\cos }^2}\varphi ({\Omega_{\ell}})}{\sigma^2}$ for the CFC case. 
		Furthermore, we prove in Appendix \ref{Appendix:SINRratio_CFC} that $\frac{1}{M} {\rm{Tr}}\left( {{{\bm{\varLambda }}^2}} \right)$ is upper bounded by $\zeta$ for the CFC case. Based on the above results and upon defining $\mathsf{SNR}_{1,k} \triangleq {\varrho} \frac{ {\cal P}_{s}}{\sigma^2} {\beta_k}$, we arrive at \eqref{eq:AchievableRate_MRC_Asymptotic_1}. 
	\end{IEEEproof}

	\subsubsection{ZF}
	Upon applying the ZF scheme to the RAQ-MIMO of \eqref{eq:SignalModel_MatrixForm_SC}, we reformulate the expectation of the SINR \eqref{eq:Gamma_MRC} as  
	\begin{align}
		\label{eq:Exp_gamma_ZF}
		&\mathbb{E} \left\{ \mathsf{SINR}_{k} \right\} 
		= \frac{ {\cal P}_s }{ \mathbb{E} \left\{ {\bm{c}}_k^*{{\bm{c}}_k} \right\} } 
		= \frac{ {\cal P}_s }{ \mathbb{E} \left\{ {{{\left[ \left( {{\bm{H}}^*} {\bm{\varPhi }} \bm{H} \right)^{-1} \right]}_{kk}}} \right\} }. 
	\end{align}
	Furthermore, we have the following Lemma.

	\begin{lemma}
		\label{lemma4}
		Considering model \eqref{eq:SignalModel_MatrixForm} and both UFC and CFC of \eqref{eq:ChannelModel}, as well as assuming a perfect CSI, we have the EAR of the ZF RAQ-MIMO for the $k$-th user as follows 
		\begin{align}
			\label{eq:AchievableRate_ZF_Asymptotic}
			R_{k}^{\rm{(zf)}} \ge 
			\log_{2} \left( 1 + {\cal P}_s \beta_k \left\{ C - \frac{1}{C}{\rm{Tr}} \left[ \left( {\bm{\varOmega }} \bar{\bm{P}}_{k}^{\perp} \right)^2 \right] \right\} \right), 
		\end{align}
		where $C \triangleq {\rm{Tr}}({\bm{\varOmega }}) - (K - 1)\frac{{{\rm{Tr}}({{\bm{\varOmega }}^2})}}{{{\rm{Tr}}({\bm{\varOmega }})}}$ and $\bar{\bm{P}}_{k}^{\perp} \triangleq \bm I_M - \frac{K-1}{\operatorname{Tr}(\bm\varOmega)} \bm\varOmega $. More particularly, when ${\bm{\varPhi }} = \frac{ \varrho {{\cos }^2}\varphi ({\Omega_{\ell}}) }{ \sigma^2 } \bm{I}_{M}$ forms a scaled identity matrix based on \eqref{eq:SignalModel_MatrixForm_SC}, we can reformulate \eqref{eq:AchievableRate_ZF_Asymptotic} as 
		\begin{equation}
			\label{eq:AchievableRate_ZF_Asymptotic_1}
			\hspace{-0.2em} 
			R_{k}^{\rm{(zf)}} \ge 
			{\log _2} \left[
				1 + \left[ {M - \varepsilon (K - 1)} \right] \mathsf{SNR}_{1,k} {{\cos }^2} \varphi ({\Omega _\ell }) 
			\right], 
		\end{equation}
		where we have $\mathsf{SNR}_{1,k} = {\varrho} \frac{ {\cal P}_{s}}{\sigma^2} {\beta_k}$ and $\varepsilon \in \left\{ 1, \zeta \right\}$. 
	\end{lemma}
	\begin{IEEEproof}
		Based on \eqref{eq:ChannelModel}, we can obtain ${{\bm{H}}^*} {\bm{\varPhi }} \bm{H} = {\bm{D}}^{\frac{1}{2}} {{\bm{G}}^*}$ ${\bm{\varOmega G}} {\bm{D}}^{\frac{1}{2}}$. As a consequence, we can obtain \eqref{eq:Exp_gamma_ZF} as follows
		\begin{align}
			\label{eq:Exp_gamma_ZF_1}
			&\mathbb{E} \left\{ \mathsf{SINR}_{k} \right\} 
			= \frac{ {\cal P}_s \beta_k }{ \mathbb{E} \left\{ \left[ \left( {{\bm{G}}^*} {\bm{\varOmega G}} \right)^{-1} \right]_{kk} \right\} }. 
		\end{align}
		Let us define $\bm{W} = {{\bm{G}}^*} {\bm{\varOmega G}}$ and $\bm{G} = [\bm{g}_{k}, \bm{G}_{-k}]$, where $\bm{g}_{k}$ is the $k$-th column and $\bm{G}_{-k}$ collects the remaining $K-1$ columns. By exploiting the block matrix inversion, we obtain 
		\begin{align}
			\label{eq:Weight}
			\left[ \bm{W}^{-1} \right]_{kk}
			&=
			\frac{1}{\bm{g}_{k}^* \bm{\varOmega}^{\frac{1}{2}} \bm{P}_k^\perp \bm{\varOmega}^{\frac{1}{2}} \bm{g}_{k}},
		\end{align}
		where $\bm{P}_{k}^{\perp}$ is the orthogonal projection matrix onto the null space of $\bm{\varOmega}^{\frac{1}{2}} \mathbf{G}_{-k}$ and it is given by 
		\begin{align}
			\label{eq:proj}
			\bm{P}_{k}^{\perp}
			= \bm{I}_{M} - \bm{\varOmega}^{\frac{1}{2}} \bm{G}_{-k}
			\left( \bm{G}_{-k}^{*} \bm{\varOmega} \bm{G}_{-k} \right)^{-1}
			\bm{G}_{-k}^{*} \bm{\varOmega}^{\frac{1}{2}}.
		\end{align}
		Additionally, let us define $X_{k} = \bm{g}_{k}^* \bm{\varOmega}^{\frac{1}{2}} \bm{P}_k^\perp \bm{\varOmega}^{\frac{1}{2}} \bm{g}_{k}$. By exploiting the moment matching principle, $X_k$ can be approximated by a Gamma distributed random variable $X_{k} \sim \varGamma(k,\theta)$, where
		\begin{align}
			k =
			\frac{\operatorname{Tr}^{2}( \bm{P}_k^\perp \bm{\varOmega} )}
			{\operatorname{Tr}( \bm{P}_k^\perp \bm{\varOmega} \bm{P}_k^\perp \bm{\varOmega} )},
			\quad
			\theta =
			\frac{\operatorname{Tr}( \bm{P}_k^\perp \bm{\varOmega} \bm{P}_k^\perp \bm{\varOmega} )}
			{\operatorname{Tr}(\bm{P}_k^\perp \bm{\varOmega})}. 
		\end{align}
		Therefore, we can furthermore obtain 
		\begin{align}
			\nonumber
			\mathbb{E}\!\left\{
			\left[\bm{W}^{-1}\right]_{kk}
			\right\}
			&= \mathbb{E}\!\left\{\frac{1}{X_k}\right\}
			= \frac{1}{(k-1)\theta} \\
			\label{eq:EXP1}
			&= \frac{\operatorname{Tr}( \bm{P}_k^\perp \bm{\varOmega} )}
			{\operatorname{Tr}^{2}( \bm{P}_k^\perp \bm{\varOmega} ) - \operatorname{Tr}( \bm{P}_k^\perp \bm{\varOmega} \bm{P}_k^\perp \bm{\varOmega} )}.
		\end{align}
		For sufficiently large $M$, we have
		\begin{align}
			\frac1M \bm G_{-k}\bm G_{-k}^H
			&\xrightarrow[]{M \rightarrow \infty}
			\frac{K-1}{M} \bm I_M, \\
			\frac1M\bm G_{-k}^H\bm\varOmega\bm G_{-k}
			&\xrightarrow[]{M \rightarrow \infty}
			\frac1M\operatorname{Tr}(\bm\varOmega)\bm I_M.
		\end{align}
		Therefore, upon substituting these results into \eqref{eq:proj}, we can reformulate \eqref{eq:proj} as follows 
		\begin{align}
			\label{eq:proj_1}
			\bm{P}_{k}^{\perp}
			&\xrightarrow[]{M \rightarrow \infty} 
			\bm I_M -
			\frac{K-1}{\operatorname{Tr}(\bm\varOmega)}
			\bm\varOmega 
			\triangleq \bar{\bm{P}}_{k}^{\perp}.
		\end{align}
		Furthermore, by substituting \eqref{eq:proj_1} into \eqref{eq:EXP1} and finding its reciprocal, we can obtain 
		\begin{align}
			\label{eq:EXP2}
			\frac{1}{
				\mathbb{E}\!\left\{
				\left[\bm{W}^{-1}\right]_{kk}
				\right\}
			}
			&\xrightarrow[]{M \rightarrow \infty} 
			C - \frac{1}{C}{\rm{Tr}} \left[ \left( {\bm{\varOmega }} \bar{\bm{P}}_{k}^{\perp} \right)^2 \right]. 
		\end{align} 
		Upon substituting \eqref{eq:EXP2} into \eqref{eq:Exp_gamma_ZF_1} and exploiting the relationship of $R_{k}^{\rm{(zf)}} \ge \log_{2} \left( 1 + \mathbb{E} \left\{ \mathsf{SINR}_{k} \right\} \right)$, we arrive at \eqref{eq:AchievableRate_ZF_Asymptotic}. 
		
		For the specific model \eqref{eq:SignalModel_MatrixForm_SC}, we obtain ${\bm{\varPhi }} = \frac{ \varrho {{\cos }^2}\varphi ({\Omega_{\ell}}) }{ \sigma^2 } \bm{I}_{M}$ and ${\bm{\varOmega }} \in \{ \bm{I}_{M}, {\bm{\varLambda }} \} \times \frac{ \varrho {{\cos }^2}\varphi ({\Omega_{\ell}}) }{ \sigma^2 }$ for both the UFC and CFC cases. 
		Based on the proof of \eqref{eq:AchievableRate_MRC_Asymptotic_1} in \textbf{Lemma} \ref{lemma1}, we also have ${\rm Tr} ({\bm{\varOmega }}) = M \frac{ \varrho {{\cos }^2}\varphi ({\Omega_{\ell}}) }{ \sigma^2 }$ and $\frac{ {\rm{Tr}}\left( {{\bm{\varOmega }}^2} \right) }{ {\rm{Tr}}\left( {\bm{\varOmega }} \right) } = \{ 1, \frac{1}{M} {\rm{Tr}}\left( {{{\bm{\varLambda }}^2}} \right) \} \times \frac{ \varrho {{\cos }^2}\varphi ({\Omega_{\ell}}) }{ \sigma^2 } \rightarrow \varepsilon \frac{ \varrho {{\cos }^2}\varphi ({\Omega_{\ell}}) }{ \sigma^2 }$, where $\varepsilon \in \left\{ 1, \zeta \right\}$ for both UFC and CFC cases. Based on these results, we furthermore obtain the following result for the UFC case 
		\begin{align}
			\nonumber
			\text{UFC: } 
			&C - \frac{1}{C}{\rm{Tr}} \left[ \left( {\bm{\varOmega }} \bar{\bm{P}}_{k}^{\perp} \right)^2 \right] \\
			\nonumber
			&= \left( {1 - \frac{1}{M}} \right)\left( {M - K + 1} \right) \frac{ \varrho {{\cos }^2}\varphi ({\Omega_{\ell}}) }{ \sigma^2 } \\
			\label{eq:Const1}
			&\overset{(c)}{\approx} \left( {M - K + 1} \right) \frac{ \varrho {{\cos }^2}\varphi ({\Omega_{\ell}}) }{ \sigma^2 }, 
		\end{align}
		where $(c)$ is obtained by exploiting $1 - \frac{1}{M} \approx 1$ for large $M$. 
		By following a similar derivation process as that in Appendix \ref{Appendix:SINRratio_CFC}, we obtain $\frac{1}{M} {\mathop{\rm Tr}\nolimits} \left( {{{\bm{\varLambda }}^3}} \right) \approx \frac{3}{\pi\varpi}\ln M + 1+3\left( \frac{\epsilon}{\pi\varpi}+\frac{1}{2\varpi}-\frac1\pi \right)$ and $\frac1{M^2}{\rm Tr}(\bm\Lambda^4) \approx \mathcal{O}(\frac{\ln M}{M})$ is negligible for large $M$. Therefore, we can obtain the following result for the CFC case 
		\begin{align}
			\nonumber
			&\text{CFC: } 
			C - \frac{1}{C}{\rm{Tr}} \left[ \left( {\bm{\varOmega }} \bar{\bm{P}}_{k}^{\perp} \right)^2 \right] \\
			\nonumber
			&\overset{(d)}{\approx} \left[ {1 - \frac{{M\zeta }}{{{{\left[ {M - \zeta (K - 1)} \right]}^2}}}} \right]\left[ {M - \zeta (K - 1)} \right] \frac{ \varrho {{\cos }^2}\varphi ({\Omega_{\ell}}) }{ \sigma^2 } \\
			\label{eq:Const2}
			&\overset{(e)}{\approx} \left[ {M - \zeta (K - 1)} \right] \frac{ \varrho {{\cos }^2}\varphi ({\Omega_{\ell}}) }{ \sigma^2 },  
		\end{align}
		where $(d)$ and $(e)$ are obtained by neglecting $\frac{ 2(K - 1) }{MC} {\rm Tr}\left( {{\bm{\varLambda }}^3} \right)$ and exploiting ${1 - \frac{{M\zeta }}{{{{\left[ {M - \zeta (K - 1)} \right]}^2}}}} \approx 1$ for large $M$. Upon substituting \eqref{eq:Const1} and \eqref{eq:Const2} into \eqref{eq:AchievableRate_ZF_Asymptotic}, we arrive at \eqref{eq:AchievableRate_ZF_Asymptotic_1}. 
	\end{IEEEproof}

		\vspace{-0.8em}
		\subsection{Homogeneous RAQ-MIMO Array}
		
		For simplicity and easy access, we first consider the homogeneous RAQ-MIMO array case where ${\bm{\varPhi }} = \frac{ \varrho {{\cos }^2}\varphi ({\Omega_{\ell}}) }{ \sigma^2 } \bm{I}_{M}$ forms a scaled identity matrix. This corresponds to the proposed signal model \eqref{eq:SignalModel_MatrixForm_SC} and derived EARs in \eqref{eq:AchievableRate_MRC_Asymptotic_1} and \eqref{eq:AchievableRate_ZF_Asymptotic_1}.

		\subsubsection{EAR of Homogeneous RAQ-MIMO in SQL/PSL Regimes}
		\label{subsec:SQLandPSL}
		The results \eqref{eq:AchievableRate_MRC_Asymptotic_1} and \eqref{eq:AchievableRate_ZF_Asymptotic_1} derived in \textbf{Lemma} \ref{lemma1} - \textbf{Lemma} \ref{lemma4} are functions of $\mathsf{SNR}_{1,k} {{\cos }^2}\varphi ({\Omega_{\ell}})$. Let us furthermore define $\overline{\mathsf{SNR}}_{1,k} \triangleq \mathsf{SNR}_{1,k} {{\cos }^2}\varphi ({\Omega_{\ell}})$. It can be specified in different regimes for the RAQ-MIMO system, namely the SQL and PSL regimes, as firstly discussed in \cite{gong2026rydberg} for RAQ-SISO systems. Specifically, in the SQL regime, the QPN dominates, while other noise sources diminish. This reveals the ultimate fundamental limit of RAQ-MIMO systems with respect to their quantum property. By contrast, the PSL regime, namely where the PSN dominates, reflects the fundamental limit of the optical readout in the photodetection. Furthermore, the SQL is widely employed as a theoretical limit, even though it is challenging to approach by employing a standard optical readout scheme. By contrast, the PSL is achievable for the employed BCOD scheme.

		To characterize the SQL, we denote the parameter set by $\mathsf{S}_{\rm{sensor}}^{(\rm{sql})} = \{ N_{\text{atom}}, T_{2} \}$, where $T_{2}$ is the coherence time of a single Rydberg atom. Likewise, $\mathsf{S}_{\rm{sensor}}^{(\rm{psl})} = \{ \mathcal{P}_{0}, N_{\text{atom}}, \chi (\Omega_{\ell}),$ $\varphi (\Omega_{\ell}), \kappa (\Omega_{\ell}) \}$ reflects the parameter set for characterizing the PSL. Let us define $\eta_{\rm coh} \triangleq C_1 N_{\text{atom}} T_{2}$ as the \emph{collective atomic quantum coherence factor} for all sensors, where $C_1 \triangleq \frac{ 2 Z_0 \mu_{34}^{2} }{ \hbar^{2} }$. It measures the number of atoms $N_{\text{atom}}$ effectively participating in the coherent quantum transduction weighted by their coherence time $T_2$. This factor characterizes the genuinely quantum-limited SQL regime. Let us further define $\eta_{\rm opt} \triangleq$ $C_2 {\kappa^2}(\Omega_{\ell}) \exp (- \frac{N_{\text{atom}} {\bar \chi}}{A_p} )$ as the \emph{collective atomic quantum-optical transduction factor} for all sensors, where $C_2 \triangleq \frac{4 \alpha Z_0}{ q }$, ${\bar \chi} \triangleq k_p \mathscr{I} \left\{ \chi (\Omega_{\ell}) \right\}$. It characterizes the efficiency of converting user signals into detectable variations of the probe beam in the PSL regime. Then, we present the following Theorem. 
		\begin{theorem}
			\label{theorem1}
			Considering both the UFC and CFC scenarios of \eqref{eq:ChannelModel} and assuming a perfect CSI, we have the EAR of the homogeneous RAQ-MIMO for the $k$-th user formulated as  
			\begin{align}
				\label{eq:AchievableRate_MRC_Asymptotic3}
				\hspace{-0.6em}
				R_{k} &\ge {\log _2} \left[ 
				1 + \mathcal{G}^{\rm{(mrc/zf)}} ( \underbrace{\mathsf{S}_{\rm{sensor}}^{(\rm{sql/psl})}, M}_{\rm{Rx}}, \underbrace{\beta_k}_{\rm{Ch}}, \underbrace{ {\cal P}_{{s}}, K }_{\rm{Tx}} ) 
				\right], 
			\end{align}
			where $\mathcal{G}^{\rm{(mrc)}}$ and $\mathcal{G}^{\rm{(zf)}}$ are functions of the parameter set $\{ \mathsf{S}_{\rm{sensor}}^{(\rm{sql/psl})}, M, \beta_k, {\cal P}_{{s}}, K \}$. Specifically, we have 
			\begin{align}
				\label{eq:mrcfunc}
				\mathcal{G}^{\rm{(mrc)}}& ( \mathsf{S}_{\rm{sensor}}^{(\rm{sql/psl})}, M, \beta_k, {\cal P}_{{s}}, K ) \\
				\nonumber
				&= \left\{ 
				\begin{aligned}
					&\frac{  M - \varepsilon }{ \left( \frac{B}{ {\cal P}_{s} \beta_{k} } \right) \left( \frac{ 1 }{ {\eta_{\rm coh}} } \right) + \frac{\varepsilon}{ \beta_k } \sum_{i = 1, \ne k}^K {{\beta _i}} }, & \rm{SQL}, \\
					&\frac{  M - \varepsilon }{ \left( \frac{B}{ {\cal P}_{s} \beta_{k} } \right) \left( \frac{1}{ {\overline {\cal P}_0 } {\eta_{\rm opt}} } \right) + \frac{\varepsilon}{ \beta_k } \sum\nolimits_{i = 1, \ne k}^K {\beta _i} }, & \rm{PSL},
				\end{aligned} \right. \\
				\nonumber \\
				\label{eq:zffunc}
				\mathcal{G}^{\rm{(zf)}}& ( \mathsf{S}_{\rm{sensor}}^{(\rm{sql/psl})}, M, \beta_k, {\cal P}_{{s}}, K ) \\
				\nonumber
				&= \left\{ 
				\begin{aligned}
					&{\eta_{\rm coh}} \left( \tfrac{ {\cal P}_{s} \beta_{k} }{B} \right) \left[ M - \varepsilon (K-1) \right], & \rm{SQL}, \\
					&{\overline {\cal P}_0 } {\eta_{\rm opt}} \left( \tfrac{ {\cal P}_{s} \beta_{k} }{B} \right) \left[ M - \varepsilon (K-1) \right], & \rm{PSL}, 
				\end{aligned} \right. 
			\end{align}
			where $\varepsilon \in \{ 1, \zeta \}$ for the UFC and CFC scenarios, respectively, and we have defined ${\overline {\cal P}_0 } \triangleq {\cal P}_0 {{\cos }^2}\varphi ({\Omega_{\ell}})$. 
		\end{theorem}
		\begin{IEEEproof}
			In the SQL regime, we have $w_{m} \sim \mathcal{CN} (0, \sigma^2)$, where $\sigma^2$ becomes QPN-dominated, namely we have $\sigma^2 \approx \frac{ N_{\rm{QPN}} }{2} = \frac{ \varrho c \epsilon_0 \cos^{2} \varphi (\Omega_{\ell}) }{2} \big( \frac{U_{\text{SQL}}}{\sqrt{\text{Hz}}} \big)^{2} B$ and $\frac{U_{\text{SQL}}}{\sqrt{\text{Hz}}} = \frac{\hbar}{\mu_{34} \sqrt{N_{\text{atom}} T_2}}$  \cite{gong2026rydberg}. Based on this result, we obtain $\overline{\mathsf{SNR}}_{1,k} = \mathsf{SNR}_{1,k} {{\cos }^2}\varphi ({\Omega_{\ell}}) = \frac{ {\varrho} {{\cos }^2}\varphi ({\Omega_{\ell}}) {\cal P}_{s} \beta_{k} }{\sigma^2}$ as follows 
			\begin{align}
				\label{eq:SNR1_SQL}
				\overline{\mathsf{SNR}}_{1,k}^{(\rm{sql})} 
				= \tfrac{ 2 {\cal P}_{s} {\beta_k} }{ c \epsilon_0 B} \Big/ \left(  \tfrac{U_{\text{SQL}}}{\sqrt{\text{Hz}}} \right)^{2} 
				= \overbrace{C_1 N_{\text{atom}} T_2}^{= \eta_{\rm coh}} \left( \tfrac{ {\cal P}_{s} \beta_{k} }{B} \right). 
			\end{align}
			In the PSL regime, the PSN noise dominates, where $\sigma^2$ reduces and becomes  $\sigma^2 \approx \frac{ N_{\rm{PSN}} }{2} = qB \alpha [ \mathcal{P}_{\text{lob}} + \mathcal{P} (\Omega_{\ell}) ] G_{\text{pd}}$ \cite{gong2026rydberg}. As a consequence, we can obtain $\overline{\mathsf{SNR}}_{1,k}$ as follows  
			\begin{align}
				\nonumber
				\overline{\mathsf{SNR}}_{1,k}^{(\rm{psl})} 
				&= C_2 \frac{ \mathcal{P}_{\text{lob}} \mathcal{P} (\Omega_{\ell}) \kappa^{2}({\Omega_{\ell}}) \cos^2 {\varphi ({\Omega_{\ell}})} }{ \mathcal{P}_{\text{lob}} + \mathcal{P} (\Omega_{\ell}) } \left( \tfrac{ {\cal P}_{s} \beta_{k} }{B} \right) \\
				\nonumber
				&\overset{(d)}{\approx} C_2 \mathcal{P} (\Omega_{\ell}) \kappa^{2}({\Omega_{\ell}}) \cos^2 {\varphi ({\Omega_{\ell}})} \left( \tfrac{ {\cal P}_{s} \beta_{k} }{B} \right) \\
				\label{eq:SNR1_PSL} 
				&\overset{(e)}{=} {\overline {\cal P}_0 } \overbrace{ {C_2} \kappa^{2}({\Omega_{\ell}}) \exp (- \tfrac{N_{\text{atom}} {\bar \chi}}{A_p} ) }^{= \eta_{\rm opt}} \left( \tfrac{ {\cal P}_{s} \beta_{k} }{B} \right),
			\end{align}
			where the approximation $(d)$ is obtained by exploiting $\mathcal{P}_{\text{lob}} \gg \mathcal{P} (\Omega_{\ell})$; the equality $(e)$ is derived by exploiting the relationship of ${\overline {\cal P}_0 } \triangleq {\cal P}_0 {{\cos }^2}\varphi ({\Omega_{\ell}})$ and $\mathcal{P} (\Omega_{\ell}) = {\cal P}_0 \exp (- \frac{N_{\text{atom}} {\bar \chi}}{A_p} )$ based on \eqref{eq:OutputProbePower}. 
			Furthermore, upon substituting \eqref{eq:SNR1_SQL} and \eqref{eq:SNR1_PSL} into the above results of \textbf{Lemma} \ref{lemma1} - \textbf{Lemma} \ref{lemma4}, we can obtain \eqref{eq:AchievableRate_MRC_Asymptotic3}, which completes the proof. 
		\end{IEEEproof}

		\textit{\textbf{Remark 3}: As seen in \textbf{Theorem} \ref{theorem1}, the EAR variation behaviour of the homogeneous RAQ-MIMO receiver is determined by the macroscopic array dimensions $M - \varepsilon$ and $M - \varepsilon (K-1)$, as well as the microscopic quantum-related and optical-related parameters $\mathsf{S}_{\mathrm{sensor}}^{(\mathrm{sql})}$ and $\mathsf{S}_{\mathrm{sensor}}^{(\mathrm{psl})}$. These behaviors fundamentally differ from those of classical MIMO systems, which are primarily governed by coherently combining the different antennas' signals.}
		
		\textit{\textbf{Remark 4}: The EAR of both the MRC and ZF RAQ-MIMO receivers in their PSL regime can be maximized by optimally configuring the $m$-th LOB source via ${\phi_{m}^{\text{(lob)}}} = {\phi_{p,m}}({\Omega_{\ell,m}}) - \arccos \left[ \frac{ \varsigma_1 {\mathscr I}\{ \chi_{m}^{\prime} ({\Omega_{\ell,m}})\} }{ \mathcal{C}_m (\Omega_{\ell,m}) } \right]$. Consequently, we can obtain $\cos^2 {\varphi_{m}({\Omega_{\ell}})} = 1$ and ${\overline {\cal P}_0 } = {\cal P}_0$, as well as the maximal EARs by replacing ${\overline {\cal P}_0 }$ with ${\cal P}_0$ in \eqref{eq:mrcfunc} and \eqref{eq:zffunc}, respectively.}

		\begin{table*}[t]
			\centering
			\caption{\textsc{Comparisons of the Homogeneous RAQ-MIMO in the SQL and PSL regimes.}}
			\renewcommand{\arraystretch}{0.6}
			\label{tab:comparison_sql_psl}
			\resizebox{\linewidth}{!}{
			\begin{tabular}{@{} 
					p{0.9cm}
					p{7.2cm}
					p{8.2cm}
					@{}}
				\toprule
				\textbf{Aspect} 
				& \textbf{\makecell[c]{SQL Regime}} 
				& \textbf{\makecell[c]{PSL Regime}} \\
				\midrule
				
				\vspace{0.1em} \textbf{Key factor} 
				& QPN-limited performance jointly governed by the effective number of Rydberg atoms within a single sensor and the quantum coherence time, namely $N_{\text{atom}}$ and $T_{2}$. 
				& PSN-limited performance jointly governed by the collective atomic enhancement and optical-depth-dependent attenuation in the atomic quantum–optical transduction, i.e., $\big[ \frac{ N_{\text{atom}} {\cal C} (\Omega_{\ell}) }{A_p} \big]^{2}$ and $\exp \big(- \frac{N_{\text{atom}} {\bar \chi}}{A_p} \big)$. \\
				\midrule
				
				\vspace{0.2em} \textbf{EAR scaling law} 
				& \textbf{ZF}: EAR grows logarithmically without bound as the ensemble number $N_{\text{atom}}$ or the coherence time $T_2$ of Rydberg atoms increase, namely $R_k^{\text{(zf,ufc/cfc,sql)}} \propto \log_2 (N_{\text{atom}} T_2)$. 
				
				\vspace{0.4em}
				\textbf{MRC}: EAR saturates as $N_{\text{atom}} \rightarrow \infty$ or $T_2 \rightarrow \infty$. 
				& \textbf{ZF} \& \textbf{MRC}: EAR exhibits non-monotonic trade-off between the collective atomic enhancement and optical-depth-dependent attenuation. It is maximally determined by the optimal atom number $N_{\text{atom}}^{\star} = \frac{2 A_p}{{\bar \chi}}$ or at an optimal cell's length of $L^{\star} = \frac{2}{{N_0} {\bar \chi}}$. \\ 
				\midrule
				
				\vspace{0.2em} \textbf{Power scaling law} 
				& \textbf{ZF} \& \textbf{MRC}: If $M = C_{\text{sql}} N_{\text{atom}} T_2$, the users' transmit power ${\cal P}_{s}$ can be scaled down quadratically with $N_{\text{atom}} T_2$ as $N_{\text{atom}} \rightarrow \infty$, namely ${\cal P}_{s} = {\cal E}/(N_{\text{atom}} T_2)^{2}$ for a certain fixed energy ${\cal E}$, while attaining $R_k^{\text{(mrc/zf)}} = {\log_2} \big[ 1 + C_{\text{sql}} {C_1} \big( \frac{{\cal E} \beta_k}{B} \big) \big]$.
				& \textbf{ZF} \& \textbf{MRC}: If $M = C_{\text{psl}} \exp \big( \frac{N_{\text{atom}} {\bar \chi}}{A_p} \big)$, the users' transmit power ${\cal P}_{s}$ can be scaled down quadratically with $N_{\text{atom}} {\cal C} (\Omega_{\ell}) / A_p$ as $N_{\text{atom}} \rightarrow \infty$, namely ${\cal P}_{s} = {\cal E}/$ $[ N_{\text{atom}} {\cal C} (\Omega_{\ell}) / A_p ]^{2}$ for a certain fixed energy ${\cal E}$, while attaining $R_k^{\text{(mrc/zf)}} = \log_2 \big[ 1 + C_{\text{psl}} {C_2} \overline{{\cal P}_0} \big( \frac{{\cal E} \beta_k}{B} \big) \big]$.  \\
				\bottomrule
			\end{tabular}}
			\vspace{-0.3em}
		\end{table*}

		\vspace{0.6em}
		\subsubsection{Scaling Behavior of Homogeneous RAQ-MIMO}
		\label{subsubsec:Insights}
		Based on \eqref{eq:AchievableRate_MRC_Asymptotic3}--\eqref{eq:zffunc} obtained in \textbf{Theorem} \ref{theorem1}, we present the following deeper insights, as highlighted in TABLE \ref{tab:comparison_sql_psl}. 
		
		\textbf{EAR Scaling Law (SQL regime)}: 
		In this regime, we have the following theorem for the EAR scaling behavior.
		\begin{theorem}
			\label{theorem2}
			For the homogeneous RAQ-MIMO having independent Rydberg atomic ensemble, as the effective number or coherence time of the atomic ensemble increase (namely $\eta_{\rm coh}$ increases), the EAR of the MRC RAQ-MIMO receiver in the SQL regime becomes saturated, formulated as 
			\begin{equation}
				\label{eq:EAR_MRC_ScalingLaw_SQL}
				\hspace{-0.2em} R_{k}^{\rm{(mrc, ufc/cfc, sql)}}
				\xrightarrow[\text{or}\; T_2 \rightarrow \infty]{N_{\text{atom}} \rightarrow \infty} 
				{\log _2} \left[ 
				\begin{split}
					1 + \frac{ M - \varepsilon }{ \frac{\varepsilon}{ \beta_{k} } \sum_{i = 1, \ne k}^K {{\beta _i}} }
				\end{split} 
				\right]. 
			\end{equation} 
			By contrast, the EAR of the homogeneous ZF RAQ-MIMO receiver in the SQL regime scales logarithmically without bound as the effective number of or the coherence time of the atomic ensemble increase. Explicitly, we have 
			\begin{align}
				\label{eq:EAR_ZF_ScalingLaw_SQL}
				\hspace{-0.2em} 
				R_{k}^{\rm{(zf, ufc/cfc, sql)}} \propto {\log _2} \big( N_{\text{atom}} T_2 \big) \xrightarrow[]{N_{\text{atom}} \rightarrow \infty \;\text{or}\; T_2 \rightarrow \infty} \infty. 
			\end{align} 
		\end{theorem}
		\begin{IEEEproof}
			Based on the result of the SQL regime in \eqref{eq:mrcfunc}, we can directly obtain \eqref{eq:EAR_MRC_ScalingLaw_SQL} as $N_{\text{atom}} \rightarrow \infty$ or $T_2 \rightarrow \infty$ (namely $\eta_{\rm coh} \rightarrow \infty$). Furthermore, by exploiting ${\eta_{\rm coh}} \left( \frac{ {\cal P}_{s} \beta_{k} }{B} \right) \left[ M - \varepsilon (K-1) \right] \gg 1$ in \eqref{eq:zffunc} and the high-SINR approximation $\log_{2} (1+x) \approx \log_{2} x$, as $N_{\text{atom}} \rightarrow \infty$ or $T_2 \rightarrow \infty$ (namely $\eta_{\rm coh} \rightarrow \infty$), we arrive at \eqref{eq:EAR_ZF_ScalingLaw_SQL}. 
		\end{IEEEproof}

		\textbf{EAR Scaling Law (PSL regime)}: 
		As observed from the expression $\eta_{\rm opt} \triangleq C_2 {\kappa^2}(\Omega_{\ell}) \exp (- \frac{N_{\text{atom}} {\bar \chi}}{A_p} )$ that $\eta_{\rm opt}$ reflects the interplay between \textbf{(i)}~the enhancement term ${\kappa^2}({\Omega_{\ell} }) = \left[ \frac{ N_{\text{atom}} {\cal C} (\Omega_{\ell}) }{A_p} \right]^{2}$, reflecting the \emph{collective atomic quantum responsivity} induced by a single sensor, and \textbf{(ii)}~the exponential attenuation term $\exp (- \frac{N_{\text{atom}} {\bar \chi}}{A_p} )$ arising from the optical-depth-dependent attenuation during the probe beam's propagation within a vapor cell. Therefore, the EAR scaling law in the PSL regime encapsulates a nonlinear trade-off between the collective enhancement and the propagation dissipative regimes, making the dependence on $N_{\text{atom}}$ non-monotonic, as formulated in the following Theorem. 
		\begin{theorem}
			\label{theorem3}
			For the homogeneous RAQ-MIMO having an independent Rydberg atomic ensemble, the EAR of the MRC and ZF schemes in the PSL regime can be maximized at an atomic number of $N_{\text{atom}}^{\star} = \frac{2 A_p}{{\bar \chi}}$ or at a cell's length of $L^{\star} = \frac{2}{{N_0} {\bar \chi}}$ or at a laser's cross-sectional area of $A_{p}^{\star} = \frac{1}{2} N_{\text{atom}} {\bar \chi}$, yielding 
			\begin{equation}
				\label{eq:EAR_MRC_ZF_ScalingLaw_PSL} 
				\begin{aligned}
					&R_{k}^{\rm{(mrc/zf,ufc/cfc,psl)} \star} 
					\xrightarrow[]{N_{\text{atom}} \rightarrow N_{\text{atom}}^{\star} \;\text{or}\; L \rightarrow L^{\star} \;\text{or}\; A_{p} \to A_{p}^{\star} } \\ 
					&\left\{ 
					\begin{aligned}
						&{\log_2} \left[ 
						\begin{split} 
							1 + \frac{ M - \varepsilon }{ \frac{1}{ C_3 {\overline {\cal P}_0 } } \left[ \frac{ {\bar \chi} }{ {\cal C}({\Omega_{\ell}}) } \right]^{2} \left( \frac{B}{ {\cal P}_{s} \beta_{k} } \right)  + \dfrac{\varepsilon}{ \beta_k } \sum\limits_{i = 1 \atop i \ne k}^K {\beta _i} } 
						\end{split} 
						\right], & \rm{MRC}, \\
						&\log_2 \left[ 1 + C_3 {\overline {\cal P}_0 } \left[ \tfrac{{\cal C}({\Omega_{\ell}})}{{\bar \chi}} \right]^{2} \tfrac{ {\cal P}_{s} \beta_{k} }{B} \left[ M - \varepsilon (K-1) \right] \right], & \rm{ZF},\;\;\;  
					\end{aligned} \right. 
				\end{aligned}
			\end{equation}
			where we have $C_3 \triangleq \frac{{4{C_2}}}{{\exp (2)}} \approx 0.54 {C_2}$. 
		\end{theorem}
		\begin{IEEEproof}
			By deriving the solution of $\frac{ d \eta_{\rm opt} }{ d N_{\text{atom}} } = 0$, we obtain the optimal configuration of $N_{\text{atom}}^{\star} = \frac{2 A_p}{{\bar \chi}}$ or alternatively $A_{p}^{\star} = \frac{1}{2} N_{\text{atom}} {\bar \chi}$. By exploiting furthermore the relationship of $L = \frac{ N_{\text{atom}} }{ N_0 A_p }$, we obtain $L^{\star} = \frac{2}{{N_0} {\bar \chi}}$. Upon substituting $N_{\text{atom}} = N_{\text{atom}}^{\star}$ or $L = L^{\star}$ or $A_{p} \to A_{p}^{\star}$ into the PSL result of \eqref{eq:mrcfunc} and \eqref{eq:zffunc}, respectively, we arrive at \eqref{eq:EAR_MRC_ZF_ScalingLaw_PSL}. 
		\end{IEEEproof}

		\vspace{0.6em}
		\textbf{Power Scaling Law}: 
		Due to the distinctive characteristics of the RAQ-MIMO receiver, the power scaling law is different from that of classical M-MIMO systems. Specifically, the power scaling law is unveiled in the following Theorem. 
		\begin{theorem}
			\label{theorem4}
			For non-negative constants $C_{\rm{sql}}$, $C_{\rm{psl}}$, 
			we assume $M = C_{\rm{sql}} N_{\rm{atom}} T_2$ and $M = C_{\rm{psl}} \exp \left( \frac{N_{\rm{atom}} {\bar \chi}}{A_p} \right)$ for SQL and PSL, respectively. As $N_{\rm{atom}}$ becomes large, the transmit power ${\cal P}_{s}$ can be scaled down quadratically with $N_{\rm{atom}} \tau$, where $\tau \in \{ T_2, \frac{ {\cal C} (\Omega_{\ell}) }{A_p} \}$ for SQL and PSL, respectively. Explicitly, we have ${\cal P}_{s} = \frac{{\cal E}}{ (N_{\rm{atom}} \tau)^{2} }$ for a certain fixed energy ${\cal E}$, yet the sum-rate grows linearly with $K$, where the EAR of the $k$-th user for both the MRC and ZF receivers is formulated as 
			\begin{equation}
				\label{eq:PowerScalingLaw} 
				\begin{aligned}
					&R_{k}^{\rm{(mrc/zf, ufc/cfc, sql/psl)}} 
					\xrightarrow[ ]{ {\cal P}_{s} = \frac{{\cal E}}{ (N_{\rm{atom}} \tau)^{2} } \;\text{and}\; N_{\rm{atom}} \rightarrow \infty }  \\ 
					&\qquad \left\{ 
					\begin{aligned}
						&{\log_2} \left[ 1 + C_{\rm{sql}} {C_1} \left( \frac{{\cal E} \beta_k}{B} \right) \right], & \rm{SQL}, \\
						&\log_2 \left[ 1 + C_{\rm{psl}} {C_2} \overline{{\cal P}_0} \left( \frac{{\cal E} \beta_k}{B} \right) \right], & \rm{PSL}. 
					\end{aligned} \right. 
				\end{aligned}
			\end{equation}
		\end{theorem}
		\begin{IEEEproof}
			Upon substituting ${\cal P}_{s} = \frac{{\cal E}}{(N_{\rm{atom}} \tau)^{2}}$ into both the SQL and PSL results of \eqref{eq:mrcfunc} and \eqref{eq:zffunc}, respectively, we obtain 
			\begin{align}
				\nonumber
				&\mathcal{G}^{\rm{(mrc)}} ( \mathsf{S}_{\rm{sensor}}^{(\rm{sql/psl})}, M, \beta_k, {\cal P}_{{s}}, K ) \\
				\nonumber
				&= \left\{ 
				\begin{aligned}
					&\frac{ 1 }{ \left( \frac{B}{ {\cal E} \beta_{k} } \right) \left( \frac{ 1 }{ C_1 } \right) + \frac{\varepsilon}{ \beta_k N_{\text{atom}} T_2 } \sum_{i = 1, \ne k}^K {{\beta _i}} } 
					\times \frac{ M - \varepsilon }{ N_{\text{atom}} T_2 }, & \rm{SQL}, \\
					&\frac{ (M - \varepsilon) \exp \left( - \frac{N_{\text{atom}} {\bar \chi}}{A_p} \right) }{ \left( \frac{B}{ {\cal E} \beta_{k} } \right) \left( \frac{1}{ C_2 {\overline {\cal P}_0 } } \right) + \frac{\varepsilon}{ \beta_k } \exp(- \frac{N_{\text{atom}} {\bar \chi}}{A_p} ) \sum_{i = 1, \ne k}^K {\beta _i} }, & \rm{PSL},
				\end{aligned} \right. \\ 
				\nonumber
				&\mathcal{G}^{\rm{(zf)}} ( \mathsf{S}_{\rm{sensor}}^{(\rm{sql/psl})}, M, \beta_k, {\cal P}_{{s}}, K ) \\
				\nonumber
				&= \left\{ 
				\begin{aligned}
					&C_1 \left( \frac{ {\cal E} \beta_{k} }{B} \right) 
					\times \frac{ M - \varepsilon (K-1) }{ N_{\text{atom}} T_2 }, & \rm{SQL}, \\
					&{C_2} {\overline {\cal P}_0 } \left( \frac{ {\cal E} \beta_{k} }{B} \right) 
					\left[ M - \varepsilon (K-1) \right] \exp ( - \frac{N_{\text{atom}} {\bar \chi}}{A_p} ), & \rm{PSL}. 
				\end{aligned} \right. 
			\end{align}
			By further exploiting the relationship of $M = C_{\text{sql}} N_{\text{atom}} T_2$ for both $\mathcal{G}^{\text{(mrc)}}$ and $\mathcal{G}^{\text{(zf)}}$ in the SQL regime, we obtain $\frac{  M - \varepsilon }{ N_{\text{atom}} T_2} \xrightarrow[ ]{ N_{\text{atom}} \rightarrow \infty } C_{\text{sql}}$ and $\frac{ M - \varepsilon (K-1) }{ N_{\text{atom}} T_2} \xrightarrow[ ]{ N_{\text{atom}} \rightarrow \infty } C_{\text{sql}}$, which produces the SQL result of \eqref{eq:PowerScalingLaw} for both the MRC and ZF schemes with the aid of $\frac{\varepsilon}{ \beta_k N_{\text{atom}} T_2 } \sum_{i = 1, \ne k}^K {{\beta _i}} \xrightarrow[ ]{ N_{\text{atom}} \rightarrow \infty } 0$. Likewise, by using the relationship of $M = C_{\text{psl}} \exp \big( \frac{N_{\text{atom}} {\bar \chi}}{A_p} \big)$ for $\mathcal{G}^{\text{(mrc)}}$ and $\mathcal{G}^{\text{(zf)}}$ in the PSL regime, we arrive at $(M - \varepsilon)$ $\exp \big( - \frac{N_{\text{atom}} {\bar \chi}}{A_p} \big) \xrightarrow[ ]{ N_{\text{atom}} \rightarrow \infty } C_{\text{psl}}$ and $\left[ M - \varepsilon (K-1) \right] \exp \big( - \frac{N_{\text{atom}} {\bar \chi}}{A_p} \big)$ $\xrightarrow[ ]{ N_{\text{atom}} \rightarrow \infty } C_{\text{psl}}$, hence yielding the PSL result of \eqref{eq:PowerScalingLaw} for both the MRC and ZF schemes with the help of $\frac{\varepsilon}{ \beta_k } \exp \big(- \frac{N_{\text{atom}} {\bar \chi}}{A_p} \big)$ $\sum_{i = 1, \ne k}^K {\beta _i} \xrightarrow[ ]{ N_{\text{atom}} \rightarrow \infty } 0$. 
		\end{IEEEproof}
		
		\textit{\textbf{Remark 5}: In \textbf{Theorem} \ref{theorem4}, $N_{\text{atom}} \tau \in \{ N_{\text{atom}} T_2, N_{\text{atom}} \frac{ {\cal C} (\Omega_{\ell}) }{A_p} \}$ reflects a collective effect of the participating atoms weighted by the coherence time $T_2$ in the SQL regime or by the atomic quantum responsivity per atom $\frac{ {\cal C} (\Omega_{\ell}) }{A_p}$ in the PSL regime. The combined effect of $N_{\text{atom}}$ and $\tau$ quantifies the enhancement aspect of the RAQ-MIMO. To enhance the intuitiveness, the value of $N_{\text{atom}} \tau$ is approximately $22$ and $8$ in the SQL and PSL regimes, respectively, where $N_{\text{atom}} \approx 1.11 \times 10^8$, $T_2 = 2 \times 10^{-7}$, and $\frac{ {\cal C} (\Omega_{\ell}) }{A_p} \approx 1.265 \times 10^{-7}$ based on parameters configured in TABLE \ref{tab:parameters}.}  
		
		\textit{\textbf{Remark 6}: From \textbf{Theorem} \ref{theorem4}, the quadratic power scaling law of ${\cal P}_{s} = \frac{{\cal E}}{ (N_{\text{atom}} \tau)^{2} }$ is encountered by increasing $N_{\text{atom}}$ and $M$ simultaneously, where $M$ scales linearly with $N_{\text{atom}} \tau$ in the SQL regime, while scales exponentially with $\frac{N_{\text{atom}} {\bar \chi}}{A_p}$ in the PSL regime ($\frac{N_{\text{atom}} {\bar \chi}}{A_p} \approx 0.46$ based on the parameters configured in TABLE \ref{tab:parameters}). The exponential escalation of $M$ in the latter case is exploited for canceling out the exponential attenuation $\exp \big(- \frac{N_{\text{atom}} {\bar \chi}}{A_p} \big)$ arising from the optical-depth-dependent attenuation during the probe beam's propagation within a vapor cell.}

		\vspace{-0.8em}
		\subsection{Inhomogeneous RAQ-MIMO Array}

		\subsubsection{EAR reformulation}
		Let us define $\gamma_m \triangleq \frac{\varrho_m|\varPhi_m|^2}{\sigma_m^2}$ as the $m$-th sensor gain normalized by noise. We can reformulate \eqref{eq:PHi} as ${\bm{\varPhi }} = {\rm{diag}} \left\{ \gamma_1, \cdots, \gamma_M \right\}$. Based on ${\bm{\varOmega }} = {\bm{\varPhi }}$ for the UFC case and ${\bm{\varOmega }} = {{\bm{\varLambda }}^{\frac{1}{2}}}{{\bm{U}}^*}{\bm{\varPhi U}}{{\bm{\varLambda }}^{\frac{1}{2}}}$ for the CFC case shown in \textbf{Lemma} \ref{lemma1}, we further define 
		\begin{align}
			\label{eq:def}
			\overline{\gamma} \triangleq \frac{1}{M}\text{Tr}(\bm{\varOmega}), 
			\quad
			\overline{\Delta} \triangleq M\frac{\text{Tr}(\bm{\varOmega}^2)}{\text{Tr}^2(\bm{\varOmega})}.
		\end{align}
		Specifically, $\overline{\gamma}$ represents the array-average gain normalized by noise for the UFC case, while it represents the eigenmode-weighted array-average gain normalized by noise for the CFC case. Additionally, $\overline{\Delta} \in [1, M]$ is an indicating factor reflecting the array inhomogeneity. More particularly, the array is extremely inhomogeneous with few dominant sensors when $\overline{\Delta} \to M$. By contrast, the array is homogeneous when $\overline{\Delta} \to 1$.

		\textbf{MRC}: 
		By furthermore defining $\mathsf{SNR}_{1,k}^{\rm ave} = {\cal P}_s \beta_k \overline{\gamma}$ as the received SNR corresponding to the averaged RAQ-MIMO receiver sensors and the $k$-th user, without experiencing any small-scale fading and any IUI, then \eqref{eq:AchievableRate_MRC_Asymptotic} is reformulated as follows by substituting \eqref{eq:def} and $\mathsf{SNR}_{1,k}^{\rm ave}$ into \eqref{eq:AchievableRate_MRC_Asymptotic}
		\begin{equation}
			\label{eq:R_MRC}
			R_k^{(\text{mrc})} \ge \log_2 \left[ 1 + \frac{(M - \overline{\Delta}) \mathsf{SNR}_{1,k}^{\rm ave} }{1 + \overline{\Delta} \frac{\mathsf{SNR}_{1,k}^{\rm ave}}{\beta_k}\sum_{i \ne k} \beta_i} \right].
		\end{equation}

		\textbf{ZF}: 
		Based on \eqref{eq:def}, we can reformulate $C$ and the asymptotic value of $\frac{1}{C}{\rm{Tr}} \left[ \left( {\bm{\varOmega }} \bar{\bm{P}}_{k}^{\perp} \right)^2 \right]$ in \eqref{eq:AchievableRate_ZF_Asymptotic} as follows 
		\begin{align}
			\label{eq:C}
			\hspace{-0.3em}
			&C = \bar{\gamma} \left[ M - \overline{\Delta} (K-1) \right], \\
			\label{eq:CC}
			&\frac{1}{C}{\rm{Tr}} \left[ \left( {\bm{\varOmega }} \bar{\bm{P}}_{k}^{\perp} \right)^2 \right] \xrightarrow[ ]{ M \rightarrow \infty } \mathcal{O}(1). 
		\end{align}
		Therefore, \eqref{eq:CC} can be neglected compared to \eqref{eq:C} and we can reformulate \eqref{eq:AchievableRate_ZF_Asymptotic} in its asymptotic form as 
		\begin{align}
			\label{eq:R_ZF}
			R_{k}^{\rm{(zf)}} \ge 
			\log_{2} \left\{ 1 + \left[ M - \overline{\Delta} (K-1) \right] \mathsf{SNR}_{1,k}^{\rm ave} \right\}.  
		\end{align}

		\subsubsection{EAR of Inhomogeneous RAQ-MIMO in SQL and PSL}
		In these regimes, we can furthermore obtain $\overline{\gamma}$ as follows 
		\begin{align}
			\label{eq:overline_SQL_gamma}
			&\text{SQL: } \overline{\gamma} = \frac{1}{B} \overline{\gamma}_{\rm{coh}}^{\rm{ave}}, \;\;\;\;\;\; 
			 \overline{\gamma}_{\rm{coh}}^{\rm{ave}} = \frac{1}{M} \sum_{m=1}^M \overline{\lambda}_m \eta_{{\rm coh},m}, \\
			 \label{eq:overline_PSL_gamma}
			&\text{PSL: } \overline{\gamma} = \frac{ \overline{{\cal P}_0^{\rm ref}} }{B} \overline{\gamma}_{\rm{opt}}^{\rm{ave}}, \;\;\; 
			 \overline{\gamma}_{\rm{opt}}^{\rm{ave}} = \frac{1}{M} \sum_{m=1}^M \overline{\lambda}_m \upsilon_m \eta_{{\rm opt},m}, 
		\end{align}
		where $\overline{\gamma}_{\rm{coh}}^{\rm{ave}}$ represents the average eigenmode-weighted collective quantum coherence factor across all $M$ sensors; $\overline{\gamma}_{\rm{opt}}^{\rm{ave}}$ represents the average eigenmode-weighted collective quantum-optical transduction factor across all $M$ sensors; $\overline{{\cal P}_0^{\rm ref}}$ is a predefined reference power; ${\overline {\cal P}_{0,m} } \triangleq {\cal P}_{0,m} {{\cos }^2}\varphi ({\Omega_{\ell,m}})$ and $\upsilon_m = \frac{ {\overline {\cal P}_{0,m} } }{ \overline{{\cal P}_0^{\rm ref}} }$ represents the power and normalized factor related to the $m$-th sensor; $\overline{\lambda}_m$ represents the eigenmode related to the wireless channel. Specifically, $\overline{\lambda}_m = 1$ and $\overline{\lambda}_m = \sum_{n=1}^M \lambda_n |U_{mn}|^2$ for the UFC and CFC case, respectively, where $\lambda_n$ and $U_{mn}$ are elements of the matrices $\bm{\varLambda}$ and $\bm{U}$ in the Weichselberger model \eqref{eq:ChannelModel}. Additionally, we have defined the collective quantum coherence factor $\eta_{\rm{coh},m} = C_1 N_{\text{atom},m}T_{2,m}$ and collective quantum-optical transduction factor $\eta_{\rm{opt},m} = C_2 \kappa_{m}^{2}({\Omega_{\ell,m}}) \exp (- \tfrac{N_{\text{atom},m} {\bar \chi}_{m} }{A_{p,m}} ) = C_2 \left[ \frac{ N_{\text{atom},m} {\cal C} (\Omega_{\ell,m}) }{A_{p,m}} \right]^{2} \exp$ $(- \tfrac{N_{\text{atom},m} {\bar \chi}_{m} }{A_{p,m}} )$ for the $m$-th sensor.

		Based on the above results, we have the following Theorem for the inhomogeneous case in analogy with \textbf{Theorem} \ref{theorem1}. 
		\begin{theorem}
			\label{theorem5}
			Considering both the UFC and CFC scenarios of \eqref{eq:ChannelModel} and assuming a perfect CSI, we have the EAR of the inhomogeneous RAQ-MIMO for the $k$-th user formulated as  
			\begin{align}
				\label{eq:AchievableRate_1}
				\hspace{-1.5em}
				R_{k} \ge {\log _2} \left[ 
				1 + \mathcal{G}^{\rm{(mrc/zf)}} ( \underbrace{\mathsf{S}_{\rm{array}}^{(\rm{sql/psl})}, M, \overline{\Delta}}_{\rm{Rx}}, \underbrace{\beta_k}_{\rm{Ch}}, \underbrace{ {\cal P}_{{s}}, K }_{\rm{Tx}} ) 
				\right], 
			\end{align}
			where $\mathsf{S}_{\rm{array}}^{(\rm{sql})} = \{ N_{{\rm atom},m}, T_{2,m} \}_{m=1}^{M}$ and $\mathsf{S}_{\rm{array}}^{(\rm{psl})} = \{ \mathcal{P}_{0,m},$ $N_{{\rm atom},m}, \chi_m (\Omega_{\ell,m}), \varphi_m (\Omega_{\ell,m}), \kappa_m (\Omega_{\ell,m}) \}_{m=1}^{M}$, as well as
			\begin{align}
				\label{eq:mrcfunc_1}
				\mathcal{G}^{\rm{(mrc)}} &( \mathsf{S}_{\rm{array}}^{(\rm{sql/psl})}, M, \overline{\Delta}, \beta_k, {\cal P}_{{s}}, K ) \\
				\nonumber
				&= \left\{ 
				\begin{aligned}
					&\frac{ M - \overline{\Delta} }{ \left( \frac{B}{ {\cal P}_{s} \beta_{k} } \right) \left( \frac{1}{ \overline{\gamma}_{\rm{coh}}^{\rm{ave}} } \right) + \frac{\overline{\Delta}}{ \beta_k } \sum_{i = 1, \ne k}^K {{\beta _i}} }, & \rm{SQL}, \\
					&\frac{ M - \overline{\Delta} }{ \left( \frac{B}{ {\cal P}_{s} \beta_{k} } \right) \left( \frac{1}{ \overline{{\cal P}_0^{\rm ref}} \overline{\gamma}_{\rm{opt}}^{\rm{ave}} } \right) + \frac{\overline{\Delta}}{ \beta_k } \sum\nolimits_{i = 1, \ne k}^K {\beta _i} }, & \rm{PSL},
				\end{aligned} \right. \\
				\nonumber \\
				\label{eq:zffunc_1}
				\mathcal{G}^{\rm{(zf)}} &( \mathsf{S}_{\rm{array}}^{(\rm{sql/psl})}, M, \overline{\Delta}, \beta_k, {\cal P}_{{s}}, K ) \\
				\nonumber
				&= \left\{ 
				\begin{aligned}
					&\overline{\gamma}_{\rm{coh}}^{\rm{ave}} \left( \tfrac{ {\cal P}_{s} \beta_{k} }{B} \right) \left[ M - \overline{\Delta} (K-1) \right], & \rm{SQL}, \\
					&\overline{{\cal P}_0^{\rm ref}} \overline{\gamma}_{\rm{opt}}^{\rm{ave}} \left( \tfrac{ {\cal P}_{s} \beta_{k} }{B} \right) \left[ M - \overline{\Delta} (K-1) \right], & \rm{PSL}. 
				\end{aligned} \right. 
			\end{align}
		\end{theorem}
		\begin{IEEEproof}
			Upon substituting \eqref{eq:overline_SQL_gamma}, \eqref{eq:overline_PSL_gamma} into \eqref{eq:R_MRC} and \eqref{eq:R_ZF}, respectively,  we arrive at \eqref{eq:AchievableRate_1}-\eqref{eq:zffunc_1}. 
		\end{IEEEproof}
		
		\textit{\textbf{Remark 7}: Compared to the former homogeneous counterpart, the EAR variation behaviour of the inhomogeneous RAQ-MIMO receiver is determined by the array inhomogeneity $\overline{\Delta}$ and by the array-average quantum-related and optical-related parameters $\mathsf{S}_{\mathrm{array}}^{(\mathrm{sql})}$ and $\mathsf{S}_{\mathrm{array}}^{(\mathrm{psl})}$ through $\overline{\gamma}_{\rm{coh}}^{\rm{ave}}$ and $\overline{\gamma}_{\rm{opt}}^{\rm{ave}}$, respectively, as observed from \textbf{Theorem} \ref{theorem5}.}
		
		\textit{\textbf{Remark 8}: Based on expressions of $\overline{\gamma}_{\rm{coh}}^{\rm{ave}}$ and $\overline{\gamma}_{\rm{opt}}^{\rm{ave}}$ in \eqref{eq:overline_SQL_gamma} and \eqref{eq:overline_PSL_gamma}, we know that the contribution of each sensor ($\eta_{{\rm coh},m}$ or $\eta_{{\rm opt},m}$) is weighted by the corresponding channel eigenmode coefficient $\overline{\lambda}_m$. Therefore, the performance is determined by the joint interaction between the array-side configurations and the spatial eigen-structure of the correlated channel. This is quite different from the homogeneous RAQ-MIMO, where all sensors contribute identically to the achievable rate through a common sensor parameter.}

		\subsubsection{Scaling Behavior of Inhomogeneous RAQ-MIMO}
		The scaling behaviors of EAR and the transmit power of inhomogeneous RAQ-MIMO arrays are investigated as follows. 
		
		\textbf{EAR Scaling Law (SQL regime)}: 
		In analogy with \textbf{Theorem} \ref{theorem2}, we have the following results for the EAR scaling behavior of the inhomogeneous case in the SQL regime. Specifically, as the average eigenmode-weighted collective quantum coherence factor increases, the EAR of the inhomogeneous MRC RAQ-MIMO receiver in the SQL regime becomes saturated, while the EAR of the inhomogeneous ZF RAQ-MIMO receiver in this regime scales logarithmically without bound, formulated as 
		\begin{align}
			\label{eq:EAR_MRC_ScalingLaw_SQL_Inhomogeneous}
			&\hspace{-0.4em}
			R_{k}^{\rm{(mrc, ufc/cfc, sql)}}
			\xrightarrow[]{\overline{\gamma}_{\rm{coh}}^{\rm{ave}} \to \infty} 
			{\log _2} \left(
			1 + \frac{ M - \overline{\Delta} }{ \frac{\overline{\Delta}}{ \beta_k } \sum_{i = 1, \ne k}^K {{\beta _i}} }
			\right), \\
			\label{eq:EAR_ZF_ScalingLaw_SQL_Inhomogeneous}
			&R_{k}^{\rm{(zf, ufc/cfc, sql)}} \propto {\log _2} \big( \overline{\gamma}_{\rm{coh}}^{\rm{ave}} \big) \xrightarrow[]{\overline{\gamma}_{\rm{coh}}^{\rm{ave}} \to \infty} \infty. 
		\end{align}

		\textbf{EAR Scaling Law (PSL regime)}: 
		In analogy with \textbf{Theorem} \ref{theorem3}, we have the following results for the EAR scaling behavior of the inhomogeneous case in the PSL regime. Specifically, for the inhomogeneous RAQ-MIMO having an independent Rydberg atomic ensemble, the EAR of the MRC and ZF schemes in the PSL regime can be maximized at an atomic number of $N_{\text{atom},m}^{\star} = \frac{2 A_{p,m}}{\chi_m}$ or at a cell's length of $L_m^{\star} = \frac{2}{N_{0} \chi_m}$ or at a laser's cross-sectional area of $A_{p,m}^{\star} =$ $\frac{1}{2} N_{\text{atom},m} {\bar \chi}_{m}$ for the $m$-th sensor. These optimal configurations yield the optimal average eigenmode-weighted collective quantum coherence factor across all $M$ sensors as follows 
		\begin{align}
			\label{eq:XXX}
			\overline{\gamma}_{\rm{opt}}^{\rm{ave} \star} 
			= \frac{C_3}{M} \sum_{m=1}^M \overline{\lambda}_m \upsilon_m \left[ \frac{ {\cal C}_m (\Omega_{\ell,m}) }{\chi_m} \right]^2,  
		\end{align}
		and lead to the maximum EAR \eqref{eq:EAR_MRC_ZF_ScalingLaw_PSL_inhomo} at the top of next page. 
		\begin{figure*}
			\begin{equation}
				\label{eq:EAR_MRC_ZF_ScalingLaw_PSL_inhomo} 
				\begin{aligned}
					\hspace{-0.2em}
					R_{k}^{\rm{(mrc/zf,ufc/cfc,psl)} \star} 
					\xrightarrow[ \text{or}\; A_{p,m} \to A_{p,m}^{\star} \atop \text{or}\; \overline{\gamma}_{\rm{opt}}^{\rm{ave}} \rightarrow \overline{\gamma}_{\rm{opt}}^{\rm{ave} \star} ]{N_{\text{atom},m} \rightarrow N_{\text{atom},m}^{\star} \atop \text{or}\; L_m \rightarrow L_{m}^{\star} } 
					\left\{ 
					\begin{aligned}
						&{\log_2} \left[ 
						\begin{split} 
							1 + \frac{ M - \overline{\Delta} }{ \dfrac{1}{ {C_3} \overline{{\cal P}_0^{\rm ref}} } \dfrac{ M }{ \sum_{m=1}^M \overline{\lambda}_m \upsilon_m \left[ \frac{ {\cal C}_m (\Omega_{\ell,m}) }{\chi_m} \right]^2 } \left( \dfrac{B}{ {\cal P}_{s} \beta_{k} } \right) + \dfrac{ \overline{\Delta} }{ \beta_k } \sum\limits_{i = 1 \atop i \ne k}^K {\beta _i} } 
						\end{split} 
						\right], & \rm{MRC}, \\
						&\log_2 \left[ 1 + C_3 \overline{{\cal P}_0^{\rm ref}} \left[ 1 - \frac{\overline{\Delta} (K - 1)}{M} \right] \frac{ {\cal P}_{s} \beta_{k} }{B} \sum_{m=1}^M \overline{\lambda}_m \upsilon_m \left[ \frac{ {\cal C}_m (\Omega_{\ell,m}) }{\chi_m} \right]^2 \right], & \rm{ZF}. \;\;\; 
					\end{aligned} \right. 
				\end{aligned}
			\end{equation}
			\hrulefill
		\end{figure*}

		\textbf{Power Scaling Law}: 
		Upon defining the number of atoms as $N_{\rm atom}^{\rm ref}$ for reference, we can reformulate the atom number of the $m$-th sensor as $N_{{\rm atom},m} \triangleq \overline{\alpha}_m N_{\rm atom}^{\rm ref}$, where $\overline{\alpha}_m$ represents the local inhomogeneity coefficient. Next, we define an average term across $N_{\rm atom}^{\rm ref}$ as $\overline X_{\rm{abs}} \triangleq N_{\rm atom}^{\rm ref} ( \frac{1}{M} \sum_{i=1}^{M} \frac{\overline{\alpha}_i\bar\chi_i}{A_{p,i}} )$. In analogy with $\tau \in \{ T_2, \frac{ {\cal C} (\Omega_{\ell}) }{A_p} \}$ defined in \textbf{Theorem} \ref{theorem4} for the homogeneous case, we define $\overline{\tau} \in \{ \overline{\tau}_{\rm sql}, \overline{\tau}_{\rm psl} \}$ for the inhomogeneous case as follows 
		\begin{align}
			\nonumber
			\overline{\tau}_{\rm{sql}} \triangleq \frac{1}{M} \sum_{m=1}^{M} \overline{\alpha}_m T_{2,m}, \;\; 
			\overline{\tau}_{\rm{psl}} \triangleq \sqrt{\frac{1}{M} \sum_{m=1}^{M} \left[ \frac{\overline{\alpha}_m {\cal C}(\Omega_{\ell,m})}{A_{p,m}} \right]^2}. 
		\end{align}
		With these definitions, we provide the following theorem. 
		\begin{theorem}
			\label{theorem6}
			For an inhomogeneous RAQ-MIMO system, we assume non-negative constants $\overline C_{\rm{sql}}$, $\overline C_{\rm{psl}}$, and that the number of sensors is $M \in \{ \overline C_{\rm{sql}} N_{\rm atom}^{\rm ref} \overline{\tau}_{\rm{sql}},\; \overline C_{\rm{psl}} \exp ( \overline X_{\rm{abs}} ) \}$ in the SQL and PSL regimes, respectively. As $N_{\rm atom}^{\rm ref} \rightarrow \infty$, the transmit power ${\cal P}_{s}$ can be quadratically scaled down with respect to $N_{\rm atom}^{\rm ref} \overline{\tau}$ in the form of ${\cal P}_{s} = {{\cal E}}/{ (N_{\rm atom}^{\rm ref} \overline{\tau})^{2} }$ for a fixed reference energy ${\cal E}$. Concurrently, the EAR of the $k$-th user for the MRC and ZF schemes converges respectively to 
			\begin{align}
				\nonumber
				&R_{k}^{\rm{(mrc/zf)}} 
				\xrightarrow[ ]{ {\cal P}_{s} = {{\cal E}}/{ (N_{\rm atom}^{\rm ref} \overline{\tau})^{2} } \;\text{and}\; N_{\rm atom}^{\rm ref} \rightarrow \infty }  \\ 
				\label{eq:GeneralizedPowerScalingLaw} 
				&\qquad \;\; {\log_2} \left[ 1 + \mathcal{G}_{\infty}^{\rm{(mrc/zf)}} ( \mathsf{S}_{\rm{array}}^{(\rm{sql/psl})}, \overline{\Delta}, \beta_k, {\cal E}, K ) \right], \\
				\nonumber \\
				\label{eq:mrc_g1}
				&\mathcal{G}_{\infty}^{\rm{(mrc)}} ( \mathsf{S}_{\rm{array}}^{(\rm{sql/psl})}, \overline{\Delta}, \beta_k, {\cal E}, K ) \\
				\nonumber
				&= \left\{ 
				\begin{aligned}
					&\frac{ \overline C_{\rm{sql}} (1 - C_{\overline{\Delta}} ) C_1 \overline{\mu}_{\rm{sql}}^{\rm{norm}} \left( \frac{ {\cal E} \beta_{k} }{B} \right) }{ 1 + \overline C_{\rm{sql}} C_{\overline{\Delta}} C_1 \overline{\mu}_{\rm{sql}}^{\rm{norm}} \left( \frac{ {\cal E} }{B} \right) \sum_{i = 1, \ne k}^K {{\beta _i}} }, & \rm{SQL}, \\
					&\frac{ \overline C_{\rm{psl}} (1 - C_{\overline{\Delta}} ) {C_2} \overline{\mu}_{\rm{psl}}^{\rm{norm}} \overline{{\cal P}_0^{\rm ref}} \left( \frac{ {\cal E} \beta_{k} }{B} \right) }{ 1 + \overline C_{\rm{psl}} C_{\overline{\Delta}} {C_2} \overline{\mu}_{\rm{psl}}^{\rm{norm}} \overline{{\cal P}_0^{\rm ref}} \left( \frac{ {\cal E} }{B} \right) \sum_{i = 1, \ne k}^K {\beta _i} }, & \rm{PSL},
				\end{aligned} \right. \\  
				\nonumber \\
				\label{eq:zf_g1}
				&\mathcal{G}_{\infty}^{\rm{(zf)}} ( \mathsf{S}_{\rm{array}}^{(\rm{sql/psl})}, \overline{\Delta}, \beta_k, {\cal E}, K ) \\
				\nonumber
				&= \left\{ 
				\begin{aligned}
					&\overline C_{\rm{sql}} \left[1 - C_{\overline{\Delta}} (K-1) \right] C_1 \overline{\mu}_{\rm{sql}}^{\rm{norm}} \left( \frac{ {\cal E} \beta_{k} }{B} \right), & \rm{SQL}, \\
					&\overline C_{\rm{psl}} \left[1 - C_{\overline{\Delta}} (K-1) \right] {C_2} \overline{\mu}_{\rm{psl}}^{\rm{norm}} \overline{{\cal P}_0^{\rm ref}} \left( \frac{ {\cal E} \beta_{k} }{B} \right), & \rm{PSL}. 
				\end{aligned} \right. 
			\end{align}
			where we have $\frac{1}{M} \le C_{ \overline{\Delta} } \le 1$ and $\frac{1}{M} \le C_{\overline{\Delta}} \le \frac{1}{K-1}$ for the MRC and ZF schemes, respectively; $\overline{\mu}_{\rm sql/psl}^{\rm{norm}}$ are given by 
			\begin{subequations}
				\label{eq:efficiencies}
				\begin{align}
					\overline{\mu}_{\rm{sql}}^{\rm{norm}} &= \frac{1}{M \overline{\tau}_{\rm{sql}}} \sum_{m=1}^{M} \overline{\lambda}_m \upsilon_m \overline{\alpha}_m T_{2,m}, \\
					\nonumber
					\overline{\mu}_{\rm{psl}}^{\rm{norm}} &= \frac{1}{M \overline{\tau}_{\rm{psl}}^2} \sum_{m=1}^{M} \overline{\lambda}_m \upsilon_m \left( \frac{\overline{\alpha}_m{\cal C}(\Omega_{\ell,m})}{A_{p,m}} \right)^2 \\
					&\times \exp \left[ -N_{\rm atom}^{\rm ref} \left( \frac{\overline{\alpha}_m\bar\chi_m}{A_{p,m}} - \frac{1}{M} \sum_{i=1}^{M} \frac{\overline{\alpha}_i\bar\chi_i}{A_{p,i}} \right) \right].
				\end{align}
			\end{subequations}
		\end{theorem}
		\begin{IEEEproof} 
			In the SQL regime, $\overline{\gamma}_{\rm coh}^{\rm ave}$ presented in \eqref{eq:overline_SQL_gamma} can be directly reformulated with respect to $\overline{\tau}_{\rm{sql}}$ in the form of 
			\begin{align}
				\label{eq:gamma_sql_reformulate}
				\overline{\gamma}_{\rm{coh}}^{\rm{ave}} 
				= C_1 \left( N_{\rm atom}^{\rm ref} \overline{\tau}_{\rm{sql}} \right) \overline{\mu}_{\rm{sql}}^{\rm{norm}}. 
			\end{align} 
			Additionally, in the PSL regime, $\overline{\gamma}_{\rm{opt}}^{\rm{ave}}$ given in \eqref{eq:overline_PSL_gamma} can be reformulated as follows 
			\begin{equation}
				\label{eq:gamma_psl_reformulate}
				\begin{aligned}
					\overline{\gamma}_{\rm{opt}}^{\rm{ave}} 
					= C_2 \left( N_{\rm atom}^{\rm ref} \overline{\tau}_{\rm{psl}} \right)^2 \overline{\mu}_{\rm{psl}}^{\rm{norm}} \exp(-\overline X_{\rm{abs}}).
				\end{aligned} 
			\end{equation}
			Next, upon substituting $\overline{\gamma}_{\rm{coh}}^{\rm{ave}}$ in \eqref{eq:gamma_sql_reformulate} and $\overline{\gamma}_{\rm{opt}}^{\rm{ave}}$ in \eqref{eq:gamma_psl_reformulate}, as well as the power scaling law ${\cal P}_s = {{\cal E}}/{( N_{\rm atom}^{\rm ref} \overline{\tau})^2}$, $\overline{\tau} \in \{ \overline{\tau}_{\rm sql}, \overline{\tau}_{\rm psl} \}$ into the SQL and PSL expressions of \eqref{eq:mrcfunc_1} and \eqref{eq:zffunc_1}, respectively, we arrive at 
			\begin{align}
				\label{eq:mrc_g}
				&\mathcal{G}^{\rm{(mrc)}} ( \mathsf{S}_{\rm{array}}^{(\rm{sql/psl})}, M, \overline{\Delta}, \beta_k, {\cal P}_{{s}}, K ) \\
				\nonumber
				&= \left\{ 
				\begin{aligned}
					&\frac{ \frac{ M - \overline{\Delta} }{ N_{\rm atom}^{\rm ref} \overline{\tau}_{\rm{sql}} } }{ \left( \frac{B}{ {\cal E} \beta_{k} } \right) \left( \frac{ 1 }{ C_1 \overline{\mu}_{\rm{sql}}^{\rm{norm}} } \right) + \frac{\overline{\Delta}}{ \beta_k \left( N_{\rm atom}^{\rm ref} \overline{\tau}_{\rm{sql}} \right) } \sum\limits_{i = 1, \ne k}^K {{\beta _i}} }, & \rm{SQL}, \\
					&\frac{ \frac{ M - \overline{\Delta} }{ \exp(\overline X_{\rm{abs}}) } }{ \left( \frac{B}{ {\cal E} \beta_{k} } \right) \left( \frac{1}{ C_2 \overline{{\cal P}_0^{\rm ref}} \overline{\mu}_{\rm{psl}}^{\rm{norm}} } \right) + \frac{\overline{\Delta}}{ \beta_k \exp(\overline X_{\rm{abs}}) } \sum\limits_{i = 1, \ne k}^K {\beta _i} }, & \rm{PSL},
				\end{aligned} \right. \\ 
				\nonumber \\
				\label{eq:zf_g}
				&\mathcal{G}^{\rm{(zf)}} ( \mathsf{S}_{\rm{array}}^{(\rm{sql/psl})}, M, \overline{\Delta}, \beta_k, {\cal P}_{{s}}, K ) \\
				\nonumber
				&= \left\{ 
				\begin{aligned}
					&C_1 \overline{\mu}_{\rm{sql}}^{\rm{norm}} \left( \frac{ {\cal E} \beta_{k} }{B} \right) 
					\times \frac{ M - \overline{\Delta} (K-1) }{ N_{\rm atom}^{\rm ref} \overline{\tau}_{\rm{sql}} }, & \rm{SQL}, \\
					&{C_2} \overline{\mu}_{\rm{psl}}^{\rm{norm}} \overline{{\cal P}_0^{\rm ref}} \left( \frac{ {\cal E} \beta_{k} }{B} \right) 
					\frac{ M - \overline{\Delta} (K-1) }{ \exp(\overline X_{\rm{abs}}) }, & \rm{PSL}. 
				\end{aligned} \right. 
			\end{align}
			Based on $\overline{\Delta} \in [1, M]$, we assume that $\overline{\Delta}$ scales linearly with $M$ in the form of $\overline{\Delta} \triangleq C_{\overline{\Delta}} M$, where we have $\frac{1}{M} \le C_{ \overline{\Delta} } \le 1$. 
			By furthermore exploiting $M = \overline C_{\rm{sql}} N_{\rm atom}^{\rm ref} \overline{\tau}_{\rm{sql}}$ for the SQL regime, we can obtain the asymptotic results $\frac{M - \overline{\Delta} }{ N_{\rm atom}^{\rm ref} \overline{\tau}_{\rm{sql}} } = \overline C_{\rm{sql}} (1 - C_{\overline{\Delta}} )$ and $\frac{M - \overline{\Delta} (K-1) }{ N_{\rm atom}^{\rm ref} \overline{\tau}_{\rm{sql}} } = \overline C_{\rm{sql}} \left[1 - C_{\overline{\Delta}} (K-1) \right]$ for the MRC and ZF schemes, respectively. Likewise, by exploiting $M = \overline C_{\rm{psl}} \exp(\overline X_{\rm{abs}})$ for the PSL regime, we can obtain the asymptotic results $\frac{ M - \overline{\Delta} }{\exp(\overline X_{\rm{abs}}) } = \overline C_{\rm psl} \left[ 1 - C_{\overline{\Delta}} \right]$ and $\frac{ M - \overline{\Delta} (K-1) }{\exp(\overline X_{\rm{abs}}) } = \overline C_{\rm psl} \left[ 1 - C_{\overline{\Delta}} (K-1) \right]$ for the MRC and ZF schemes, respectively. 
			More particularly, for the ZF scheme, $C_{\overline{\Delta}}$ should satisfy the condition of $\frac{1}{M} \le C_{\overline{\Delta}} \le \frac{1}{K-1}$ for the sake of acquiring non-negative EAR. Upon substituting these results into \eqref{eq:mrc_g} and \eqref{eq:zf_g}, respectively,	we arrive at \eqref{eq:GeneralizedPowerScalingLaw}. 
		\end{IEEEproof}
		
		\textit{\textbf{Remark 9}: The generalized scaling laws in \eqref{eq:mrc_g1} and \eqref{eq:zf_g1} exhibit an asymptotic non-equivalence triggered by $C_{\overline{\Delta}}$. Specifically, the MRC scheme tolerates a non-vanishing IUI floor in its denominator, creating a strict \textit{EAR ceiling} that forces the system into an interference-limited state. By contrast, the ZF scheme perfectly penalizes the IUI to zero, but pays a heavy toll in \textit{degrees-of-freedom bleeding}, where its effective spatial degree-of-freedom shrinks by a factor of $(K-1)C_{\overline{\Delta}}$. Additionally, as $C_{\overline{\Delta}} \to 0$, the RAQ-MIMO array becomes homogeneous and \textbf{Theorem} \ref{theorem6} reduces to \textbf{Theorem} \ref{theorem4}. }

	\vspace{-0.6em}
	\subsection{Comparison Between RAQ-MIMO and Classical M-MIMO}
	
	Following a similar derivation process, we present the asymptotic EARs for classical M-MIMO systems as follows 
	\begin{equation}
		\label{eq:MIMO_AchievableRate_MRC_Asymptotic1_UFC/CFC}
		\tilde{R}_{k}^{\rm{(mrc,ufc/cfc)}} = {\log _2} \left[ 
		\begin{split}
			1 + \frac{ ( M - \varepsilon) \mathsf{SNR}_{0,k} }{ 1 + \varepsilon \frac{ \mathsf{SNR}_{0,k} }{ \beta_{k} } \sum_{i = 1, \ne k}^K {{\beta _i}} } 
		\end{split} 
		\right], 
	\end{equation}
	\begin{equation}
		\label{eq:MIMO_AchievableRate_ZF_Asymptotic1_UFC/CFC}
		\tilde{R}_{k}^{\rm{(zf,ufc/cfc)}} = {\log _2} \;\; \Big[ 1 + \left[ M - \varepsilon (K-1) \right] \mathsf{SNR}_{0,k} \Big], 
	\end{equation}
	where $\varepsilon$ is related to the UFC/CFC scenarios and are provided in \textbf{Lemma} \ref{lemma1}; $\mathsf{SNR}_{0,k} \triangleq \frac{ \varrho_0 {A_{\rm{iso}}} {\cal P}_{s} }{ \sigma_0^2 } {\beta_k}$ represents the received SNR of each antenna corresponding to the $k$-th user, without experiencing the small-scale fading and IUIs; $\varrho_{0}$ is the gain of a single RF chain of M-MIMO and ${A_{\rm{iso}}}$ is the effective aperture of an isotropic antenna. More particularly, $\varrho_{0} \triangleq \eta_0 G_{\rm{Ant}} G_{\rm{LNA}}$ is determined by the antenna efficiency $\eta_0$, antenna gain $G_{\rm{Ant}}$, and the LNA gain $G_{\rm{LNA}}$ employed for the RF chain. The effective aperture of the isotropic antenna is given by ${A_{\rm{iso}}} = \lambda^2 / (4 \pi)$. The noise power of M-MIMO, namely $\sigma_{0}^{2}$, is different from that of the RAQ-MIMO. It is given by $\sigma_{0}^{2} = 10\log(k_B T_0) + 10\log B + NF + G_{\rm{LNA}}$ in dB, where $T_0 = 290$ K is the room temperature, $NF = 10 \log(F)$ represents the noise figure of the receiver, and $F$ is the corresponding noise factor. For example, $NF$ is $6$ dB and $9$ dB for the base station (BS) and the user equipment (UE), respectively, at the band of 5G FR1 n104 \cite{3GPP_IMT}.

	To proceed, let us define $\mathsf{RSG} \triangleq \frac{ \overline{\mathsf{SNR}}_{1,k} }{ \mathsf{SNR}_{0,k} }$ as the receiver SNR gain (RSG) of the RAQ-MIMO over the classical M-MIMO. Based on the definitions of $\overline{\mathsf{SNR}}_{1,k}$ and $\mathsf{SNR}_{0,k}$, the RSG can be interpreted as “$\textit{\textbf{ReceiverGainRatio}} / \textit{\textbf{NoisePowerRatio}}$” in the condition of a single sensor through the relationship of 
	\begin{align}
		\label{eq:RSG}
		\mathsf{RSG} = \underbrace{\left( \frac{ \varrho {{\cos }^2}\varphi ({\Omega_{\ell}}) }{ A_{{\rm{iso}}} {\varrho _0} } \right)}_{{\text{Receiver gain ratio}}} \bigg/ \hspace{-1.2em} \underbrace{\left( {\frac{{{\sigma ^2}}}{{\sigma _0^2}}} \right)}_{ \text{Noise power ratio} } \hspace{-1em} \triangleq \Pi. 
	\end{align}
	Furthermore, based on the results of $\overline{\mathsf{SNR}}_{1,k}$ in \eqref{eq:SNR1_SQL}, \eqref{eq:SNR1_PSL}, as well as the result of $\mathsf{SNR}_{0,k} = \frac{ \varrho_0 {A_{\rm{iso}}} }{ k_{\rm{B}} T_0 F } \left( \frac{ {\cal P}_{s} }{B} \right) {\beta_k}$, we obtain 
	\begin{align}
		\nonumber
		\Pi = \left\{ 
		\begin{aligned}
			&C_1 N_{\text{atom}} T_2 \left( \frac{ k_{\rm{B}} T_{0} F }{ {\eta_0} G_{\rm{Ant}} A_{\rm{iso}} } \right), & \rm{SQL}, \\
			&C_2 {\overline {\cal P}_0 } {\kappa^2}({\Omega_{\ell} }) \exp (- \tfrac{N_{\text{atom}} {\bar \chi}}{A_p} ) 
			\left( \frac{ k_{\rm{B}} T_{0} F }{ {\eta_0} G_{\rm{Ant}} A_{\rm{iso}} } \right), & \rm{PSL}.
		\end{aligned} \right. 
	\end{align}
	We next compare the proposed RAQ-MIMO receiver to the classical M-MIMO receivers in terms of the EAR, transmit power, and of the transmission distance.

	\subsubsection{Increase of the EAR}
	\label{subsubsect:IEAR}
	We proceed by setting the transmit power ${\cal P}_{s}$ and the number of sensors $M$ for the RAQ-MIMO to be identical to those of the classical M-MIMO, respectively. 
	\begin{corollary}
		\label{corollary:3}
		In the high-SINR regime, the difference between the EARs of MRC RAQ-MIMO and MRC M-MIMO receivers in both UFC and CFC scenarios approaches
		\begin{align}
			\label{eq:SINRratio1Asmp}
			\Delta \tilde{R}_{1,k}^{\rm{(mrc)}} 
			\rightarrow 
			\begin{cases}
				\;\;\;\; 0, & \textit{\textbf{C1}}: \left( {\cal P}_{s} \rightarrow \infty \right), \;\;\;\;\;\;\;\;\; \\
				\log_{2} \Pi, & \textit{\textbf{C2}}: \left( {{\cal P}_{s} = \tfrac{{\cal E}}{M} , M \rightarrow \infty \hfill} \right).   
			\end{cases} 
		\end{align}
		By contrast, in the high-SINR regime, the difference between the EARs of ZF RAQ-MIMO and ZF M-MIMO receivers in both UFC and CFC scenarios is always equal to
		\begin{align}
			\label{eq:SINRratio2}
			&\Delta \tilde{R}_{1,k}^{\rm{(zf)}}
			= \log_{2} \Pi. 
		\end{align}
	\end{corollary}
	\begin{IEEEproof}
		For the MRC receivers of RAQ-MIMO and M-MIMO, upon dividing the SINR terms of \eqref{eq:AchievableRate_MRC_Asymptotic_1} by the SINR term in \eqref{eq:MIMO_AchievableRate_MRC_Asymptotic1_UFC/CFC}, respectively, we obtain the following result for both the UFC and CFC scenarios 
		\begin{align}
			\nonumber
			\mathsf{RT}_{1,k}^{\rm{(mrc,ufc/cfc)}} 
			&= \frac{ \frac{ \beta_k }{ \mathsf{SNR}_{0,k} } + \varepsilon \sum_{i = 1, \ne k}^K {\beta _i} }{ \frac{ \beta_k }{ \overline{\mathsf{SNR}}_{1,k} } + \varepsilon \sum_{i = 1, \ne k}^K {\beta _i} } \\
			\label{eq:ratio1Asmp}
			&\hspace{-0.8em} \rightarrow 
			\begin{cases}
				1, & \textit{\textbf{C1}}: \left( {\cal P}_{s} \rightarrow \infty \right), \\
				\Pi, & \textit{\textbf{C2}}: \left( {\cal P}_{s} = \tfrac{{\cal E}}{M}, M \rightarrow \infty \right). 
			\end{cases} 
		\end{align}
		For the ZF receivers of RAQ-MIMO and M-MIMO, based on the SINR terms of \eqref{eq:AchievableRate_ZF_Asymptotic_1}, as well as the SINR term of \eqref{eq:MIMO_AchievableRate_ZF_Asymptotic1_UFC/CFC}, we obtain a unified SINR ratio formulated as  
		\begin{align}
			\label{eq:ratio2}
			\mathsf{RT}_{1,k}^{\rm{(zf,ufc/cfc)}}  
			&= \Pi. 
		\end{align}
		Upon using $\Delta \tilde{R}_{1,k}^{\rm{(mrc/zf)}} \approx \log_{2} \mathsf{RT}_{1,k}^{\rm{(mrc/zf)}}$ in the high-SINR regime, we arrive at \eqref{eq:SINRratio1Asmp} and \eqref{eq:SINRratio2}. 
	\end{IEEEproof}

	\textbf{Corollary} \ref{corollary:3} reveals that the comparison between RAQ-MIMO and M-MIMO receivers finally boils down to the  comparison of a single-sensor receiver chain. As verified in \cite{gong2026rydberg}, $\overline{\mathsf{SNR}}_{1,k} > \mathsf{SNR}_{0,k}$ is achievable, so that $\Delta \tilde{R}_{1,k}^{\rm{(mrc/zf)}} \ge 0$ can be realized, implying that the MRC/ZF RAQ-MIMO receivers can outperform their M-MIMO counterparts.

	The superiority of the MRC RAQ-MIMO receiver over the MRC M-MIMO receiver gradually reduces for higher ${\cal P}_{s}$, given other parameters, as shown by the \textit{\textbf{C1}} case in \eqref{eq:SINRratio1Asmp}. We also observe from \eqref{eq:SINRratio1Asmp} that the MRC RAQ-MIMO receiver outperforms the MRC M-MIMO receiver with an extra EAR contribution in the \textit{\textbf{C2}} case, when ${\cal P}_{s}$ obeys the classical power scaling law of ${\cal P}_{s} = \tfrac{{\cal E}}{M}$ for a fixed energy of ${\cal E}$. By contrast, the ZF RAQ-MIMO receiver outperforms the ZF M-MIMO receiver in any conditions.

	\subsubsection{Reduction of Transmit Power}
	In this assessment, we assume different transmit powers, while setting all other parameters to be identical for the RAQ-MIMO and M-MIMO. 
	\begin{corollary}
		\label{corollary:4}
		The transmit power of users in RAQ-MIMO systems can be reduced by a factor of $\Pi$ compared to that of classical M-MIMO, when realizing the same EAR. 
		
		For example, the transmit power of users in RAQ-MIMO systems can be reduced by a factor of $\sim 446$ (e.g., $\sim 26.5$ dB) and of $\sim 10000$ (e.g., $\sim 40$ dB) in the PSL and SQL, respectively, compared to the 5G-BS \cite{3GPP_IMT}, when obeying the simulation configuration of Section \ref{sec:Simulations}. 
	\end{corollary}
	\begin{IEEEproof}
		Upon equalling the EAR of the ZF RAQ-MIMO receiver in \eqref{eq:AchievableRate_ZF_Asymptotic_1} to that of the ZF M-MIMO receiver in \eqref{eq:MIMO_AchievableRate_ZF_Asymptotic1_UFC/CFC}, respectively, we obtain the unified result of $\overline{\mathsf{SNR}}_{1,k} = \frac{ {\varrho} {{\cos }^2}\varphi ({\Omega_{\ell}}) {\cal P}_{s,1} }{\sigma^2} {\beta_k} = \mathsf{SNR}_{0,k} = \frac{ \varrho_0 {A_{\rm{iso}}} {\cal P}_{s,0} }{ \sigma_0^2 } {\beta_k}$ for both UFC and CFC, where ${\cal P}_{s,1}$ and ${\cal P}_{s,0}$ represent the transmit power of users in the RAQ-MIMO and M-MIMO systems, respectively. \textbf{Corollary} \ref{corollary:4} can be obtained accordingly. 
	\end{IEEEproof}

	\subsubsection{Extension of Transmission Distance}
	We assume different transmission distances, while setting all other parameters to be identical for the RAQ-MIMO and M-MIMO systems. 
	\begin{corollary}
		\label{corollary:5}
		The transmission distance of RAQ-MIMO systems can be $\sqrt[\nu]{ \Pi }$-fold longer than that of classical M-MIMO systems, when achieving the same EAR, where $\nu$ represents the path-loss exponent. 
		
		For instance, the free-space transmission distance of RAQ-MIMO systems can be $\sim 21$-fold and $\sim 100$-fold higher in the PSL and SQL, respectively, compared to the 5G-BS \cite{3GPP_IMT} and obeying the simulation configuration of Section \ref{sec:Simulations}. 
	\end{corollary}
	\begin{IEEEproof}
		Upon equating \eqref{eq:AchievableRate_ZF_Asymptotic_1} to  \eqref{eq:MIMO_AchievableRate_ZF_Asymptotic1_UFC/CFC}, we obtain 
		$\overline{\mathsf{SNR}}_{1,k} = \frac{ {\varrho} {{\cos }^2}\varphi ({\Omega_{\ell}}) {\cal P}_{s}}{\sigma^2} {\beta_{k,1}} = \mathsf{SNR}_{0,k}$ $= \frac{ \varrho_0 {A_{\rm{iso}}} {\cal P}_{s} }{ \sigma_0^2 } {\beta_{k,0}}$, where ${\beta_{k,1}}$ and ${\beta_{k,0}}$ denote the large-scale fading corresponding to the different transmission distances of RAQ-MIMO and M-MIMO, respectively. By exploiting the large-scale fading model $\beta_k = \beta_{\rm{ref}} + 10 \log_{10} (1/D_k)^{\nu} + F_{k}$, where $\beta_{\rm{ref}}$ is the large-scale fading at the reference distance of $1$ meter, $D_k$ is the distance between the receiver and the $k$-th user, and $F_{k}$ is the shadow fading, we obtain \textbf{Corollary} \ref{corollary:5}. 
	\end{IEEEproof}

	\begin{table}[t]
		\centering
		\caption{\textsc{RAQ-MIMO parameters in simulations.}}
		\renewcommand{\arraystretch}{0.6}
		\label{tab:parameters}
			\begin{tabular}{@{} m{1.2cm} m{6.8cm} @{}}
				\toprule
				\textbf{Category} & \textbf{Parameter Value and Unit} \\
				\midrule
				
				\multirow{8}{*}{\textbf{\makecell[l]{Electron\\ transitions}}} 
				&\makecell[l]{ 
					Vapor cell length: $L=10$ cm \\
					Effective atomic density: $N_{0}=4.89 \times 10^{8}$ $\text{cm}^{-3}$ \\
					Dipole moment of $\ket{1}$ \textrightarrow $\ket{2}$: $\mu_{12}=2.2327 q a_{0}$ C/m \\
					Dipole moment of $\ket{2}$ \textrightarrow $\ket{3}$: $\mu_{23}=0.0226 q a_{0}$ C/m \\
					Dipole moment of $\ket{3}$ \textrightarrow $\ket{4}$: $\mu_{34}=1443.45 q a_{\mathrm{0}}$ C/m \\
					Decay rate of $\ket{2}$: $\gamma_{2}=5.2$ MHz \\
					Decay rate of $\ket{3}$: $\gamma_{3}=3.9$ kHz \\
					Decay rate of $\ket{4}$: $\gamma_{4}=1.7$ kHz \\
					Total dephasing rate: $\varGamma_2=5$ MHz \\
					Coherence time: $T_2=0.2$ \textmu s
				} \\
				\midrule
				
				\multirow{8}{*}{\textbf{\makecell[l]{Laser\\ beams\\ and RF\\ signals}}} 
				&\makecell[l]{ 
					Probe beam wavelength: $\lambda_{p}=852$ nm \\
					Coupling beam wavelength: $\lambda_{c}=510$ nm \\
					Probe beam power: $\mathcal{P}_{0}=20.7$ \textmu W \\
					Coupling beam power: $\mathcal{P}_{c}=17$ mW \\
					Local optical beam power: $\mathcal{P}_{\text{lob}}=30$ mW \\
					Probe/coupling beam radius: $r_{0}=1.7$ mm \\
					LO signal amplitude: $U_{\ell}=0.0661$ V/m \\
					Carrier frequency: $f_{c}=6.9458$ GHz \\
					Frequency difference between LO and RF: $f_{\delta}=150$ kHz \\
					RF bandwidth: $B=100$ kHz
				} \\
				\midrule
				
				\multirow{9}{*}{\textbf{\makecell[l]{Antenna,\\ LNA,\\ and\\ others}}} 
				&\makecell[l]{ 
					Antenna efficiency: $\eta_0 = 0.7$ \\
					BS antenna element gain (5G FR1 n104): $G_{\text{ant}}=5.5$ dB \\
					UE antenna element gain (5G FR1 n104): $G_{\text{ant}}=0$ dB \\
					System noise figure of the BS and UE: $NF=6, 9$ dB \\
					LNA gain of the classical RF receiver: $G_{\text{LNA}}=60$ dB \\
					Noise temperature of the LNA: $T_{\text{LNA}}=100$ K \\
					Room temperature: $T_{\text{room}}=300$ K \\
					Quantum efficiency of the PD: $\eta_1=0.8$ \\
					LNA gain of the PD: $G_{\text{pd}}=30$ dB \\
					LNA noise temperature of the PD: $T=100$ K \\
					Load resistance of the PD: $R=1$ Ohm
				} \\
				\bottomrule
			\end{tabular}
		\vspace{-0.8em}
	\end{table}

	\section{Simulation Results}
	\label{sec:Simulations}
	To characterize the performance of RAQ-MIMO and verify its potential, in this section we present simulations quantifying its average EAR versus (vs.) diverse parameters.
	We consider the four-level electron transition scheme of  6S\textsubscript{\scalebox{0.8}{1/2}} \textrightarrow 6P\textsubscript{\scalebox{0.8}{3/2}} \textrightarrow 47D\textsubscript{\scalebox{0.8}{5/2}} \textrightarrow 48P\textsubscript{\scalebox{0.8}{3/2}}. The corresponding parameters of the electron transitions, laser beams, RF signals, and electronic components are presented in TABLE \ref{tab:parameters}. We note that these parameter configurations are consistent with the physics experiments of \cite{jing2020atomic} and 3GPP specifications \cite{3GPP_IMT}. 
	Furthermore, for multiple users, each one is randomly positioned within a circular area that has a radius of $300$ meters. Its center is $400$ meters away from the RAQ-MIMO. The large-scale fading between the RAQ-MIMO and users is given by $-30 + 38 \log_{10} \left( 1/D_k \right) + F_{k}$ in dB, where $F_{k} \sim \mathcal{CN}(0, \sigma_{sf}^2)$ with $\sigma_{sf} = 10$. The small-scale fading obeys \eqref{eq:ChannelModel}. The signal bandwidth considered is $100$ kHz. Our simulation results are averaged over $5000$ realizations.

	\vspace{-0.5em}
	\subsection{Model Verification and Feasible Region}
	The nonlinear output of the PD of RAQ-MIMO is characterized by the equality of \eqref{eq:PhotodetectorOutputBCOD}, where the corresponding linear model is presented by the approximation of \eqref{eq:PhotodetectorOutputBCOD}. 
	To validate the effectiveness of the approximation, we characterize their accuracy in Fig. \ref{fig:ModelVerification}(a). Observe from this figure that the proposed linear approximation model fits the exact nonlinear model well in a large dynamic range. The accuracy significantly degrades as the received power of all user signals exceed the LO's amplitude of $-23.6$ dBV/m. We note that this phenomenon is in line with the experimental results of \cite{jing2020atomic}. 
	Furthermore, we explore the feasible region of exploiting this linear approximation model for multi-user RAQ-MIMO systems in Fig. \ref{fig:ModelVerification}(b). It is observed that the feasible region is above the boundary curves. For example, when the number of users supported in RAQ-MIMO systems is $K = 20$, the transmission distance of all users from the RAQ-MIMO receiver should be higher than $3$, $6$, and $30$ meters for the propagation environment having a path exponent of $2.0$, $3.8$, $6.0$, respectively, in order to guarantee the model's linearity. Since this condition is readily met in wireless communications, it reflects that our linear model offers a satisfactory accuracy for a wide range of communication scenarios.

	\begin{figure}[t!]
		\hspace{-0.6em}
		\subfloat[]{
			\includegraphics[width=0.24\textwidth]{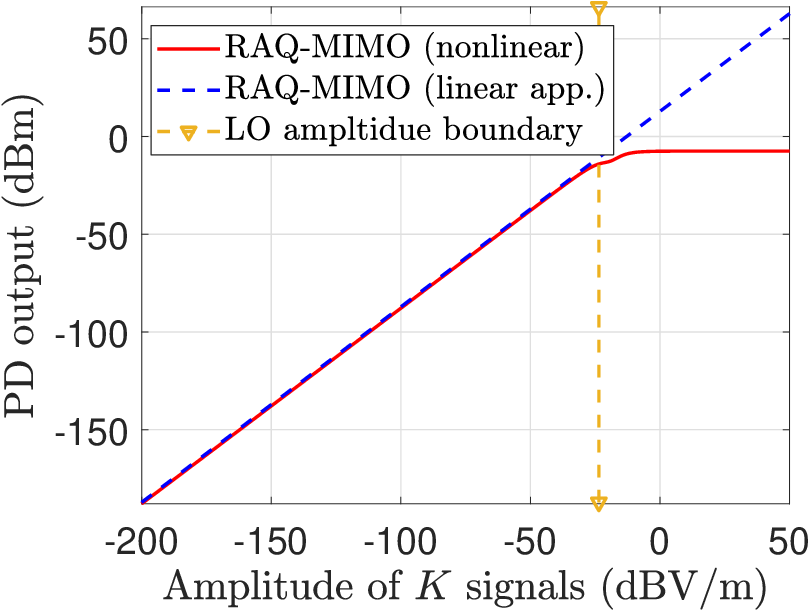}} 
		\subfloat[]{
			\includegraphics[width=0.24\textwidth]{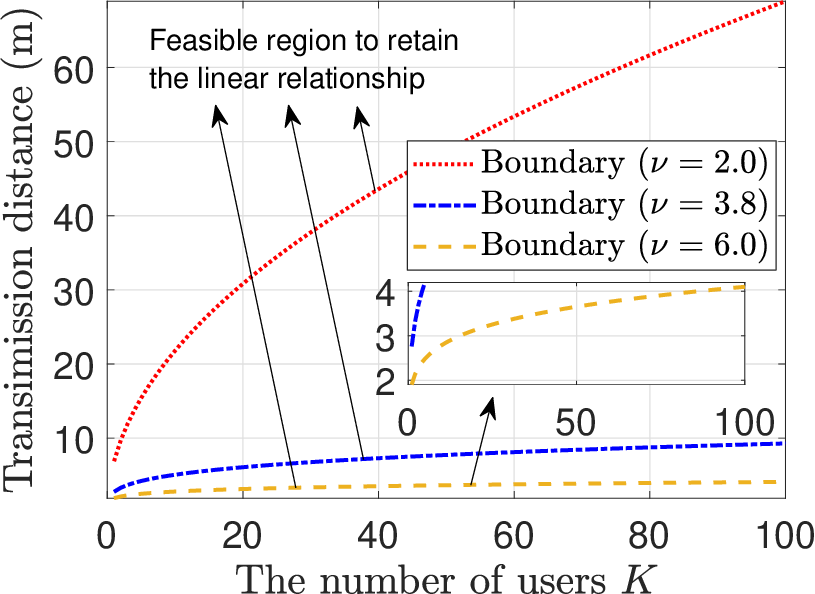}} 
		\vspace{-0.2em}
		\caption{(a) Accuracy of the proposed linear signal model and (b) feasible region to retain the linear relationship.}
		\vspace{-0.6em}
		\label{fig:ModelVerification}
	\end{figure}

	\vspace{-0.5em}
	\subsection{Scaling Behavior of Homogeneous RAQ-MIMO} 
	To understand the EAR scaling law, we characterize the EAR of both the MRC and ZF RAQ-MIMO receivers vs. the number of atoms $N_{\text{atom}}$ in Fig. \ref{fig:ScalingLaw}(a)(b), where $M=200$ and $K = 10$. We vary $N_{\text{atom}}$ logarithmically within the range $1.11 \times 10^7 \sim 2.225 \times 10^9$. When operating in the SQL regime, the MRC RAQ-MIMO curves of both the UFC and CFC scenarios become saturated as $N_{\text{atom}}$ increases, as seen in Fig. \ref{fig:ScalingLaw}(a). By contrast, the ZF RAQ-MIMO curves of both the UFC and CFC scenarios is proportional to $\log_{2}(N_{\text{atom}} T_2)$, as seen in Fig. \ref{fig:ScalingLaw}(b). The above numerical results are consistent with \textbf{Theorem} \ref{theorem2}. When operating in the PSL regime, the MRC RAQ-MIMO and ZF RAQ-MIMO curves of both the UFC and CFC scenarios increase to their maximum and then decrease, as indicated by the optimal line in Fig. \ref{fig:ScalingLaw}(a)(b). These numerical results are consistent with \textbf{Theorem} \ref{theorem3}. 
	
	Furthermore, we showcase the \textit{power scaling law} of the RAQ-MIMO receiver by presenting numerical results of both ${\cal P}_{s} = {\cal E} / (N_{\text{atom}} \tau)^2$ vs. $N_{\text{atom}}$ and $M$ vs. $N_{\text{atom}}$ in Fig. \ref{fig:ScalingLaw}(c), where we have ${\cal E} \triangleq 10 {\cal P}_{s}$. Upon configuring $C_{\text{sql}} = C_{\text{psl}} = 5$, we furthermore offer numerical results characterizing the MRC/ZF RAQ-MIMO EARs vs. $N_{\text{atom}}$ in Fig. \ref{fig:ScalingLaw}(d) under the configuration of Fig. \ref{fig:ScalingLaw}(c). We can observe from these figures that \textbf{(i)} the transmit power can be significantly reduced as $N_{\text{atom}}$ becomes large, allowing us to apply RAQ-MIMO for ultra-low power communications; \textbf{(ii)} The MRC/ZF RAQ-MIMO receivers operating in the SQL regime can achieve a significant EAR enhancement compared to its counterparts in the PSL regime, even though deploying $M \propto N_{\text{atom}} T_2$ sensors in the SQL regime compared to the exponential sensors deployment via $M \propto \exp \big( \frac{N_{\text{atom}} {\bar \chi}}{A_p} \big)$ in the PSL regime; \textbf{(iii)} As $N_{\text{atom}}$ becomes large, both the MRC and ZF RAQ-MIMO curves tend to converge to their corresponding theoretical limits in the SQL and PSL regime, respectively, as noted in \eqref{eq:PowerScalingLaw} of \textbf{Theorem} \ref{theorem4}.

	\begin{figure*}[t!]
		\hspace{-0.7em}
		\subfloat[]{
			\includegraphics[width=0.237\textwidth]{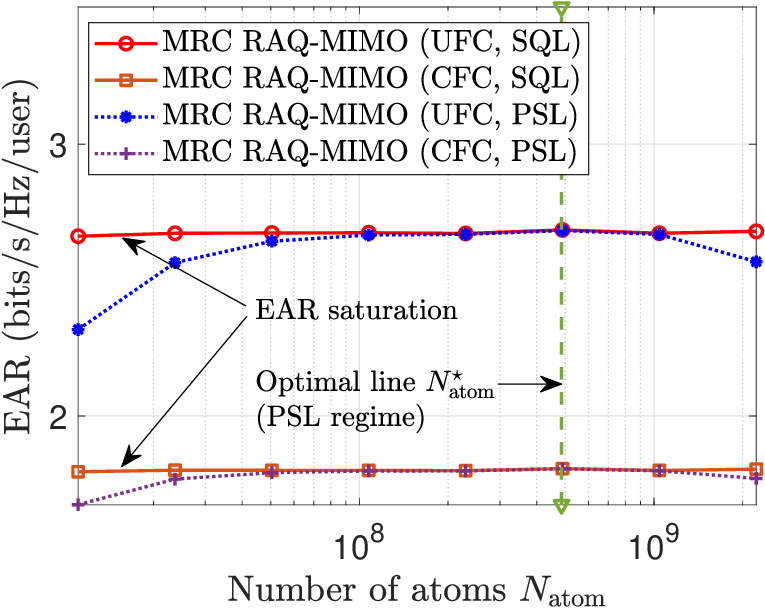}} 
		\subfloat[]{
			\includegraphics[width=0.242\textwidth]{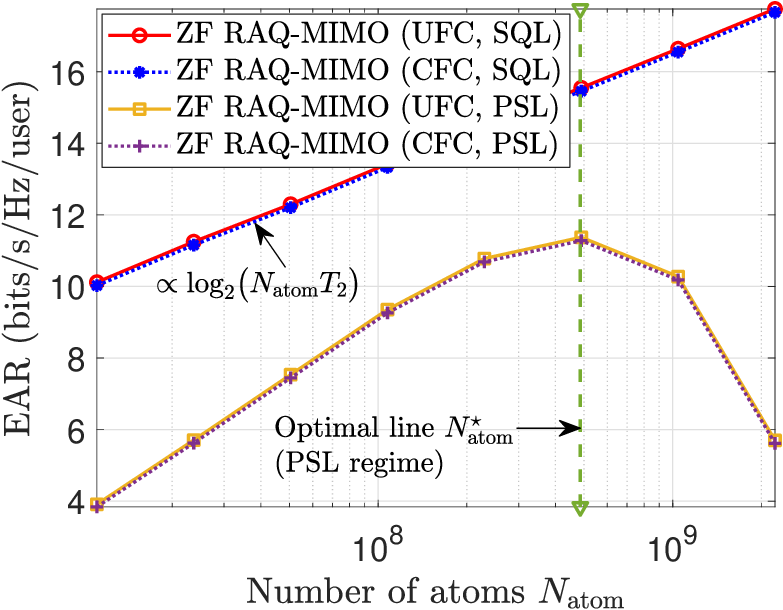}}  
		\subfloat[]{
			\includegraphics[width=0.268\textwidth]{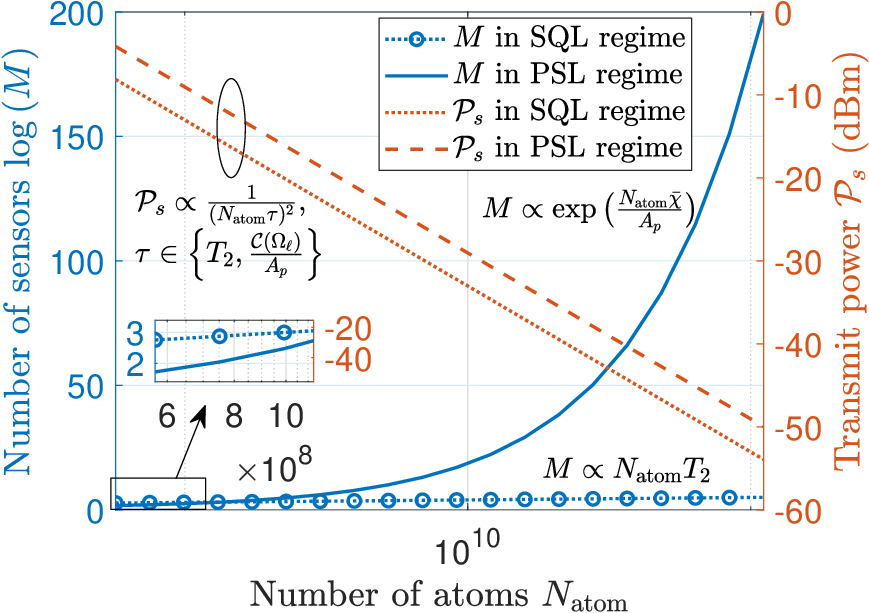}} 
		\subfloat[]{
			\includegraphics[width=0.242\textwidth]{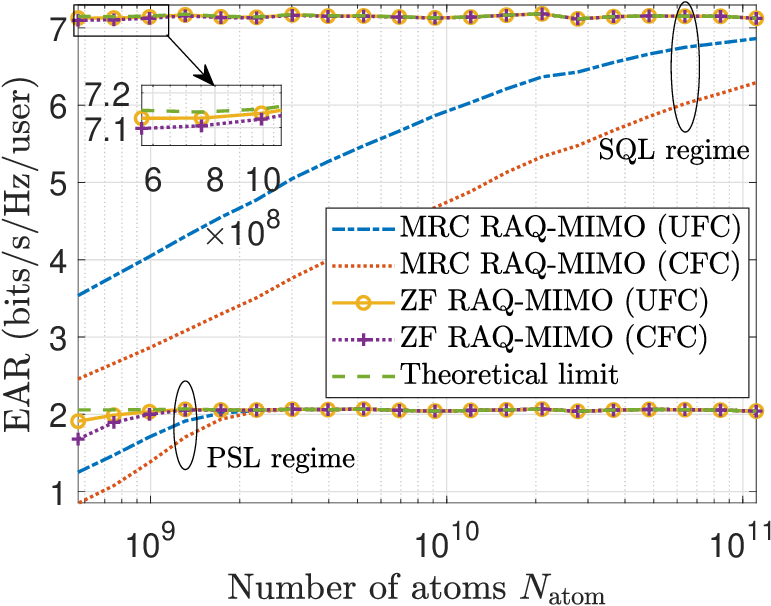}} 
		\vspace{-0.2em}
		\caption{Homogeneous array: The EAR scaling behavior of (a) the MRC and (b) ZF RAQ-MIMO receivers in UFC/CFC scenarios. (c) The scaling behavior of ${\cal P}_{s}$ and $M$, as well as (d) the EAR scaling behavior when retaining ${\cal P}_{s} = {\cal E} / (N_{\text{atom}} \tau)^2$.}
		\vspace{-0.6em}
		\label{fig:ScalingLaw}
	\end{figure*}

	\begin{figure*}[t!]
		\hspace{-0.6em}
		\subfloat[]{
			\includegraphics[width=0.238\textwidth]{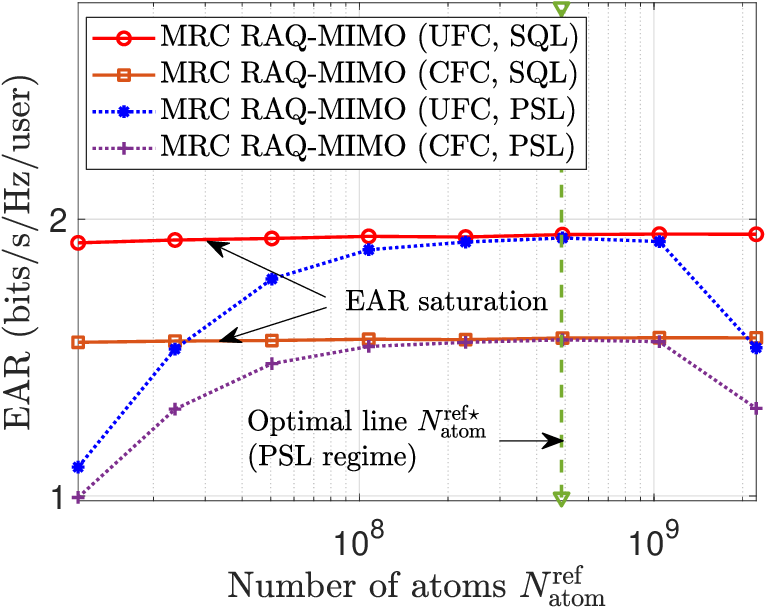}} 
		\subfloat[]{
			\includegraphics[width=0.242\textwidth]{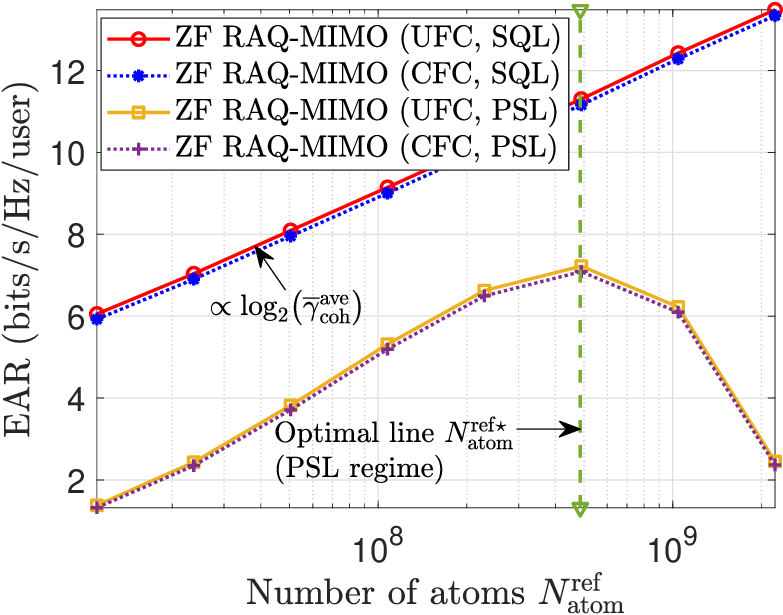}} 
		\subfloat[]{
			\includegraphics[width=0.266\textwidth]{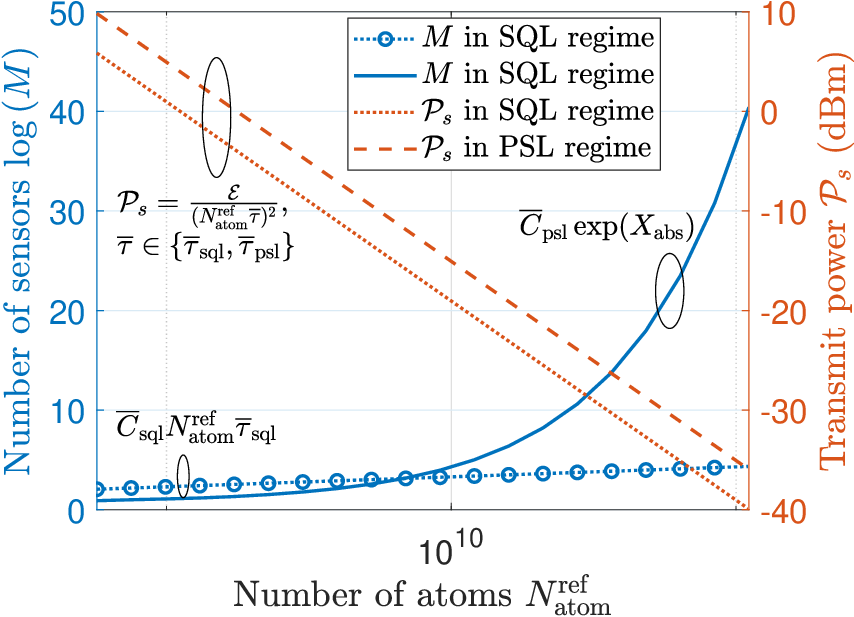}}
		\subfloat[]{
			\includegraphics[width=0.238\textwidth]{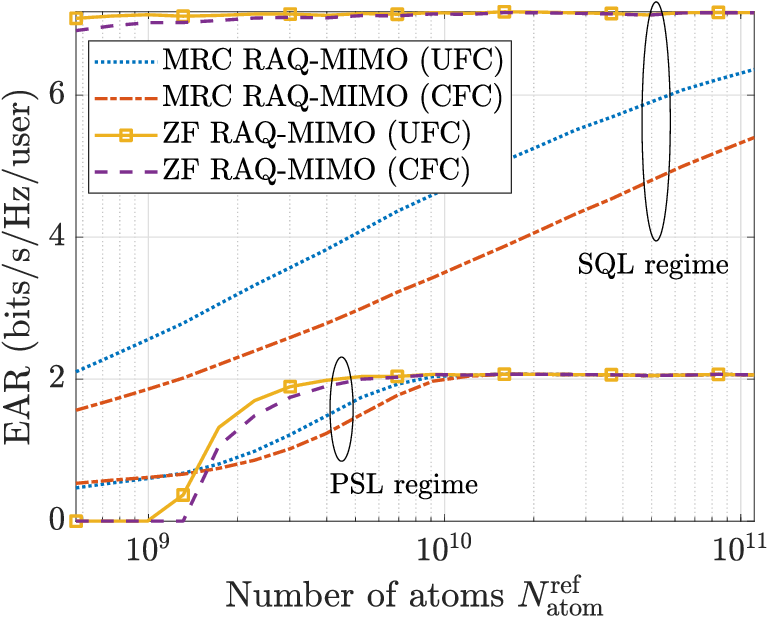}} 
		\vspace{-0.2em}
		\caption{Inhomogeneous array: The EAR scaling behavior of (a) the MRC and (b) ZF RAQ-MIMO receivers in UFC/CFC scenarios. (c) The scaling behavior of ${\cal P}_{s}$ and $M$, as well as (d) the EAR scaling behavior when retaining ${\cal P}_{s} = {\cal E} / (N_{\rm{atom}}^{\rm ref} \overline{\tau})^2$.}
		\vspace{-0.6em}
		\label{fig:ScalingLaw_inhomo}
	\end{figure*}

	\vspace{-0.5em}
	\subsection{Scaling Behavior of Inhomogeneous RAQ-MIMO} 
	
	To explicitly articulate the impact of array non-uniformity on the performance, we characterize the EAR scaling behavior under the inhomogeneous array configuration in Fig. \ref{fig:ScalingLaw_inhomo}(a)(b), where a sensor having $\overline{\gamma}_{\rm{ave}} \in \{ \overline{\gamma}_{\rm{coh}}^{\rm{ave}} = \eta_{\rm coh}, \overline{\gamma}_{\rm{opt}}^{\rm{ave}} = \eta_{\rm opt} \}$ serves as the reference sensor. By contrast, the values of $\overline{\gamma}_{\rm{ave}}$ of the remaining $M-1$ sensors are configured $1/20$ times that of the reference counterpart, so that the RAQ-MIMO array becomes highly inhomogeneous. From these sub-figures, we can observe that the inhomogeneous MRC receiver exhibits a more progressive approach toward its saturation limit compared to its homogeneous counterpart, where the saturation limit is determined by the reference number of atoms $N_{\text{atom}}^{\text{ref}}$ and the heterogeneity factor $\overline{\Delta}$. More particularly, the inhomogeneous ZF RAQ-MIMO receiver operating in the SQL regime scales as a function of the average eigenmode-weighted collective quantum coherence factor $\log_2(\overline{\gamma}_{\rm{coh}}^{\rm{ave}})$ across all $M$ sensors. When the inhomogeneous ZF RAQ-MIMO receiver operating in the PSL regime, a maximum also exist around $N_{\text{atom}}^{\rm{ref} \star}$, as theoretically presented in \eqref{eq:EAR_MRC_ZF_ScalingLaw_PSL_inhomo}. Compare to its homogeneous counterpart, the trend is similar but the EAR in the inhomogeneous case inevitably degrades. 
	
	We further characterize the power scaling law of the inhomogeneous RAQ-MIMO system in Fig. \ref{fig:ScalingLaw_inhomo}(c)(d), where a sensor having $N_{\text{atom}}^{\text{ref}} \in [ 1.11 \times 10^7, 2.225 \times 10^9 ]$ atoms serves as the reference, while values of $N_{\text{atom},m}$ of the remaining sensors are configured $1/5$ times that of the reference counterpart. Although the scaling behavior of ${\cal P}_{s}$ and $M$ in the inhomogeneous case has a similar trend compared to the homogeneous counterpart as observed from Fig. \ref{fig:ScalingLaw_inhomo}(c), the concrete scaling formulas are different, where the inhomogeneous RAQ-MIMO array relies on average operations across difference sensors. Furthermore, from Fig. \ref{fig:ScalingLaw_inhomo}(d), when retaining ${\cal P}_{s} = {\cal E} / (N_{\rm{atom}}^{\rm ref} \overline{\tau})^2$ and $M \in \{ \overline C_{\rm{sql}} N_{\rm atom}^{\rm ref} \overline{\tau}_{\rm{sql}},\; \overline C_{\rm{psl}} \exp ( \overline X_{\rm{abs}} ) \}$ in the SQL and PSL regimes, respectively, the inhomogeneity also triggers slow convergence compared to the homogeneous case. We also notice that the EAR of the ZF RAQ-MIMO receiver (in the PSL regime) drops to zero at the small value region of $N_{\text{atom}}^{\text{ref}}$, where the contribution of array inhomogeneity becomes significance. By contrast, the MRC RAQ-MIMO receiver can maintain a non-zero EAR at the same region, highlighting its advantageous operational region.

	\begin{figure*}[t!]
		\centering
		\subfloat[]{
			\;\includegraphics[width=0.3\textwidth]{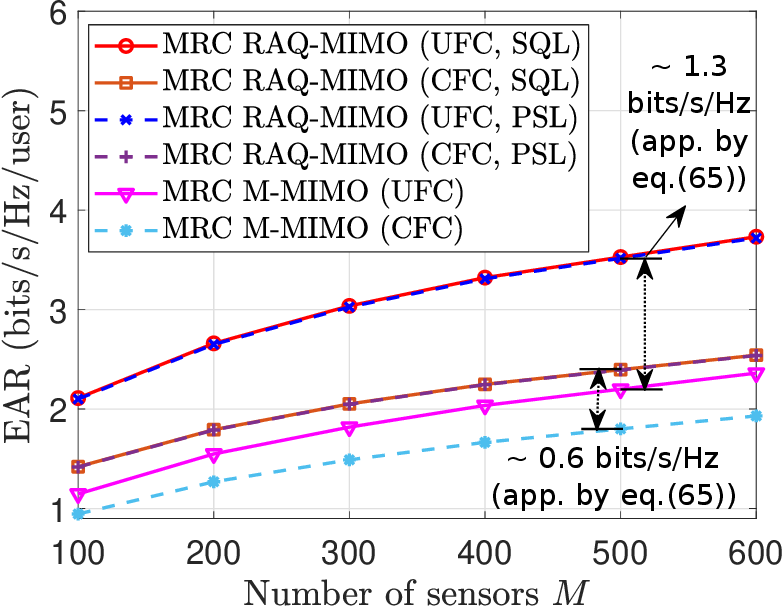}} \;
		\subfloat[]{
			\includegraphics[width=0.3\textwidth]{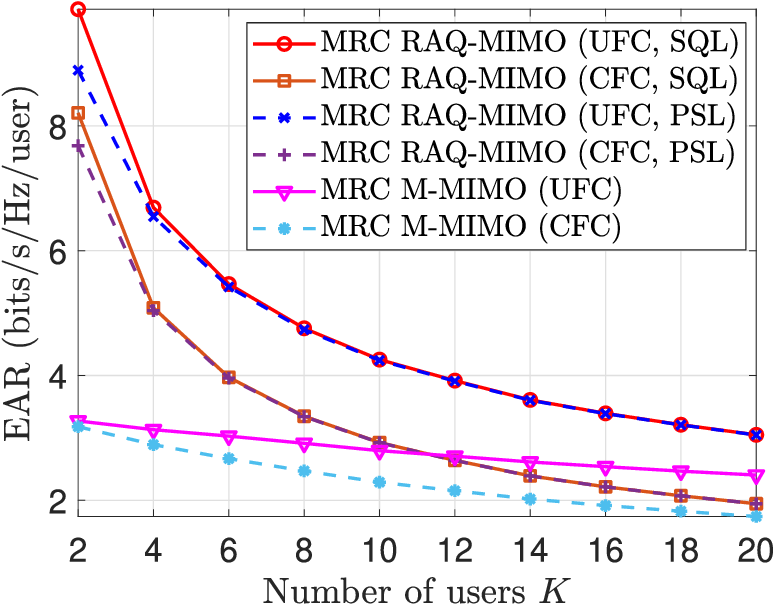}} \;
		\subfloat[]{
			\includegraphics[width=0.296\textwidth]{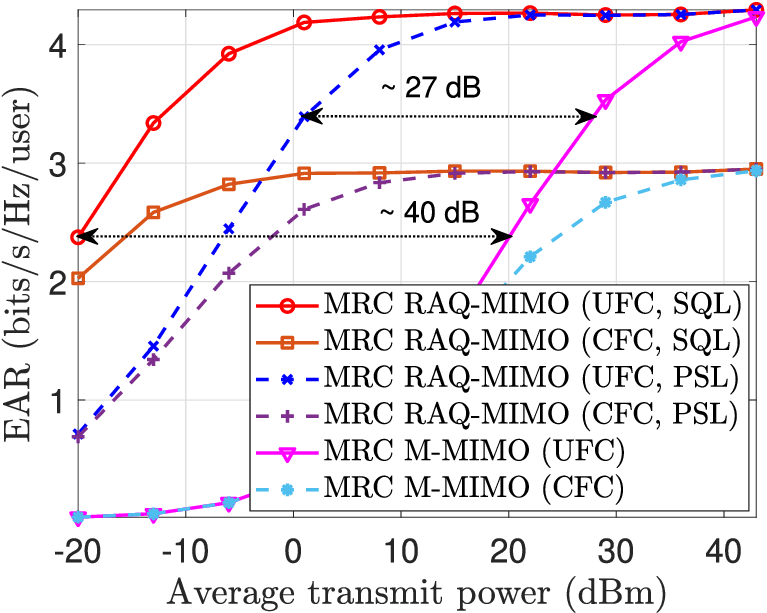}} \\ \vspace{0.6em}
		\subfloat[]{
			\includegraphics[width=0.296\textwidth]{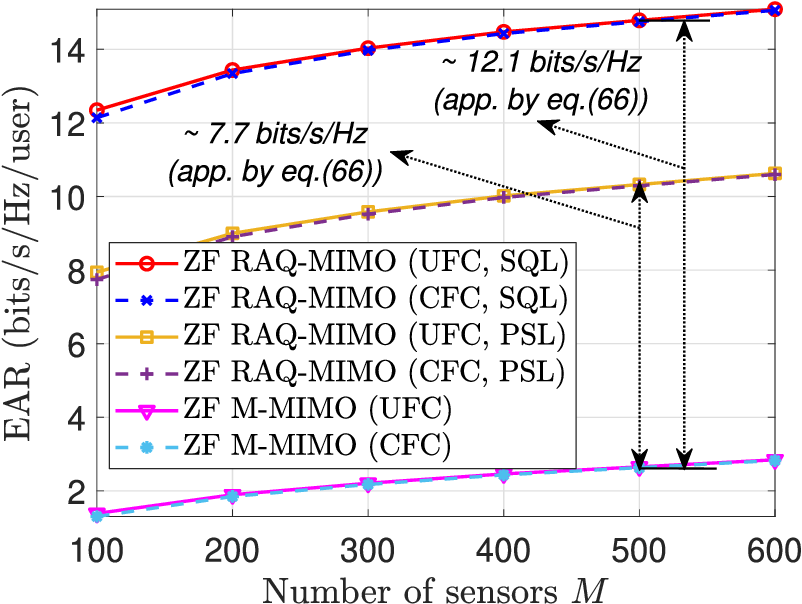}} \;
		\subfloat[]{
			\includegraphics[width=0.296\textwidth]{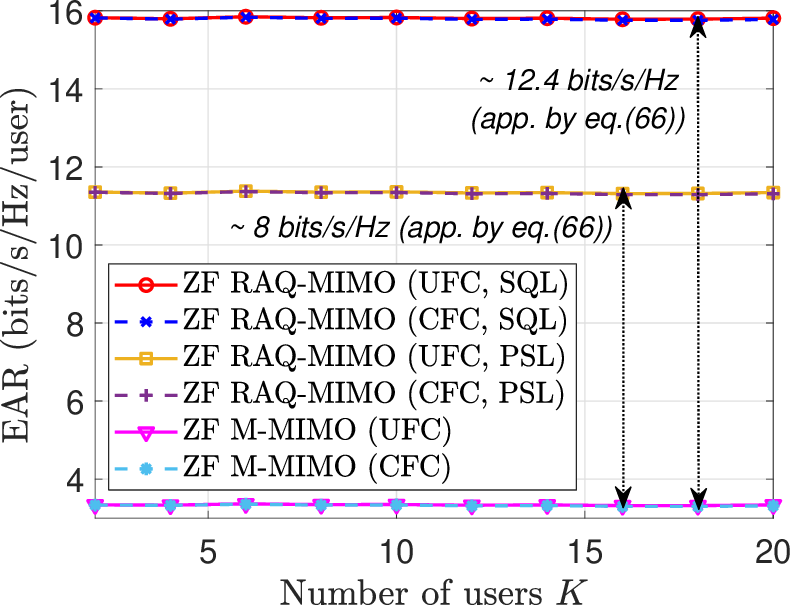}} \;
		\subfloat[]{
			\includegraphics[width=0.296\textwidth]{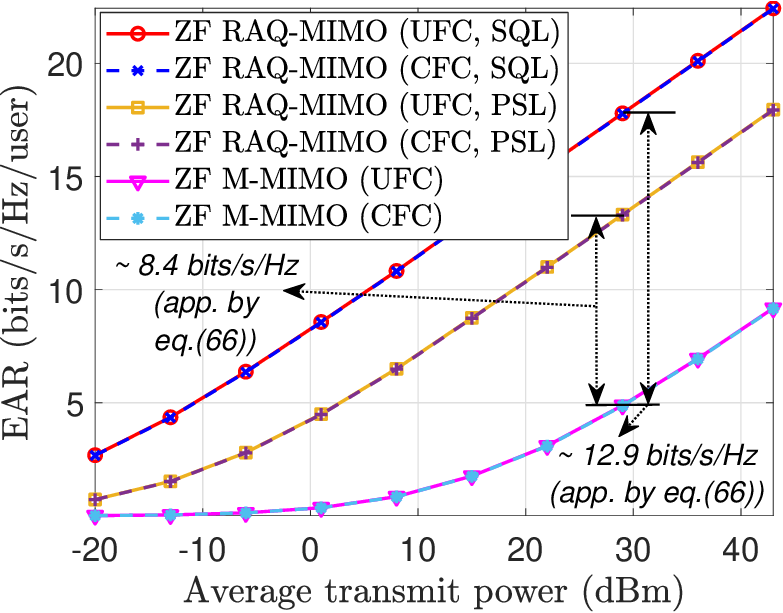}} 
		\vspace{-0.2em}
		\caption{Comparison of RAQ-MIMO to M-MIMO receivers with respect to diverse parameters: (a)(d) EAR vs. number of sensors $M$; (b)(e) EAR vs. number of users $K$; (c)(f) EAR vs. transmit power ${\cal P}_{s}$.}
		\label{fig:RAQ_MIM0_M_MIMO_Comp}
	\end{figure*}

	\vspace{-0.5em}
	\subsection{Comparison Between RAQ-MIMO and Classical M-MIMO} 
	In these comparisons, we present the RAQ-MIMO receiver in the \textbf{SQL} and \textbf{PSL} regimes. The MRC and ZF schemes are portrayed in Fig. \ref{fig:RAQ_MIM0_M_MIMO_Comp}(a)(b)(c) and Fig. \ref{fig:RAQ_MIM0_M_MIMO_Comp}(d)(e)(f), respectively.

	As observed from Fig. \ref{fig:RAQ_MIM0_M_MIMO_Comp}(a), where we have ${\cal P}_{s} = 23$ dBm and $K = 10$, the MRC RAQ-MIMO receiver operating in the SQL and PSL regimes outperforms similarly to the MRC M-MIMO receiver. This is because the IUI cannot be sufficiently suppressed by the MRC scheme, hence the extreme sensitivity of RAQ-MIMO cannot be exploited. This trend becomes more apparent in Fig. \ref{fig:RAQ_MIM0_M_MIMO_Comp}(b), where ${\cal P}_{s} = 23$ dBm and $M = 1000$. The curves of both the MRC RAQ-MIMO and MRC M-MIMO receivers move closer to each other as $K$ increases. The advantage of MRC RAQ-MIMO becomes more obvious for low transmit power scenarios, as shown in Fig. \ref{fig:RAQ_MIM0_M_MIMO_Comp}(c). The MRC RAQ-MIMO receiver allows the transmit power of users to be $\sim 40$ dB and $\sim 27$ dB lower than that of the MRC M-MIMO receiver in the SQL and PSL regimes. By contrast, as portrayed in Fig. \ref{fig:RAQ_MIM0_M_MIMO_Comp}(d)(e)(f), the ZF RAQ-MIMO receiver always outperforms the ZF M-MIMO receiver in terms of different $M$, $K$, and ${\cal P}_{s}$ in the SQL and PSL regimes. The extra EARs achieved in the SQL and PSL regimes are on the order of $12$ bits/s/Hz/user and $8$ bits/s/Hz/user, respectively.

	\begin{table}[t!]
		\footnotesize
		\renewcommand{\arraystretch}{1.6}
		\caption{\textsc{Comparison of RAQ-MIMO to M-MIMO receivers with respect to the transmit power.}}
		\label{tab:power_reduction}
		\centering
		\tabcolsep = 0.15cm
			\begin{tabular}{|c|c|c|c|c|c|}
				\rowcolor{cyan!10} 
				\hline \textbf{Schemes} & \multicolumn{4}{|c|}{ \textbf{Power (dBm)} } & \tabincell{c}{\textbf{Power reduction (dB)}} \\
				\hline Classical M-MIMO & 0 & 10 & 20 & 30 & 0 \\
				\hline RAQ-MIMO (PSL) & -26.5 & -16.5 & -6.5 & 3.5 & 26.5 \\
				\hline RAQ-MIMO (SQL) & -40 & -30 & -20 & -10 & 40 \\
				\hline
			\end{tabular}
		\vspace{-0.6em}
	\end{table}

	\begin{figure}[!t]
		\centering
		\includegraphics[width=0.4\textwidth]{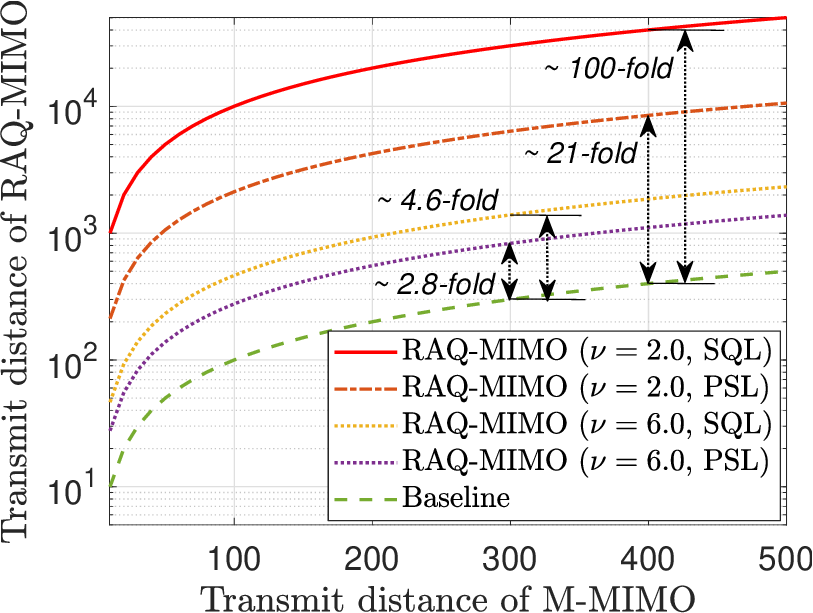}
		\caption{Comparison of RAQ-MIMO to M-MIMO receivers with respect to the transmission distance under the same EAR.}
		\vspace{-1em}
		\label{fig:Power_Dist_Comp}
	\end{figure}

	We also present the reduction of users' transmit power and the increase of the transmission distance in TABLE \ref{tab:power_reduction} and Fig. \ref{fig:Power_Dist_Comp}, respectively. As seen from TABLE \ref{tab:power_reduction}, when achieving the same EAR, the RAQ-MIMO receiver supports a reduction of $\sim 40$ dB and $\sim 26.5$ dB of the users' transmit power in the SQL and PSL regimes, respectively, compared to M-MIMO receivers. When using the same transmit power for the RAQ-MIMO and M-MIMO systems, the former supports a farther transmission distance than the latter, as shown in Fig. \ref{fig:Power_Dist_Comp}. Specifically, we consider $2.0 \le \nu \le 6.0$, covering diverse propagation environments. When $\nu = 6.0$, the transmission distances supported by the RAQ-MIMO receiver are $\sim 4.6$-fold and $\sim 2.8$-fold farther than that of M-MIMO receivers in the SQL and PSL regimes, respectively.  The transmission distance of RAQ-MIMO receivers can be further improved in free-space propagation ($\nu = 2.0$), where the RAQ-MIMO receiver can realize $100$-fold and $21$-fold longer distances than M-MIMO receivers in the SQL and PSL regimes, respectively.

	\section{Discussions}
	\label{sec:Discussions}

	\textbf{Practical imperfection}: 
	The assumption of perfect CSI and clock synchronization, as well as zero carrier frequency offset (CFO) in this article represents an idealized simplifying condition, which is adopted to isolate and highlight the intrinsic quantum-optical EAR limitations of the proposed RAQ-MIMO architecture, while avoiding any confusion with the uncertainties arising from these practical imperfections. In real-world applications, imperfect CSI, clock drift, and CFO remain unavoidable and inevitably degrades the performance of RAQ-MIMO receivers, such as the bit-error-rate and throughput. But fortunately, the constructed RAQ-MIMO model in this article constitutes a quantum-to-classical transformation framework, encapsulating the quantum-optical parameters \eqref{eq:Gain} and \eqref{eq:Phase}, while offering a classical interface \eqref{eq:SignalModel_MatrixForm_SC} that is convenient for harnessing classical signal processing techniques for RAQ-MIMO systems. This formulation allows the application of a family of classical pilot-based channel estimation, clock synchronization, and CFO correction methods to RAQ-MIMO schemes \cite{coleri2002channel,Vasudeva2013Coherentn,vasudevan2015coherent,vasudevan2019coherent,Vasudevan23}, and also facilitates the theoretical analysis on the impact of practical imperfections. These aspects constitute an important future research direction for facilitating RAQ-MIMO's practical applications.

	\textbf{Heterogeneous robustness}: 
	In practice, user distributions become more complicated, especially in large-scale heterogeneous deployment scenarios, which may result in significant impacts on the spatial correlation structure and, consequently, the overall system performance. 
	However, this article mainly considers narrowband RAQ-MIMO schemes, where their atomic quantum response neither distorts nor couples with the multipath channel. Hence the RAQ-MIMO front-end behaves as an approximately linear and channel-independent component, as modeled in \eqref{eq:SignalModel_MatrixForm} and \eqref{eq:SignalModel_MatrixForm_SC}, respectively. Based on this, the channel variations due to user distribution and the RAQ-MIMO receiver can be studied separately. Therefore, the diverse user distributions do not affect the correctness of our analysis framework. This allows us to focus on new characteristics of RAQ-MIMO receivers, while isolating the influence of external wireless channels (diverse user distributions). 
	Additionally, the Weichselberger channel model adopted provides a general framework that can be systematically extended to incorporate more realistic heterogeneous user distributions \cite{wang2022uplink}. Such extensions will be explored in future work, particularly to evaluate the robustness of the RAQ-MIMO paradigm in large-scale, heterogeneous deployment scenarios.

	\textbf{Large-scale scalability}: 
	Supporting RAQ-MIMO to have a large number of sensors may be realized through several feasible approaches. Specifically, a single stabilized laser may be distributed to multiple vapor cells through beam-splitting and beam-shaping networks, such as diffractive optical elements, micro-lens arrays, or fiber splitters, ensuring phase-coherent and frequency-aligned probe beams. Additionally, spatial light modulators or multi-core fibers further facilitate flexible beam distribution across sensor arrays. Furthermore, advances in integrated photonics and microfabricated vapor-cell arrays allow low cost, size, weight, and power (C-SWaP) aware designs, as well as low-crosstalk implementations with shared optical paths. Since all RAQ-MIMO channels share a common optical carrier, imaging-based or array-integrated photodetection may be used for a unified optical readout, eliminating the inefficient one-to-one PD pairing.

	\textbf{Implementation concerns}: 
	This article is primarily focused on the theoretical modeling and performance analysis of RAQ-MIMO receivers without delving deeply into practical impairments and environmental variations that may arise in real-world deployments. Specifically, several related technical considerations are outlined below: \textbf{(i)} Atomic density and dynamic ambient temperature stabilization are essential for maintaining stable optical depth and ensuring balanced sensitivity across different RAQ-MIMO channels. \textbf{(ii)} Phase-coherence management is required to guarantee stable microwave reference sources and to enable active phase-locking schemes for synchronizing multi-channel responses. \textbf{(iii)} Scalable implementation of large numbers of simultaneous Rydberg sensors calls for compact optical architectures and integrated vapor-cell array designs. Our theoretical results can be viewed as a beneficial best-case reference for realistic applications when addressing the above-mentioned realistic imperfections.

	\section{Conclusions}
	\label{sec:Conclusions}
	
	In this article, we have conceived and studied the novel RAQ-MIMO scheme, where a flexible receiver array is employed for assisting classical multi-user uplink transmissions. We have also constructed the corresponding equivalent baseband signal model. The proposed scheme and signal model pave the way for future system design and signal processing. We have also studied the EARs of MRC/ZF RAQ-MIMO receivers in both UFC and CFC scenarios, where we have derived closed-form expressions of the asymptotic EAR. Based on these analytical results, we have unveiled the scaling law of the EAR and transmit power, respectively, for both the homogeneous and inhomogeneous RAQ-MIMO arrays. We have also performed detailed comparisons between RAQ-MIMO and M-MIMO receivers. Our simulation results have verified the accuracy of the proposed signal model and the superiority of the RAQ-MIMO receivers.


		\appendices
		\section{Proof of \eqref{eq:PassbandRFplusLO} and \eqref{eq:Amp_PassbandRFplusLO}}
		\label{Appendix:RFsuperpositionProof}
		
		The power of $Z_{m} (t)$ can be obtained from its equivalent baseband signal $z_{m}(t) = {y_{m}}(t) + \sum_{k=1}^{K} {x_{m,k}}(t) {e^{\jmath 2\pi {f_\delta }t}}$ through ${\cal P}_{z,m} = z_{m}^{2}\left( t \right) = \left[ {y_{m}}(t) + \sum_{k=1}^{K} {x_{m,k}}(t) {e^{\jmath 2\pi {f_\delta }t}} \right]$ ${\left[ {y_{m}}(t) + \sum_{k=1}^{K} {x_{m,k}}(t) {e^{\jmath 2\pi {f_\delta }t}} \right]^*}$. It is further derived as   
		\begin{align}
			\nonumber
			{\cal P}_{z,m} 
			&= {\cal P}_{\ell,m} + \sum_{k_1=1}^{K} \sum_{k_2=1}^{K}  \sqrt{ {\cal P}_{x,m,k_1} {\cal P}_{x,m,k_2} } + \sum_{k=1}^{K} \sqrt{ {\cal P}_{\ell,m} } \\
			\nonumber
			&\;\; \times \sqrt{ {\cal P}_{x,m,k} } 
			\big[ {{e^{\jmath \left( {2\pi {f_\delta }t + {\theta_{\delta,m,k} }} \right)}} + {e^{ - \jmath \left( {2\pi {f_\delta }t + {\theta_{\delta,m,k} }} \right)}}} \big] \\
			\nonumber
			&= {\cal P}_{\ell,m} + \sum_{k_1=1}^{K} \sum_{k_2=1}^{K}  \sqrt{ {\cal P}_{x,m,k_1} {\cal P}_{x,m,k_2} } + 2 \sum_{k=1}^{K} \sqrt{ {\cal P}_{\ell,m} } \\
			\nonumber
			&\;\; \times \sqrt{ {\cal P}_{x,m,k} } \cos \left( {2\pi {f_\delta }t + {\theta_{\delta,m,k} }} \right). 
		\end{align}
		Therefore, we can derive the expression of the amplitude $U_{z,m}$ in the form of \eqref{eq:Amp_PassbandRFplusLO}. 
		
		We then write $z_{m}(t) = {y_{m}}(t) + \sum_{k=1}^{K} {x_{m,k}}(t) {e^{\jmath 2\pi {f_\delta }t}} = \sqrt{\mathcal{P}_{\ell,m}}{e^{\jmath {\theta_{\ell,m}}}} + \sum_{k=1}^{K} \sqrt{\mathcal{P}_{x,m,k}}{e^{\jmath {\theta_{x,m,k}}}}{e^{\jmath 2\pi {f_\delta }t}}$ in the form of the real part plus the imaginary part as follows
		\begin{align}
			\nonumber
			&\hspace{-0.5em} {z_m}\left( t \right) = \sqrt{\mathcal{P}_{\ell,m}} \cos {{\theta_{\ell,m}}} + \sum_{k=1}^{K} \sqrt{\mathcal{P}_{x,m,k}} \cos \left( {2\pi {f_\delta }t + {\theta_{x,m,k}}} \right) \\
			\nonumber
			&+ \jmath \left[ \sqrt{\mathcal{P}_{\ell,m}} \sin {{\theta_{\ell,m}}} + \sum_{k=1}^{K} \sqrt{\mathcal{P}_{x,m,k}} \sin \left( {2\pi {f_\delta }t + {\theta_{x,m,k}}} \right) \right].
		\end{align}
		We first denote the phase of $z_m(t)$ by $\theta_{z,m}$. 
		We then obtain $\tan(\theta_{z,m}) \approx \tan(\theta_{\ell,m})$, as seen in \eqref{eq:phase}. Therein, $(1)$ holds by using $A\sin(x) + B\sin(y) = (A+B)\sin(\frac{x+y}{2})\cos(\frac{x-y}{2}) + (A-B)\cos(\frac{x+y}{2})\sin(\frac{x-y}{2})$ and $A\cos(x) + B\cos(y) = (A+B)\cos(\frac{x+y}{2})\cos(\frac{x-y}{2}) - (A-B)\sin(\frac{x+y}{2})\sin(\frac{x-y}{2})$, respectively. The approximation $(2)$ holds relying on $\sqrt{\mathcal{P}_{\ell,m}} \gg \sum_{k=1}^{K} \sqrt{\mathcal{P}_{x,m,k}}$.
		\begin{figure*}
			\begin{align}
				\label{eq:phase}
				&\tan {\theta_{z,m}} 
				\overset{(1)}{=} \frac{ \left[ \begin{array}{l} \sum\limits_{k=1}^{K} \left( {\sqrt{\mathcal{P}_{x,m,k}} + \frac{1}{K} \sqrt{\mathcal{P}_{\ell,m}}} \right)\sin \left( {\frac{{2\pi {f_\delta }t + {\theta_{x,m,k}} + {\theta_{\ell,m}}}}{2}} \right)\cos \left( {\frac{{2\pi {f_\delta }t + {\theta_{x,m,k}} - {\theta_{\ell,m}}}}{2}} \right) \\
						+ \sum\limits_{k=1}^{K} \left( {\sqrt{\mathcal{P}_{x,m,k}} - \frac{1}{K} \sqrt{\mathcal{P}_{\ell,m}}} \right)\cos \left( {\frac{{2\pi {f_\delta }t + {\theta_{x,m,k}} + {\theta_{\ell,m}}}}{2}} \right)\sin \left( {\frac{{2\pi {f_\delta }t + {\theta_{x,m,k}} - {\theta_{\ell,m}}}}{2}} \right) \end{array} \right] }{ \left[ \begin{array}{l} \sum\limits_{k=1}^{K} \left( {\sqrt{\mathcal{P}_{x,m,k}} + \frac{1}{K} \sqrt{\mathcal{P}_{\ell,m}}} \right)\cos \left( {\frac{{2\pi {f_\delta }t + {\theta_{x,m,k}} + {\theta_{\ell,m}}}}{2}} \right)\cos \left( {\frac{{2\pi {f_\delta }t + {\theta_{x,m,k}} - {\theta_{\ell,m}}}}{2}} \right) \\
						- \sum\limits_{k=1}^{K} \left( {\sqrt{\mathcal{P}_{x,m,k}} - \frac{1}{K} \sqrt{\mathcal{P}_{\ell,m}}} \right)\sin \left( {\frac{{2\pi {f_\delta }t + {\theta_{x,m,k}} + {\theta_{\ell,m}}}}{2}} \right)\sin \left( {\frac{{2\pi {f_\delta }t + {\theta_{x,m,k}} - {\theta_{\ell,m}}}}{2}} \right) \end{array} \right] } \\
				\nonumber
				&\overset{(2)}{\approx} \frac{ \sum\limits_{k=1}^{K} \left[ \sin \left( {\frac{{2\pi {f_\delta }t + {\theta_{x,m,k}} + {\theta_{\ell,m}}}}{2}} \right)\cos \left( {\frac{{2\pi {f_\delta }t + {\theta_{x,m,k}} - {\theta_{\ell,m}}}}{2}} \right) - \cos \left( {\frac{{2\pi {f_\delta }t + {\theta_{x,m,k}} + {\theta_{\ell,m}}}}{2}} \right)\sin \left( {\frac{{2\pi {f_\delta }t + {\theta_{x,m,k}} - {\theta_{\ell,m}}}}{2}} \right) \right] }{ \sum\limits_{k=1}^{K} \left[ \cos \left( {\frac{{2\pi {f_\delta }t + {\theta_{x,m,k}} + {\theta_{\ell,m}}}}{2}} \right)\cos \left( {\frac{{2\pi {f_\delta }t + {\theta_{x,m,k}} - {\theta_{\ell,m}}}}{2}} \right) + \sin \left( {\frac{{2\pi {f_\delta }t + {\theta_{x,m,k}} + {\theta_{\ell,m}}}}{2}} \right)\sin \left( {\frac{{2\pi {f_\delta }t + {\theta_{x,m,k}} - {\theta_{\ell,m}}}}{2}} \right) \right] } = \tan {\theta_{\ell,m}}. 
			\end{align}
			\hrulefill
			\vspace{-0.2cm}
		\end{figure*}
		Upon applying the Taylor series expansion to $U_{z,m}$ in \eqref{eq:Amp_PassbandRFplusLO}, we obtain 
		\begin{align}
			\nonumber
			{U_{z,m}} \approx U_{\ell,m} + \sum\limits_{k=1}^{K} U_{x,m,k} \cos( 2 \pi f_{\delta} t + \theta_{\delta,m,k} ).
		\end{align}
		The proofs of \eqref{eq:PassbandRFplusLO} and \eqref{eq:Amp_PassbandRFplusLO} are completed.

		\section{Proof of \eqref{eq:AchievableRate_MRC_Asymptotic_1}}
		\label{Appendix:SINRratio_CFC}
		To prove the CFC case of \eqref{eq:AchievableRate_MRC_Asymptotic_1}, we have to first observe how $\frac{1}{M} {\rm{Tr}} \left( \bm{\varLambda}^{2} \right)$ varies as $M \to \infty$. Since $\frac{1}{M} {\rm{Tr}} \left( \bm{\varLambda}^{2} \right) = \frac{1}{M} {\rm{Tr}} \left( \bm{R}^{*} \bm{R} \right)$ and $\bm{R} = [r_{|m_1 - m_2|}] \in \mathbb{C}^{M \times M}$, $m_1, m_2 = 0, 1, \cdots, M-1$ is a symmetric real-valued Toeplitz matrix, we hence obtain the diagonal elements of $\bm{R}^{*} \bm{R}$ as follows 
		\begin{align}
			\nonumber
			{e_1} = {e_M} = \sum\limits_{m = 0}^{M - 1} {r_m^2} \;\; {\rm{and}} \;\;
			{e_k} = r_0^2 + \sum\limits_{n = 1}^{k - 1} {r_n^2}  + \sum\limits_{m = 1}^{M - k} {r_m^2},
		\end{align}
		for $2 \le k \le M - 1$. 
		We observe that ${e_k}$ is bounded as follows 
		\begin{align}
			\nonumber
			r_0^2 + \sum\limits_{m = 1}^{M - 1} {r_m^2} 
			\le {e_k} \le \left\{ 
			\begin{aligned}
				&r_0^2 + 2\sum\nolimits_{m = 1}^{ \lceil \frac{M}{2} \rceil - 1 } {r_m^2}, & {\text{odd}} \; M, \\
				&r_0^2 + 2\sum\nolimits_{m = 1}^{\frac{M}{2} - 1} {r_m^2} + r_{\frac{M}{2}}^2, & {\text{even}} \; M.
			\end{aligned} \right. 
		\end{align}
		Therefore, we have $\mathsf{LB} \le \frac{ {\rm{Tr}} \left( \bm{\varLambda}^{2} \right) }{M} = \frac{1}{M} \sum_{k = 1}^M {{e_k}} \le \mathsf{UB}$, where 
		\begin{align}
			\nonumber
			\mathsf{LB} &= \sum\limits_{m = 0}^{M - 1} J_0^2 (\varpi m) 
			\overset{(3)}{\approx} r_0^2 + \sum_{m = 1}^{M - 1} \left( \frac{1}{\pi \varpi m} \right) \cos^2 \left( \varpi m - \frac{\pi}{4} \right) \\
			\nonumber
			&= r_0^2 + \left( \frac{1}{\pi \varpi} \right) \sum_{m = 1}^{M - 1} \frac{1}{m}  + \left( \frac{1}{\pi \varpi} \right) \sum_{m = 1}^{M - 1} \frac{ \sin \left( 2 \varpi m \right) }{m} \\
			\label{eq:LB}
			&\overset{(4)}{\approx} r_0^2 + \frac{1}{\pi \varpi} \left[ \epsilon + \ln \left( M - 1 \right) \right] + \left( \frac{1}{2\varpi} - \frac{1}{\pi} \right). 
		\end{align}
		The equality $(3)$ of \eqref{eq:LB} is obtained by exploiting $J_0 (\varpi m) \approx \sqrt{ \frac{2}{\pi \varpi m} } \cos ( \varpi m - \frac{\pi}{4} )$ \cite[Ch. 5]{boyce2021elementary}. The equality $(4)$ of \eqref{eq:LB} is derived by harnessing the harmonic series $\lim_{M \rightarrow \infty}$ $\sum_{m = 1}^{M - 1} \frac{1}{m} = \ln(M-1) + \epsilon$ with $\epsilon \approx 0.577$ and the trigonometric series $\lim_{M \rightarrow \infty} \sum_{m = 1}^{M - 1} \frac{ \sin \left( m x \right) }{m} = \frac{\pi - x}{2}$ \cite[Ch. 1.44]{gradshteyn2014table}. Upon obeying a similar derivation process as \eqref{eq:LB}, we obtain $\mathsf{UB}$ for odd $M$ and even $M$ as follows 
		\begin{align}
			\nonumber
			&\mathsf{UB}_{\rm{odd}} 
			\approx r_0^2 + \frac{2}{\pi \varpi}\left[ \epsilon + \ln \left( \left\lceil \frac{M}{2} \right\rceil - 1 \right) \right] + \left( \frac{1}{\varpi} - \frac{2}{\pi} \right), \\
			\nonumber
			&\mathsf{UB}_{\rm{even}}  
			\approx r_0^2 + \frac{2}{\pi \varpi} \left[ \epsilon + \ln \left( \frac{M}{2} - 1 \right) \right] + \left( \frac{1}{\varpi} - \frac{2}{\pi} \right) \\
			\label{eq:UB}
			&\qquad \qquad \quad + \frac{4}{\pi \varpi} \left[ \frac{{1 + \sin \left( {\varpi M} \right)}}{M} \right]. 
		\end{align}
		Upon defining $\zeta = \mathsf{UB}$, we arrive at the CFC case of \eqref{eq:AchievableRate_MRC_Asymptotic_1}.

\bibliographystyle{IEEEtran}
\bibliography{IEEEabrv,references} 

\end{document}